\documentclass[10pt,journal,draftclsnofoot,onecolumn]{IEEEtran}
\usepackage{mathpazo}
\usepackage{times}

\usepackage{graphicx}
\usepackage{amssymb}
\usepackage{amsmath}
\usepackage[mathscr]{euscript}
\usepackage{enumerate}
\usepackage{mathcomp}

\usepackage{theorem}  
\usepackage{cite}     
 
\usepackage{upref}

\usepackage{comment}  
\usepackage{psfig}

\usepackage{color}

\usepackage{supertabular}
\usepackage{multicol}

\makeatletter
\let\mcnewpage=\newpage
\newcommand{\TrickSupertabularIntoMulticols}{%
  \renewcommand\newpage{%
    \if@firstcolumn
      \hrule width\linewidth height0pt
      \columnbreak
    \else
      \mcnewpage
    \fi
  }%
}
\makeatother

\title{Thermal-Optimal Encoding for Low-Power Buses}
\title{Cooling Codes:
         Optimal Thermally-Aware Coding for Low-Power Buses}
\title{Cooling Codes:
         Thermally-Aware Coding for High-Performance Interconnects}
\title{Cooling Codes:
         Temperature Management Coding for High-Performance Interconnects}
\title{Cooling Codes:
        Thermal Management Coding for On-Chip Interconnects}
\title{Cooling Codes:
        Thermal-Management Coding for~High-Performance Interconnects}
\title{Cooling Codes:
        Optimal Coding to Reduce Peak~Temperature and Power Dissipation
        in High-Performance Interconnects}

\title{Cooling Codes:
        Thermal-Management Coding for High-Performance Interconnects}

\date{\today}
\author{\textbf{Yeow Meng Chee$^\text{x}$}, \textbf{Tuvi Etzion$^*$}, \textbf{Han Mao Kiah$^\text{x}$}, \textbf{Alexander Vardy$^{+\text{x}}$}\\
{\small $^\text{x}$School of Physical and Mathematical Sciences, Nanyang Technological University, Singapore}\\
{\small $^*$Computer Science Department, Technion, Israel Institute of Technology, Haifa 32000, Israel.\newline
The work was done while the author was visiting SPMS, Nanyang Technological University, Singapore,\linebreak
and the Department of Electrical and Computer Engineering, University of California}\\
{\small $^+$Department of Electrical and Computer Engineering, University of California, San Diego, CA 92093, USA}\\
{\small {\it YMChee@ntu.edu.sg}, {\it etzion@cs.technion.ac.il}, {\it HMKiah@ntu.edu.sg}, {\it avardy@ucsd.edu}\vspace{-0.13ex}}
\thanks{Y. M. Chee and H. M. Kiah were supported in part by the Singapore Ministry of
Education under Research Grant MOE2015-T2-2-086.
H.~M.~Kiah was also supported in part by the Singapore Ministry of
Education under Research Grant and MOE2016-T1-001-156.
T. Etzion and A. Vardy were supported in part by the United States --- Israel
Binational Science~Foundation (BSF), Jerusalem, Israel, under Grant 2012016.}
}

\theoremstyle{plain}
\theorembodyfont{\normalfont\slshape}

\newtheorem{thm}{Theorem\hspace{-1pt}}
\newenvironment{theorem}
{\begin{thm}\hspace*{-1ex}{\bf.}}{\end{thm}}

\newtheorem{lem}[thm]{Lemma\hspace{-.75pt}}
\newenvironment{lemma}{\begin{lem}\hspace*{-1ex}{\bf.}}{\end{lem}}

\newtheorem{prop}[thm]{Proposition$\!$}
\newenvironment{proposition}{\begin{prop}\hspace*{-1ex}{\bf.}}{\end{prop}}

\newtheorem{cor}[thm]{Corollary$\!$}
\newenvironment{corollary}{\begin{cor}\hspace*{-1ex}{\bf.}}{\end{cor}}

\newtheorem{conj}[thm]{Conjecture$\!$}
\newenvironment{conjecture}{\begin{conj}\hspace*{-1ex}{\bf.}}{\end{conj}}

\newtheorem{defn}{Definition$\!$}
\newenvironment{definition}{\begin{defn}\hspace*{-1ex}{\bf.}}{\end{defn}}

\newcommand{\cA}{{\cal A}}
\newcommand{\cB}{{\cal B}}
\newcommand{\cC}{{\cal C}}
\newcommand{\cD}{{\cal D}}
\newcommand{\cE}{{\cal E}}

\newcommand{\cL}{{\cal L}}

\newcommand{\cS}{{\cal S}}

\DeclareMathAlphabet{\mathbfsl}{OT1}{ppl}{b}{it} 

\newcommand{\ccc}{\mathbfsl{c}}

\newcommand{\uuu}{\mathbfsl{u}}
\newcommand{\vvv}{\mathbfsl{v}}

\newcommand{\xxx}{\mathbfsl{x}}
\newcommand{\yyy}{\mathbfsl{y}}
\newcommand{\zzz}{\mathbfsl{z}}

\newcommand{\ceil}[1]{\left\lceil #1 \right\rceil}
\newcommand{\floor}[1]{\left\lfloor #1 \right\rfloor}

\newcommand{\be}[1]{\begin{equation}\label{#1}}
\newcommand{\ee}{\end{equation}}
\newcommand{\eq}[1]{(\ref{#1})}

\renewcommand{\le}{\leqslant}
\renewcommand{\leq}{\leqslant}
\renewcommand{\ge}{\geqslant}
\renewcommand{\geq}{\geqslant}

\newcommand{\script}[1]{{\mathscr #1}}
\renewcommand{\frak}[1]{{\mathfrak #1}}

\renewcommand{\Bbb}{\mathbb}
\newcommand{\C}{{\Bbb C}}

\newcommand{\FF}{{\Bbb F}}
\newcommand{\Fq}{{{\Bbb F}}_{\!q}}
\newcommand{\Ftwo}{\smash{\Bbb{F}_{\kern-1pt2}}}
\newcommand{\Fn}{\smash{\Bbb{F}_{\kern-1pt2}^{\hspace{0.5pt}n}}}
\newcommand{\Fk}{\smash{\Bbb{F}_{\kern-1pt2}^{\hspace{0.5pt}k}}}

\newcommand{\Fkap}{\smash{\Bbb{F}_{\kern-1pt2}^{\hspace{0.5pt}\kappa}}}
\newcommand{\Fqkap}{\smash{\Bbb{F}_{\kern-1ptq}^{\hspace{0.5pt}\kappa}}}
\newcommand{\Fqn}{\smash{\Bbb{F}_{\kern-1ptq}^{\hspace{0.5pt}n}}}
\newcommand{\Fqm}{\smash{\Bbb{F}_{\kern-1ptq}^{\hspace{0.5pt}m}}}

\newcommand{\bfsl}{\bfseries\slshape}
\newcommand{\bfit}{\bfseries\itshape}

\newcommand{\dfn}{\bfseries\itshape}

\newcommand{\Dref}[1]{Def\-i\-ni\-tion\,\ref{#1}}
\newcommand{\Tref}[1]{The\-o\-rem\,\ref{#1}}

\newcommand{\Lref}[1]{Lem\-ma\,\ref{#1}}
\newcommand{\Cref}[1]{Co\-ro\-lla\-ry\,\ref{#1}}

\newcommand{\deff}{\mbox{$\stackrel{\rm def}{=}$}}

\DeclareMathOperator{\wt}{wt}
\DeclareMathOperator{\diam}{diam}
\DeclareMathOperator{\supp}{supp}

\newcommand{\zero}{{\mathbf 0}}

\newcommand{\sC}{{\script C}}

\newcommand{\sS}{{\script S}}

\newcommand{\fT}{{\frak T}}

\newcommand{\A}{{\bf A}}
\newcommand{\B}{{\bf B}}

\newcommand{\al}{\alpha}

\newcommand{\E}{{\script E}}

\newcommand{\HH}{{\mathcal H}}

\makeatletter
\let\over\@@over
\let\atop\@@atop
\let\atopwithdelims\@@atopwithdelims
\makeatother

\makeatletter
\renewcommand{\@endtheorem}{\endtrivlist}
\makeatother

\makeatletter
\renewcommand{\thefigure}{{\bf \@arabic\c@figure}}
\renewcommand{\fnum@figure}{{\bf Figure}\,\thefigure}
\makeatother

\renewcommand{\QEDclosed}{\mbox{\rule[-1pt]{1.3ex}{1.3ex}}}

\begin{document}

\maketitle

\vspace*{-3.50ex}

\hspace*{-1pt}\begin{abstract}
High temperatures have dramatic negative effects on interconnect
performance and, hence, numerous techniques have been proposed to
reduce the power consumption of on-chip buses. However, existing
methods fall short of fully addressing the thermal challenges posed
by high-performance interconnects. In this paper, we introduce new
efficient coding schemes that make it possible to directly control
the \emph{peak temperature} of a bus by effectively cooling its
hottest wires. This is achieved by avoiding state transitions on the
hottest wires for as long as necessary until their temperature drops
off. We also reduce the \emph{average power consumption} by making
sure that the total number of state transitions on all the wires
is below a prescribed threshold. We show how each of these two features
can be coded for separately or, alternatively, how both can be achieved
at the same time. In addition, \emph{error-correction} for the
transmitted information can be provided while controlling the peak
temperature and/or the average power consumption.

\looseness=-1
In general, our cooling codes use $n > k$ wires to encode a given
$k$-bit bus. One of our goals herein is to determine the minimum
possible number of wires $n$ needed to encode $k$ bits while satisfying
any combination of the three desired properties. We provide full theoretical
analysis in each case. In particular, we show that $n = k+t+1$ suffices
to cool the $t$ hottest wires, and this is the best possible. Moreover,
although the proposed coding schemes make use of sophisticated tools
from combinatorics, discrete geometry, linear algebra, and coding
theory, the resulting encoders and decoders are fully practical.
They do not require significant computational overhead and can be
implemented without sacrificing a large circuit area.
\end{abstract}

\newpage
\section{Introduction}
\label{sec:Introduction}
\vspace{0.50ex}

\looseness=-1
\PARstart{P}{ower} and heat dissipation 
limits have emerged as a first-order design constraint for chips,
whether targeted~for~bat\-tery-powered devices or for high-end systems.
With the migrat\-ion to process geometries
of 65\,nm and below, power dissi\-pat\-ion has become as important
an issue as timing and signal in\-teg\-rity.
Aggressive technology scaling results in 
smaller feature size, greater packing density, increasing microarchitectural
complexity, and higher clock frequencies.
This is pushing~chip level power consumption to the edge.
It is not uncommon for on-chip hot spots to have
temperatures exceeding 100\textdegree C, while inter-chip temperature
differentials often exceed~20\textdegree C.

Power-aware design alone is not sufficient to
address this thermal challenge, since it does not directly target the
spatial and temporal behavior of the operating environment. For this
reason, thermally-aware approaches 
have emerged as one~of~the most
important domains of research in chip design today.

High temperatures have dramatic negative
effects~on circuit
behavior, with interconnects being among the most 
impacted circuit components. This is due, in part, to the
ever decreasing interconnect pitch and
the introduction of low-k dielectric insulation which has
low thermal conductivity. For example, as shown
in \cite{Ajamietal:2001}, the Elmore delay~\cite{Elm48} of an interconnect
increases 5\% to 6\% for every 10\textdegree C increase in temperature,
whereas the leakage current grows exponentially with temperature.
Therefore, minimizing the temperature of interconnects is 
of para\-mount importance for thermally-aware design.

$\,$\vspace*{-2.25mm}

\centerline{%
\hspace*{0.00mm}\includegraphics[width=3.05in]{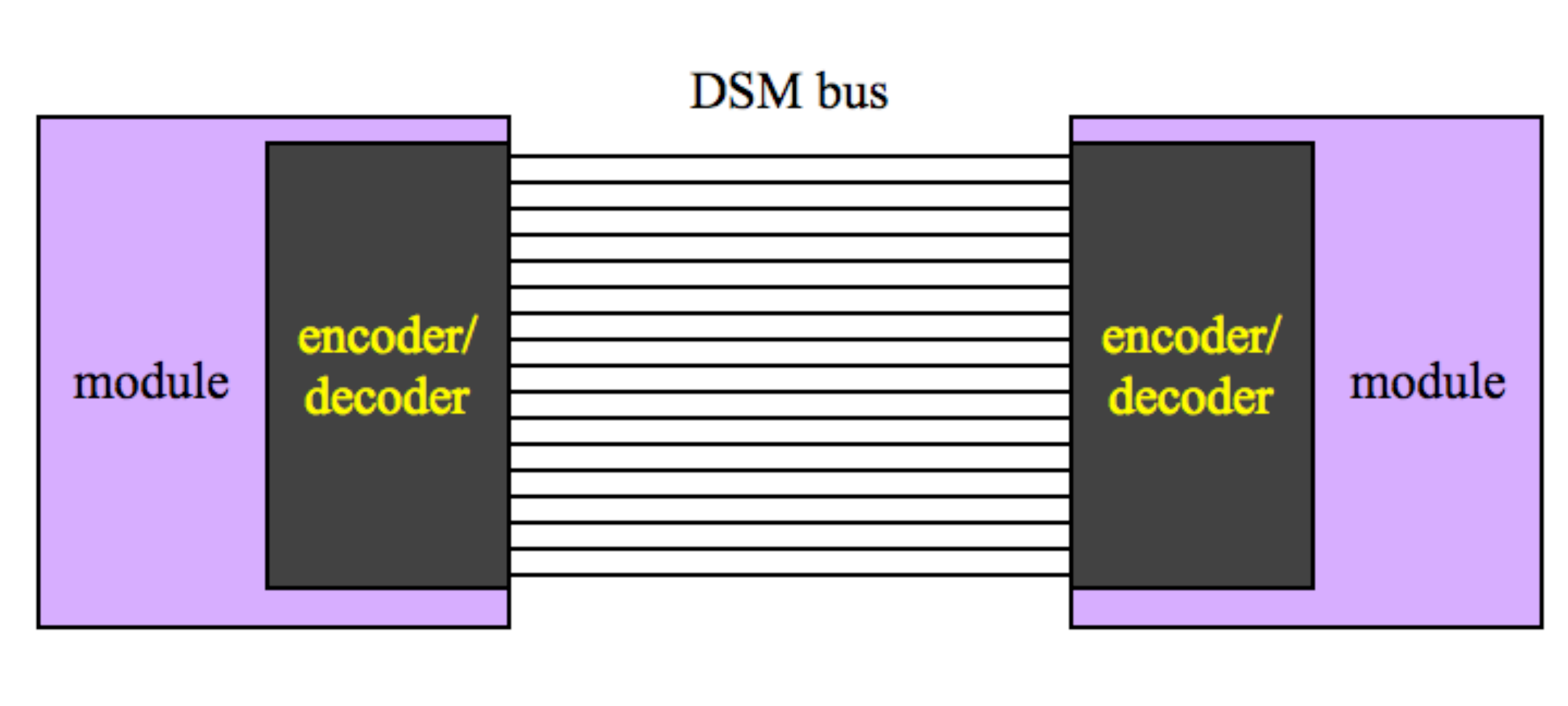}}
\vspace*{0.00mm}
\label{fig1}

\noindent
\centerline{\small {\bf Figure\,1.}
\sl Block diagram of the proposed bus architecture}
\vspace*{-0.75mm}

\subsection{Related Work}
\vspace{-0.00ex}
\label{RelatedWork}

\noindent\looseness=-1
Numerous encoding techniques have been 
proposed in the literature
\cite{SoCh01,SoCh03,SWC00,SB:1995,Benini+:1997,CCL2006,KF:2001,PO:2004,Wang+:2006,Wang+:2007,SMM:2007,Calimera+:2008}
in order to reduce the overall power dissipation
consumption of both on-chip and off-chip buses.
It is known to be of importance from practical point of view for over thirty years and
optimization of the related integrated circuits were considered by the electric companies~\cite{Fle87}.
It is well established
\cite{Wang+:2006,Wang+:2007,MPS:1998,Chiangetal:2001,SotiriadisChandrakasan:2002,SundaresanMahapatra:2005}
that bus power is directly~proportional
to the product of line capacitance and the average number of
signal state transitions on the bus wires. Thus the general idea is
to encode~the data transmitted over the bus so as~to~reduce the average
number of transitions.
For example, the ``bus-invert'' code of \cite{SB:1995}
potentially complements the data on all the wires,
according to the Hamming distance between consecuti\-ve transmissions,
thereby ensuring that the total number of state 
transitions on $n$ bus wires never exceeds $n/2$. Unfortunately,
encoding techniques designed to minimize
power consumption, do {not} directly address
peak temperature minimization. In order to reduce the 
temp\-erature of a wire, it is not sufficient to minimize its
average switching activity. Rather, it is necessary to control
the \emph{tempo\-ral distribution of the state 
transitions} on the wire. To~reduce the peak temperature of an
interconnect, it is necessary to exercise such control for
\emph{all} of its constituent wires.

In \cite{Wang+:2006,Wang+:2007}, the authors propose a
\emph{thermal spreading} approach. They present an efficient
encoding scheme that evenly spreads the switching activity among all
the bus wires, using a simple architecture consisting
of a shift-register and a crossbar logic. This is designed
to avoid the situation where a few wires get 
hot while the majority are at a lower temperature.
This spreading approach is further extended in \cite{Calimera+:2008,SMM:2007}
using on-line monitoring of the switching activity on all the wires.
Thermal spreading can be regarded as an attempt to control
peak temperature indirectly, by equalizing the distribution
of signal transitions over all the 
wires.

Finally, analysis from information theory point of view, which is highly related to our work,
including solutions with data compression are given in several papers, e.g.~\cite{KLS09,RSH99,STC03}.

\subsection{Our Contributions}
\vspace{-0.00ex}
\label{Contributions}

\noindent\looseness=-1
As technology continues to scale, 
existing methods may fall short
of fully addressing the thermal challenges posed by 
high-performance interconnects in deep submicron (DSM) circuits.
In this paper, we introduce new~efficient
coding schemes 
that simultaneously 
control both the \emph{peak temperature} and the
\emph{average power~consump\-tion} of interconnects.

The proposed coding schemes are distinguished from existing
state-of-the-art
by having some or all of the following features:

\vspace{-1.00ex}

\sl
\begin{list}{}
{
\addtolength{\leftmargin}{-2.00ex}
\setlength{\rightmargin}{\leftmargin}
}
\item\noindent
\begin{itemize}
\item[\bf A.\hspace*{-1pt}]
$\!$We directly control the peak temperature of  \vspace{-0.10ex}
a bus~by~effectively \vspace{-0.10ex}
cooling its hottest wires. This is achieved by avoiding \vspace{-0.10ex}
state transitions on the hottest wires for as long as necessary
until their temperature decreases.~~
\vspace{0.850ex}
\item[\bf B.\hspace*{-1pt}]
$\!$We reduce the overall power \vspace{-0.10ex} dissipation
by guarantee\-ing that the total number \vspace{-0.10ex} of
transitions on the bus wires is below a specified threshold
in every transmission.
\vspace{0.90ex}
\item[\bf C.\hspace*{-1pt}]
$\!$We combine
properties\kern1.0pt\ {\bf A} \kern0.5ptand/or\kern1.5pt\ {\bf B}
\vspace{-0.10ex}
with~coding 
for~improved reliability \vspace{-0.10ex}
(e.g.,~for~low-swing signaling),  
using existing error-correcting codes.\hspace*{-1.50pt}
\end{itemize}
\end{list}
\vspace{1.00ex}
\em

To achieve these desirable features, we propose to insert~at
the interface of the bus specialized circuits implementing
the encoding and decoding functions, denoted herein by $\cE$ and $\cD$,
respectively. This is illustrated in Figure\,1.
The various coding schemes introduced in this paper
employ tools from various fields such as combinatorics, graph theory,
block designs, discrete
geometry, linear algebra, and the theory of error-correcting codes.
Nonetheless, in each~case the
\emph{resulting encoders/decoders $\cE$ and $\cD$ are efficient:}
they do not require significant computational overhead and can be
implemented without sacrificing a large circuit area.
This is especially true for Property\,{\bf A}, 
where the complexity of encoding and decoding scales linearly
with the number of wires.~~ 

\looseness=-1
We 
consider both adaptive and nonadaptive (memoryless) coding schemes.
The advantage of nonadaptive schemes is that they are easier to implement
and do not require 
memory. The disadvantage is that it is not possible to 
implement Property\,{\bf A} with nonadaptive encoding. For this reason,
most of the coding schemes developed in this paper will be adaptive,
based on the idea of \emph{differential encoding}.
Notably, however, all of our schemes require the encoder and decoder
circuits 
to keep track of \emph{only one} (the most recent) previous
transmission.

\looseness=-1
Unlike the thermal spreading methods of \cite{Wang+:2006,Wang+:2007,SMM:2007}
that lead~to irredundant coding schemes, the solutions we propose do
introduce redundancy: we require $n > k$ wires to encode a given
$k$-bit bus. A key consideration in this situation 
is the \emph{area overhead due to the additional $n-k$ wires}.
Therefore, it is important to determine the theoretically minimum
possible number of wires~$n$ needed to encode $k$ bits while
satisfying the desired properties. We provide full theoretical
analysis in each~case.
We moreover show that 
the number of additional wires requir\-ed to satisfy
Property\,{\bf A} becomes negligible 
when $k$ is large.~~

\subsection{Thermal Model}
\vspace{-0.00ex}
\label{Model}

\noindent\looseness=-1
Chiang, Banerjee, and Saraswat~\cite{Chiangetal:2001} came up with
an analytic model that characterizes thermal effects due to Joule heating
in high-performance Cu/low-k interconnects, under both steady-state
and transient stress conditions. Shortly thereafter,~Soti\-riadis
and Chandrakasan~\cite{SotiriadisChandrakasan:2002} gave a power
dissipation model for DSM buses.
These two models, accounting for thermal
and power effects separately, were later unified and refined
by Sundaresan and Mahapatra~\cite{SundaresanMahapatra:2005}.
Finally, building upon this work, 
Wang, Xie, \kern-1ptVijaykrishnan, and Irwin \cite{Wang+:2006} 
proposed~a~more~accurate thermal-and-power model for DSM buses.
In all these papers, an $n$-bit bus
(illustrated in Figure\,2)~is~modeled in terms


\vspace{0.8cm}

\centerline{%
\hspace*{0.00mm}\includegraphics[width=2.65in]{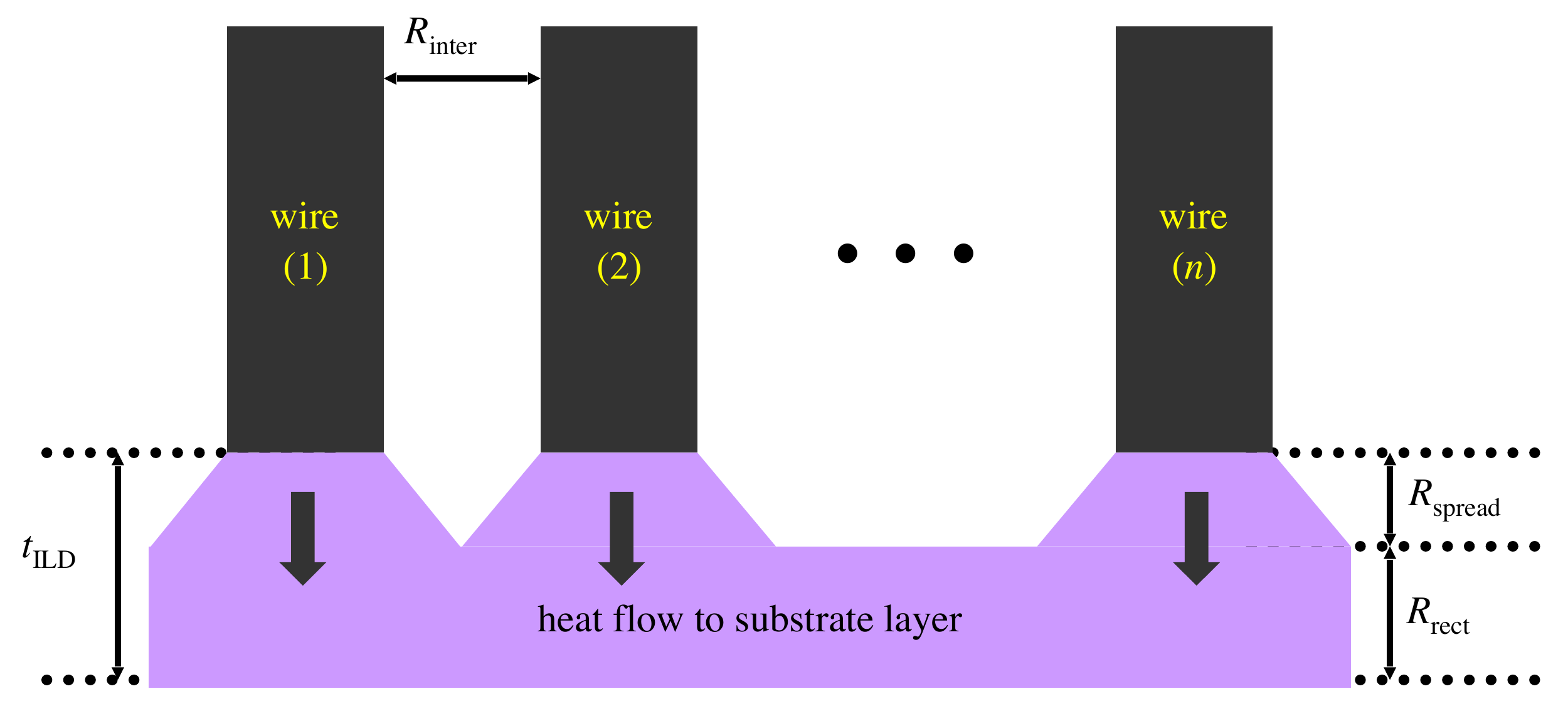}}
\vspace*{1.50mm}
\label{fig2}

\noindent
\centerline{\small {\bf Figure\,2.} \sl
Geometry used for calculating $R_{\rm spread}$, $R_{\rm rect}$, and $R_{\rm inter}$}

\vspace{0.75mm}

\centerline{%
\hspace*{-2.00mm}\includegraphics[width=3.15in]{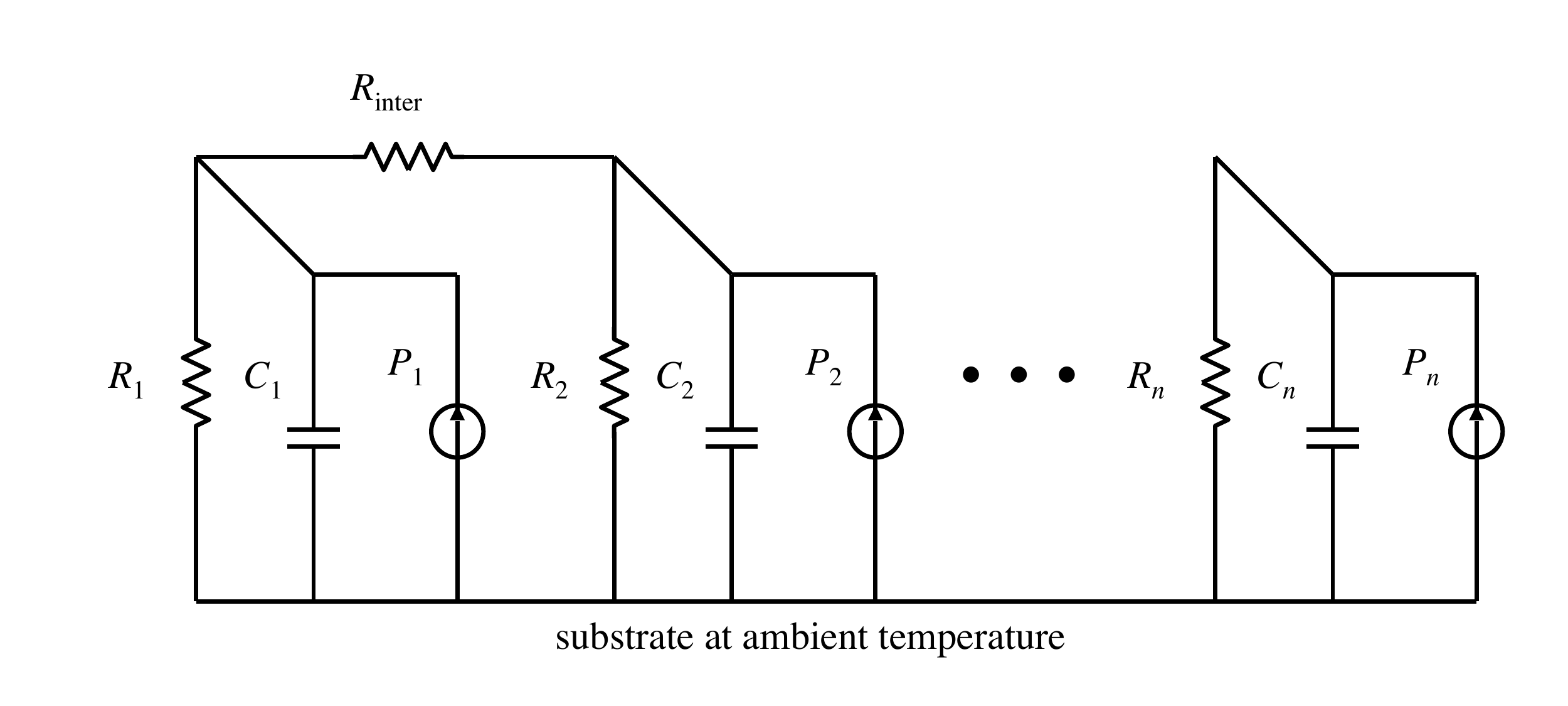}}
\vspace*{0.00mm}
\label{fig3}

\noindent
\centerline{\small {\bf Figure\,3.}
\sl Equivalent thermal RC-network for a $k$-bit bus}
\vspace*{-0.00mm}

\noindent
of the equivalent thermal-RC network 
in Figure\,3. Sunda\-resan and Mahapatra~\cite{SundaresanMahapatra:2005}
show that this thermal-RC network is governed by the following
differential equations:\vspace{0.50ex}
\begin{align}
\label{edge1}
P_1 & =\,
C_1 \frac{\partial\theta_1}{\partial t} + \frac{\theta_1-\theta_0}{R_1} +
\frac{\theta_1-\theta_2}{R_{\rm inter}}~,
\\[0.75ex]
\label{edge2}
P_n & =\,
C_k \frac{\partial \theta_k}{\partial t} + \frac{\theta_k-\theta_0}{R_k} +
\frac{\theta_k-\theta_{k-1}}{R_{\rm inter}}~,
\text{~~and}
\\[0.75ex]
\label{middle}
P_i & =\,
C_i \frac{\partial \theta_i}{\partial t} + \frac{\theta_i-\theta_0}{R_i} +
\frac{2\theta_i-\theta_{i-1}-\theta_{i+1}}{R_{\rm inter}}~,
\end{align}
for $i = 2,3,\ldots,n{-}1$, \looseness=-1
where $P_i$ is the instantaneous power~dissipated by wire $i$,
$C_i$ is the thermal capacitance per unit length of wire $i$,
$R_i=R_{\rm spread}{+}R_{\rm rect}$ is the thermal resistance per~unit
length of wire $i$ along the heat transfer path downwards,
$R_{\rm inter}$ is the lateral thermal resistance used to account
for the parallel thermal coupling effect between the wires,
$\theta_i$ is the temperature of wire $i$, and
$\theta_0$ is the substrate ambient temperature.

\looseness=-1
In any bus model for which \eq{edge1}\,--\,\eq{middle} hold,
the temperature~of a wire will increase whenever the wire
undergoes~a~state~transition; conversely, in the absence
of state transitions, the temperature will gradually decrease.
We let $\sigma_i$ denote the switching activity of wire $i$,
which is the number of times the wire changes state. Then
the power dissipated by a bus is determin\-ed by
its \emph{total switching activity} 
$\sigma_1 + \sigma_2 + \cdots + \sigma_n$.~~

In order to directly control the peak temperature of a bus~by
avoiding transitions on its hottest wires, \kern-1ptwe need to
know~\emph{which wires are the hottest} at every transmission.
There are two general ways to obtain this information.
We can use an analytical model~\cite{SundaresanMahapatra:2005}, such as \eq{edge1}\,--\,\eq{middle},
to estimate the current temperatures of the wires. For each
wire, such an estimate can be implemented with a
counter that is incremented on transition and decremented
on non-transition, where the precise magnitude of the
increments/decrements is determined by the model.
Alternatively, we can have actual temperature sensors
for each wire. For DSM buses, accurate temperature sensing
can be implemented using, for example,
ring oscillators~\cite{DB:008}.
As shown in~\cite{DB:008},
sensors based on ring oscillators
provide a resolution of 1\textdegree C while consuming
an active power of only $65$--$112\mu W$.~~

\subsection{Organization}

\noindent
The rest of this paper is organized as follows.
We begin in the next section
with a precise formulation of the coding problems that result
from the thermal-management features we propose to implement --- namely,
Properties {\bf A}, {\bf B}, and {\bf C}.
In Section\,\ref{anticodes},
we present a nonadaptive coding scheme that combines
Property\,{\bf B} (reducing the average power dissipation)
with~the~thermal spreading approach of~\cite{Wang+:2006}.
Our constructions in Section\,\ref{anticodes}
are based on the notions of \emph{anticodes} and \emph{quorum systems},
and use key results from the theory of combinatorial designs.
Section\,\ref{spreads} is devoted to Property\,{\bf A}:
we show how state transitions on the $t$ hottest wires
can be avoided by using only $t+1$ additional
bus lines. This optimal construction is based on
combining \emph{differential coding} with the notion
of \emph{spreads} and \emph{partial spreads} in projective geometry.
The optimal construction can be applied when $t+1 \leq (n+1)/2$.
When $t+1 > (n+1)/2$ we use another technique from the theory
of error-correcting codes to construct efficient codes.
The designed codes can be viewed as sunflowers, while
the partial spreads, are also sunflowers, and hence can be also viewed
as a special case of these codes.
The technique used is generalized with the notion of
generalized Hamming weights.
In Section\,\ref{sec:LPCC}, we show how
Properties~{\bf A} and {\bf B}
can be all achieved \emph{at the same time}. That is,
we design coding schemes that simultaneously control peak temperature
and average power consumption in every transmission.
For this purpose, we present three types of constructions.
The first construction is based upon
the \emph{Baranyai theorem}~on complete hypergraph decomposition into pairwise
disjoint perfect matchings. The second construction is based on concatenation
of low weights codes based on appropriate non-binary dual codes or non-binary
partial spreads. The third construction is the previous sunflower construction,
which also satisfy Property\,{\bf B}.
Section\,\ref{sec:ECC} is devoted to codes which satisfy Property\,{\bf C}
simultaneously with either Property\,{\bf A} or Property\,{\bf B} or both,
i.e. we add also correction for possible transmission errors on the bus wires.
The constructions in this section will also be of two types. The first type of
constructions is based on resolutions in block design. The second type of constructions
will be to employ the previously given constructions, where our set of transitions
is restricted to the set of codewords in a given error-correcting code.
In all these sections our bounds and constructions are applied for infinite sets
of parameters, but there is no asymptotic analysis in any of these cases.
The asymptotic analysis is postponed to
Section\,\ref{sec:asymptotic}, where the asymptotic behavior
of our codes is analyzed. In particular we analyze area overhead of our constructions and prove
that when $k$ is large enough, the additional number of wires required to satisfy the desired
properties is negligible.
Finally, Section\,\ref{sec:conclude} summarizes our comprehensive work and presents a
brief discussion of possible directions for future research.
We would like to remark that this work is mainly of theoretical nature. A follow up work
will present more constructions, especially for practical parameters. Examples and numerical
experiments will be given and also comparison between the various constructions with emphasis
on practical parameters. It will be illustrated and discussed how practical our methods
and constructions are.

\vspace{2.00ex}
\section{Problem Formulation and Preliminaries}
\label{sec:Preliminaries}
\vspace{0.50ex}

\noindent\looseness=-1
Let us now elaborate upon Properties {\bf A}, {\bf B}, {\bf C}
introduced in the previous section. For each of these properties,
we will charac\-terize the performance of the corresponding coding
scheme~by a \emph{single integer parameter}.
All of our 
coding~schemes will use $n > k$ wires to encode a 
$k$-bit bus. \kern-1ptWe assume that communication~across the bus is
synchronous,
occurring in clocked cycles called~\emph{transmissions}.
This leads to the following definition.~

\begin{definition}
Consider a coding scheme for communication over a bus
consisting of $n$ wires. Let $t$, $w$, $e$ be positive
integers less than $n$. We say that the coding scheme
has\vspace{0.25ex}
\begin{description}
\item[]
\hspace*{-11mm}\textup{\bfsl Property\,${\bf A}(t)$}{\bf:}
if every transmission does not cause state transitions on the $t$ hottest wires;\vspace{0.5ex}

\item[]
\hspace*{-11mm}\textup{\bfsl Property ${\bf B}(w)$}{\bf:}
if the total number of state transitions on all
the wires is at most $w$, in every transmission;\vspace{0.5ex}

\item[]
\hspace*{-11mm}\textup{\bfsl Property ${\bf C}(e)$}{\bf:}
if up to $e$ transmission errors
\textup{(}$0$ received as~$1$,
or\/ $1$ received as $0$\textup{)}
on the $n$ wires can be corrected.\vspace{0.25ex}
\end{description}
We presume that, at the time of transmission, it is known
which $t$ wires are the hottest; Property\,${\bf A}(t)$ is
required to hold assuming that\/ {\bfit any} $t$ wires can be
designated as the hottest. 
\end{definition}

The values of $t$, $w$, $e$ are \emph{design parameters},
to be determin\-ed by the specific thermal requirements of
specific interconnects. The proposed coding schemes will work for
various values of $t$,~$w$,~and~$e$. Nevertheless, it might be helpful to think
of $t$ as a small constant, since significant reductions in
the peak temperature can be achieved by cooling only a few
of the hottest wires. Thus, the most important values of $t$ are small ones,
say, less than half of the bus wires. But, our constructions in the following
sections will consider also solutions for large values of $t$, specifically,
any value of $t$. Similarly, $w$ is also usually small since large $w$ means a large
number of state transitions on all the wires, which might result in too many hot wires.
Finally, we also expect $e$ to be small, especially as it must be smaller than $w/2$ as otherwise
we won't be able to correct the errors.

Codes which satisfy Properties
$\A(\kern-0.5ptt\kern-0.5pt)$,\kern1pt$\B(\kern-0.5ptw\kern-0.5pt)$, and
${\bf C}(\kern-0.5pte\kern-0.5pt)$
simultaneously in every transmission, will be called
{\dfn $(n,t,w,e)$-low-power error-correcting cooling codes}
or $(n,t,w,e)$-LPECC codes for short. When a nonempty meaningful subset of the
three properties (property ${\bf C}(\kern-0.5pte\kern-0.5pt)$
is not interesting alone in our context) will be satisfied, only the parameters and description
related to this subset of properties will remain in the name. For example,
$(n,t,e)$-LPEC codes stands for $(n,t,e)$-low-power error-correcting codes.
Six such nonempty subsets exist and for each one we suggest coding schemes
and related codes. It the conclusion of Section~\ref{sec:conclude} a pointer
will be given, where each one of these six subsets was considered.

\vspace{1.00ex}
We view the collective state of
the $n$ wires during each~transmission
as a binary vector $\xxx = (x_1,x_2,\ldots,x_n)$.
The set of all such binary vectors is the \emph{Hamming $n$-space}
$\HH(n)=\{0,1\}^n$.
We will identify $\HH(n)$ with the vector space $\Fn$. Given any
$\xxx,\yyy \,{\in}\, \Fn$,~the
\emph{Hamming distance} $d(\xxx,\yyy)$ is the number of
positions~where $\xxx$ and $\yyy$ differ.
The \emph{Hamming weight}
of a vector $\xxx\,{\in}\,\Fn$,~denoted $\wt(\xxx)$,
is the number~of nonzero positions in $\xxx$.

Conventionally, a binary \emph{code \kern1pt$\C$ of length $n$} is simply
a~subset of $\Fn$. The elements of $\C$ are called \emph{codewords}.
Given a code~$\C$, its \emph{minimum distance} $d(\C)$ and
\emph{diameter} $\diam(\C)$ are defined as follows:
$$
d(\C) ~\deff~ \min_{\xxx,\yyy\in\C} d(\xxx,\yyy)
\hspace{3.00ex}\text{and}\hspace{3.00ex}
\diam(\C) ~\deff~ \max_{\xxx,\yyy\in\C} d(\xxx,\yyy)~.
$$
Later, in Sections \ref{spreads} and \ref{sec:LPCC}, we will need to modify
and generalize this conventional definition of binary codes in an important way.
This modification will be needed for codes which satisfy Property\,${\bf A}(t)$.

\vspace{2.00ex}
\section{Nonadaptive Low-Power Codes}
\label{anticodes}
\vspace{0.50ex}

\noindent\looseness=-1
The encoding schemes considered in this section belong to the
\emph{nonadaptive} kind, in that the choice which codeword to transmit
across the bus in the current transmission does not depend on codewords
that have been transmitted earlier. Such coding schemes are also
known as \emph{memoryless}.
The advantage of nonadaptive schemes is that they are 
simpler to implement: they do not need a continuously changing data model,
and they do not require memory to track the history of previous
transmissions.

In the nonadaptive case, an \emph{$n$-bit coding scheme}
for~a~source $\sS\subseteq\Fk$ is a triple
$\E= \langle \C,\cE,\cD \rangle$, where
\begin{enumerate}
\item $\C$ is a binary code of length $n$,
\item $\cE{:}~ \sS \to \C$
is a bijective map called an {\em encoding function},
\item $\cD{:}~ \C \to \sS$
is a bijective map called a {\em decoding function},
such that $\cD\bigl(\cE(\uuu)\bigr) = \uuu$ for all $\uuu \in \sS$.
\end{enumerate}
Encoding and decoding circuits that implement $\cE$ and $\cD$
are~inserted at the interface of the bus (see Figure\,1).

Suppose $\uuu,\vvv \,{\in}\, \sS$ are two words that are to be
communicated across the bus during consecutive transmissions.
In the absence
of a~coding scheme, the total switching activity
of the bus is then given by $|\{i:u_i\not=v_i\}|$.
This is precisely the Hamming\
distance $d(\uuu,\vvv)$, which could be as high as $k$.
If an $n$-bit~coding scheme is used, then $\xxx=\cE(\uuu)$
and $\yyy=\cE(\vvv)$ are transmitted instead.
The resulting total switching activity of the bus
is therefore $d(\xxx,\yyy)$, which is upper bounded~by~$\diam(\C)$.

It follows that the coding scheme \emph{satisfies
Property\,${\bf B}(w)$~if\, and only if $\diam(\C) \,{\le}\, w$}.
As the power consumption~of a~bus 
is directly related
to its total switching activity, we call such a~code $\C$
an {\dfn $(n,w)$-low-power code} ($(n,w)$-LP code for short).

In this section, we are interested in $(n,w)$-LP codes
that also achieve low peak temperatures by \emph{spreading~the
switching activities} among the bus wires as uniformly as possible.
In doing so, we are following the analysis of~\cite{CCL2006,Wang+:2006,Wang+:2007}
and the resulting \emph{thermal spreading} approach~\cite{CCL2006,Wang+:2007,SMM:2007}.
In order~to~quantify the thermal spreading achieved by a given
coding scheme $\E= \langle \C,\cE,\cD \rangle$, let us
treat the source $\sS$ as a random variable taking on values in $\Fk$,
and assume that $\sS$ is uniformly distributed. This is a common assumption
in bus analysis --- see, for example, \cite{SotiriadisChandrakasan:2002}.
With this assumption, the expected switching activity of wire $i$ is
given by
\be{mu}
\mu_i
\: = \:
\frac{1}{|\sS|^2} \hspace{-0.75ex}
\sum_{\uuu,\vvv\in \sS}  \bigl| \cE(\uuu)_i - \cE(\vvv)_i \bigr|
\: = \:
\frac{2r_i (|\C|-r_i)}{|\sS|^2}
\ee
where $r_i$ is the number of codewords $(x_1,x_2,\ldots,x_n)\,{\in}\,\C$
such that $x_i=1$.
If $\mu_1,\mu_2,\ldots,\mu_n$ are all equal, we 
say that~the~code $\C$ is {\dfn thermal-optimal}, since
the expected switching activities of the bus wires are
then uniformly distributed. This leads to the following
problem:
\be{problem1}
\begin{array}{c}
\textsl{Given $n$ and $w$, determine the maximum size of
a thermal-optimal $(n,w)$-low-power code}
\end{array}
\ee
The size of $\C$ \looseness=-1 
is important because we wish to minimize
the \emph{area overhead} introduced by our coding scheme.
This overhead is largely determined by the number $n-k$ of
additional wires that we need to encode a given source $\sS\subseteq\Fk$.
Clear\-ly, to encode such a source, we need 
a code $\C$ with $|\C| \geq 2^k$.~~

It is easy to see from \eq{mu}
that $\mu_1,\mu_2,\ldots,\mu_n$ are all equal~if
and only if $r_1,r_2,\ldots,r_n$ are all equal. Hence in
a thermal-optimal code $\C$, the number of codewords
$(x_1,x_2,\ldots,x_n)\kern1pt{\in}\,\C$~having \mbox{$x_i=1$}
is the same for all $i$. Such codes are said to be {\dfn equireplicate}
in the combinatorics literature. 
To construct such codes, we will need tools from the theory
of set systems as was suggested by Chee, Colbourn, and Ling~\cite{CCL2006}.

\subsection{Set Systems}
\vspace{-0.00ex}
\label{SetSystems}

\noindent
Given a positive integer $n$, the set $\{1,2,\ldots,n\}$
is abbreviated as $[n]$. 
For a finite set $X$ and $k\leq |X|$, we
define
$$
~~2^X ~\deff~ \bigl\{A:A\subseteq X\bigr\}
\hspace{2ex}\text{and}\hspace{2ex}
{X\choose k} ~\deff~ \bigl\{A\in2^X:|A|=k\bigr\}~~
$$
\looseness=-1
A {\em set system of order $n$} is a pair $(X,\cA)$,
where $X$ is a finite~set
of $n$ \emph{points} and $\cA\subseteq 2^X$.
The elements of $\cA$ are~called~\emph{blocks}.
A set system $(X,2^X)$ is 
a {\em complete set system}.
The \emph{replication number} of 
$x\,{\in}\,X$ is the number of
blocks containing~$x$. A set system is \emph{equireplicate} if
its replication numbers are all equal.



\looseness=-1
There is a natural one-to-one 
correspondence between the Hamming $n$-space
$\Fn$ and the complete set system $([n],2^{[n]})$ of order $n$.
For a vector $\xxx = (x_1,x_2,\ldots,x_n) \,{\in}\, \Fn$,
the~\emph{support of $\xxx$} is defined as
$$
\supp(\xxx)
\,~\deff~\,
\bigl\{i \in [n] ~:~ x_i \ne 0\bigr\}
$$
With this, the positions of vectors in $\Fn$
correspond to~points
in $[n]$, each vector $\xxx\,{\in}\,\Fn$ corresponds to the block
$\supp(\xxx)$, and $d(\xxx,\yyy)=|\supp(\xxx)\,\triangle\,\supp(\yyy)|$,
where $\triangle$ stands for the symmet\-ric difference.
It follows from the above that there is
a $1$-$1$ correspondence between the set of all 
codes of length~$n$
and the set of all set systems of order~$n$.
Thus we may speak~of the \emph{set system of a code}
or the \emph{code of a set system}.

\subsection{Thermal-Optimal Low-Power Codes}
\label{sec:LP_optimal}

\noindent
The set system $([n],\cA)$
of a thermal-optimal $(n,w)$-low-power code
is defined by the following properties:
\begin{enumerate}
\item $|A_1 \,\triangle\, A_2| \leq w$ for all $A_1,A_2\in\cA$, and
\item $([n],\cA)$ is equireplicate.
\end{enumerate}
It follows that our problem in \eq{problem1} can be recast
as an equivalent problem in extremal set systems, as follows:
\be{problem2}
\begin{array}{c}
\textsl{Given $n$ and $w$, determine $T(n,w)$, the maximum size of an equireplicate set}\\
\textsl{system $(X,\cA)$ of order $n$ such that $|A_1 \,\triangle\, A_2|\leq w$ for all $A_1,A_2\,{\in}\,\cA$}
\end{array}
\ee
If the equireplication condition is removed,
the resulting set system is known
as an \emph{anticode of length $n$ and diameter~$w$}.
Hence, thermal-optimal
low-power codes are equivalent to~\emph{equi- replicate anticodes}.
Anticodes in general, and the size of anticodes of maximum size have
been a subject of intensive research in coding theory, see~\cite{AAK01,AhBl08,BuEt15,Etz11,ScEt02,MWS}
and references therein.

The determination \looseness=-1
of equireplicate anticodes of maximum size appears to be a new problem,
also to the combinatorics and coding theory communities.
However, the maximum size~of an anticode has been completely
determined by Kleitman~\cite{Kleitman:1966}, and even earlier by
Katona \cite{Katona:1964}, in a different but equivalent setting.
Thus the following theorem is from \cite{Katona:1964}
and \cite{Kleitman:1966}.
\begin{theorem}
\label{Kleitman}
Let\/ $\fT(n,w)$ be 
the maximum number of blocks in a set system $([n],\cA)$
with $|A_1 \,\triangle\, A_2|\leq w$ for all $A_1,A_2\,{\in}\,\cA$.
Then\vspace{-1.50ex}
\begin{equation*}
\fT(n,w)
\,=\,
\begin{cases}
~\displaystyle\sum_{i=0}^{w/2} \hspace{-2pt}\binom{n}{i}
&
\text{if\, $w \equiv 0 \hspace*{-1.50ex}\pmod{2}$}
\\
\displaystyle {\kern-1ptn{-}1\kern-0.5pt\choose \frac{w-1}{2}}
\,+
\sum_{i=0}^{\frac{w-1}{2}} \hspace{-2pt}\binom{n}{i}
&
\text{if\, $w \equiv 1 \hspace*{-1.50ex}\pmod{2}$}
\end{cases}~.
\end{equation*}
For all even $w$,
an extremal set system $([n],\cA)$ with $\fT(n,w)$ blocks
is given by:
\be{even-anticode}
\cA \, = \, \bigcup_{i=0}^{w/2}{\kern-0.5pt[n]\kern-0.5pt \choose i}~.
\ee
If $w$ is odd, let $x$ be any fixed element of $[n]$. Then
an extremal set system $([n],\cA)$ 
is given by:
\be{odd-anticode}
\cA
\, = \,
\bigcup_{i=0}^{\frac{w-1}{2}}\kern-2pt{\kern-0.5pt[n]\kern-0.5pt \choose i}
~~\bigcup~
\left\{
A\cup\{x\}~:~A \in\kern-1pt
{\kern-0.5pt[n]{\setminus}\{\kern-1pt x\kern-1pt\}\kern-1pt \choose \frac{w-1}{2} }
\right\}~.
\ee
\end{theorem}

\vspace{1.50ex}
We observe here that when $w$ is even, the extremal set system in
\Tref{Kleitman} is equireplicate. It consists of all the vectors
of length $n$ and weight at most $w/2$. Hence, we have the following
result, which solves \eq{problem1} and \eq{problem2} for even~$w$.
\begin{corollary}\vspace{-0.50ex}
$$
T(n,w)\kern1pt= \sum_{i=0}^{w/2} \hspace{-2pt}\binom{n}{i}
\hspace{3.00ex}
\text{when\, $w \equiv 0 \hspace*{-1.75ex}\pmod{2}$}~.
$$
\end{corollary}

\looseness=-1
The situation when $w$ is odd is much more difficult. The set system
in \eq{odd-anticode} is not equireplicate.
In particular, we do not know if there exists an equireplicate
anticode of order $n$ and diameter $w$ having size $\fT(n,w)$.
Hence,~from \Tref{Kleitman}~we can derive only that for all odd $w$, we have:
\begin{equation}
\label{bounds}
\sum_{i=0}^{\frac{w-1}{2}} \kern-2pt \binom{n}{i}
\, \leq \ T(n,w) \ \leq \
 {\kern-1ptn{-}1\kern-0.5pt\choose \frac{w-1}{2}}
\,+\,
\sum_{i=0}^{\frac{w-1}{2}} \hspace{-2pt}\binom{n}{i}~.
\end{equation}
The left hand side of the equation is obtained from a code which consists
of all the vectors of length $n$ and weight at most $(w-1)/2$. The right hand side
is obtained from the upper bound on $\fT(n,w)$ given in Theorem~\ref{Kleitman}.

%
The next three propositions
establishes some exact values~of $T(n,w)$ for odd $w$.
\begin{corollary}\vspace{-1.50ex}
\label{cor:small1}
\begin{align*}
T(n,1) &= \,1 \hspace{5.10ex}\text{for\, $n \ge 2$}~.
\end{align*}
\end{corollary}
\begin{proof}
Since the distance between two different vectors
of length $n$ and weight one is two, it follows that $T(n,1)=1$ when $n \geq 2$.
A code with maximum size consists
of the unique all-zero vector of length $n$.
\end{proof}

\pagebreak

\begin{proposition}\vspace{-1.50ex}
\label{prop:small_n-1}
\begin{align*}
T(n,n{-}1) &= \,2^{n-1} \hspace{1.60ex}\text{for\, $n \ge 3$} \hspace{1.50ex}~.
\end{align*}
\end{proposition}
\begin{proof}
When the distance between two codeword is at most $n-1$, the code cannot contain
two complement codewords and hence its size is at most $2^{n-1}$. For odd $n$, an
equireplicate set system of size $2^{n-1}$ is obtained from all vectors of length $n$
and even weight. For even $n$, we give the following construction (which also
works for any odd $n \geq 5$ if induction is
applied). Let $([n-1],\cA)$ be a set system which attains $T(n-1,n-2)=2^{n-2}$ and each
element is contained in $2^{n-3}$ blocks.
We define the following set system $([n],\cB)$.
$$
\cB \,~\deff~\, \{ X \cup \{n\} ~:~ X \in \cA \} \bigcup \{ X  ~:~ X \in \cA \} ~.
$$
We claim that $\cB$ is an equireplicate set system which attains $T(n,n-1)=2^{n-1}$
and each element is contained in $2^{n-2}$ blocks of $\cB$.
Clearly, $| \cB | = 2 |\cA| = 2^{n-1}$ and the fact that $\cA$ does not contain
complement blocks immediately implies that also $\cB$ does not contain complement blocks.
Finally, $\cA$ is equireplicate set system of size $2^{n-2}$ and each element of $[n-1]$ is contained
in $2^{n-3}$ blocks of~$\cA$. Therefore,
each $i \in [n]$ is contained in $2^{n-2}$ blocks in $\cB$.
Hence, $\cB$ is an equireplicate set system which attains $T(n,n-1)=2^{n-1}$.
\end{proof}

\begin{proposition}\vspace{-1.50ex}
\label{prop:small3}
\begin{align*}
T(n,3) &= \,n+1 ~~\text{for\, $n \ge 5$}~.
\end{align*}
\end{proposition}
\begin{proof}
First note, that by Proposition~\ref{prop:small_n-1} we have $T(4,3)=8$. Also,
$T(n,2)= \sum_{i=0}^1 \binom{n}{i} =n+1$ for all $n \geq 2$.
Finally, it is easy to verify that $T(5,3)=6$,
These facts will be used in the current proof that $T(n,3)=n+1$ if $n>4$.
Let $\C$ be the largest possible equireplicate anticode of length $n>5$ and
diameter 3.

Let $\xxx$ and $\zzz$ be two codewords of $\C$ such that $d(\xxx,\zzz)=3$.
W.l.o.g. $\xxx$ and $\zzz$ differ in the last three coordinates
and the first $n-3$ coordinates in both have $x_1,x_2,...,x_{n-3}$.

Let $\alpha \beta \gamma$ and $\bar{\alpha} \bar{\beta} \bar{\gamma}$
be the values of the last three coordinates in $\xxx$ and $\zzz$, respectively.
There is no other codeword in $\C$
which ends with either $\alpha \beta \gamma$ or $\bar{\alpha} \bar{\beta} \bar{\gamma}$
since such a codeword should also start with
$x_1,x_2,...,x_{n-3}$, to avoid distance greater than 3 from either $\xxx$ or $\zzz$,
and hence such a codeword will be equal to either $\xxx$ or $\zzz$.
Since each one of the last three columns has at least one \emph{zero}
and at least one \emph{one} and the anticode is equireplicate, it follows that the weight of a column
is at least $1$ and at most $|\C|-1$.
The Hamming distance of any two of 100, 010, 001, 111 is 2 and hence
the prefixes of length $n-3$ related to the codewords ending with these suffices
differ in at most one coordinate. Since anticode with diameter one
has at most two codewords, it follows that all codewords with these suffices (if differ) have
different values in exactly the same coordinate (in the prefix of length $n-3$).
The same argument holds also for the codewords ending with 011, 101, 110, and 000
(they have the same values in $n-4$ out of the first $n-3$ coordinates).
Note that since the suffix of either $\xxx$ or $\zzz$ is in $\{100,010,001,111\}$
and the other suffiex is in $\{011,101,110,000\}$,
it follows that all the other $n-5$ coordinates (which don't have different
values) of $\xxx$ and $\zzz$ have the same values for all the codewords.
Since $n>5$, it follows that each one of the
$n-5$ coordinates (which exist) forms either a column of \emph{zeroes} or a column of \emph{ones}.
If the column consists of \emph{zeroes} we have a contradiction since
the weight of the column is 0 which is smaller than $1$.
If the column consists of \emph{ones} we have a contradiction since
the weight of the column is $|\C|$ which
is greater than $|\C|-1$.

Thus, for $n>5$ there are no two codewords for which the Hamming distance is three.
and hence for $n>6$, $T(n,3)=T(n,2)=n+1$, which completes the proof.
\end{proof}

For other odd values of $w$, we start with the extremal anticode $\cA$
of diameter $w-1$ in \eq{even-anticode} and add blocks to $\cA$ while
maintaining the equireplication requirement. Such blocks must contain
exactly $(w{+}1)/2$ points to make sure that their distance with the
blocks of $\cA$ with $(w-1)/2$ points, or less, will not exceed $w+1$.
Any two blocks with $(w{+}1)/2$ points must intersect in at least one point as otherwise
their distance will be $w+1$.
Interestingly, these properties precisely define
a \emph{regular uniform quorum system} of rank $(w{+}1)/2$.
Quorum systems have been studied extensively in the literature
on fault-tolerant and distributed computing --- see~\cite{Vukolic:2012}
for a recent survey. There are many types of such systems. For example, if
$(w+1)/2 =q+1$, $n=q^2 +q+1$, and $q$ is a prime power, then an optimal such system consists of
$q^2 +q+1$ blocks which form a projective plane of order $q$~\cite[p. 224]{vLWi92}.
In other words
\begin{proposition}
If $w=2q+1$, $q$ is a prime power, then
$$T(q^2 +q+1,2q+1) = \sum_{i=0}^q \binom{q^2+q+1}{i} + q^2+q+1~.$$
\end{proposition}
For a proof of this proposition and other
similar constructions, we refer the reader to the work in~\cite{CDS01},
where such systems are constructed from combinatorial designs.

\vspace{2.00ex}
\section{Cooling Codes}
\label{spreads}
\vspace{0.50ex}

\noindent
Unfortunately, it is not possible to satisfy Property\,$\A(t)$
with nonadaptive coding schemes, even for $t=1$. Indeed, suppose
we wish to avoid state transitions on just the one hottest wire,
say wire $i$. 
If the encoder does not know the
current state of wire $i$,
the only way to guarantee that there is no state transition is
to have $x_i = y_i$ for \emph{any} two codewords
$(x_1,x_2,\ldots,x_n)$ and $(y_1,y_2,\ldots,y_n)$.
Since \emph{any of~the $n$ wires} could~be~the hottest,
we must have
$(x_1,x_2,\ldots,x_n) = (y_1,y_2,\ldots,y_n)$.
Thus all~co\-dewords are the same,
and no communication is possible.

In this section, we shall see 
that if the encoder and decoder
keep track of just \emph{one previous transmission}
then Property\,$\A(t)$ can be satisfied for {any $t$}
with only $t+1$ additional wires if $2(t+1) \leq n$,
by using spreads and partial spreads, notions from projective geometry.
If $t+1 > n/2$ we propose a construction based on a sunflower for which the construction
for $2(t+1) \leq n$ can be viewed as a special case.
Finally, we provide a road map for our best lower bounds
on the size of $(n,t)$-cooling codes in general, and in particular
when $1 \leq t < n \leq 100$.

\subsection{Differential Encoding and Decoding}
\vspace{-0.00ex}
\label{sec:Differential}

\noindent
The main idea of our differential encoding method is to
encode the data to be communicated across
the bus in the \emph{difference} between the current transmission
and the previous one. Similar ideas have been used in digital
communications and in magnetic recording, among other applications.

Why is differential coding useful? The most useful feature
in our context is this.
When we use the differential encoding method
to transmit a codeword $\xxx = (x_1,x_2,\ldots,x_n)$,
there is a state transition on wire $i$ if and only if
$x_i = 1$, and the total number of transitions is precisely
$\wt(\xxx)$, the Hamming weight of $\xxx$.
This makes it possible to
reduce the area~overhead significantly beyond the best
overhead achievable with nonadaptive schemes.
For example, under differential encoding,
a~code $\C$ satisfies Property\,$\B(w)$ --- and so is
an $(n,w)$-LP code --- if and only if $\wt(\xxx) \le w$
for all $\xxx \,{\in}\, \C$. It~follows that the thermal-optimal
$(n,w)$-LP code of maximum size is given by
\be{J(n,w)}
J^+(n,w)
\,~\deff~\,
\bigl\{\,\xxx \in \{0,1\}^n ~:~ \wt(\xxx) \le w \,\bigr\}
\ee
This set, distinguished from the Johnson space defined by
$J(n,w) \deff \{ \xxx \in \{0,1\}^n : \wt(\xxx) = w \}$, is clearly equireplicate
and its size is much larger than the size of the largest
anticode of diameter $w$ (cf.~\Tref{Kleitman}),
for both odd and even~$w$.

\subsection{Definition of Cooling Codes}
\vspace{-0.00ex}
\label{sec:CoolDefinition}

\noindent
Even under differential encoding, it is still not possible
to satisfy Property\,$\A(t)$ with conventional binary codes.
To see this, again suppose that we wish to avoid transitions
on just the one hottest wire, say wire $i$.
With differential encoding, in order~to guarantee no state
transitions on wire $i$ while transmitting a~co\-deword
$\xxx = (x_1,x_2,\ldots,x_n)$, we must have $x_i = 0$.
But, once again, any of the $n$ wires could be the hottest,
which implies that $\xxx = \zero$. Since this must hold
for any codeword, it follows that
$\C = \{\zero\}$ and no communication is possible.

Consequently, we henceforth modify our notion of a code $\C$
as follows. The elements of $\C$ will be \emph{sets of binary vectors}
of length $n$, say $C_1,C_2,\ldots,C_M$.
We will refer to $C_1,C_2,\ldots,C_M$ as \emph{codesets}.
We do not require these codesets to be of~the~same size,
but we do require them 
to be disjoint: 
$C_i \cap C_j = \varnothing$ for all $i \ne j$.
The elements of each codeset $C_i$ will be 
called \emph{codewords}. The goal is to guarantee
that no matter which codeset $C_i$ is chosen, for
each possible designation 
of $t$ wires~as~the~hot\-test,
there is at least one codeword in $C_i$ with zeros
on the~corresponding $t$ positions. This leads to
the following definition.~~
\begin{definition}
\label{cooling-def}
\,For positive integers $n$ and $\kern1pt t\kern-1pt <\kern-1pt n$, an
{\dfn $(n,\kern-1pt t)$-cooling code} $\C$ of size $M$ is 
defined as a set $\{C_1,C_2,\ldots,C_M\}$, where $C_1,C_2,\ldots,C_M$
are disjoint subsets of\/ $\Fn$ satisfying the following property:
for any set\/ $\cS \subset [n]$ of size $|\cS| = t$ and for~all
$i \,{\in}\,\kern-1pt [M]$,
there exists a codeword\/ $\xxx \,{\in}\, C_i\kern-1pt$
with\/ $\supp(\xxx) \cap \cS\kern-1pt = \varnothing$.\vspace*{0.50ex}
\end{definition}

\hspace*{-0.75ex}Given the foregoing definition of cooling codes, we also~need
to modify our notions of an encoding function and a decoding function,
introduced in Section\,\ref{anticodes}.
As before, we assume that the data to be communicated across the bus
is represented by a source $\sS$ taking on some $M \le 2^k$ values
in $\Fk$. In contrast to Section\,\ref{anticodes}, we no longer
need to assume that $\sS$ is uniformly distributed --- in fact, no
probabilistic model for $\sS$ is required. On the other hand, the
input to the encoding function $\cE$ now comprises, in addition to
a word $\uuu \,{\in}\, \sS$, also a set $\cS \subset [n]$ of size~$t$
representing the positions of the $t$ hottest wires. We let~~
$$
\sC \:~\deff~\,
C_1 \cup C_2 \cup \cdots \cup C_M~.
$$
Then the output of the encoding function $\cE$ 
is a vector $\xxx \,{\in}\, \sC$ such that
$\supp(\xxx) \cap \cS = \varnothing$.
For every possible 
$\cS$, the function $\cE(\cdot,\cS)$ is
a bijective map from $\sS$ to $\C$.
Since the codesets $C_1,C_2,\ldots,C_M$
are disjoint, this allows 
the decoding function $\cD$ to recover 
$\uuu \in \sS$ from the encoder output $\xxx \in \sC$.
We summarize the foregoing discussion in the next 
definition.
\begin{definition}
\label{cool-encoding}
For integers $n$ and $\kern1pt t\kern-1pt <\kern-1pt n$, an
{\dfn $(n,\kern-1pt t)$-cooling coding scheme}
for a~source\/ $\sS\subseteq\Fk$ is a triple\/
$\E= \langle \C,\cE,\cD \rangle$,~where\vspace{0.50ex}
\begin{enumerate}
\item[\rm 1)] The code\/ $\C$ is an $(n,t)$-cooling code;\vspace{0.750ex}
\item[\rm 2)]
The encoding function\/
\smash{$\cE{:}~\kern1pt \sS\times\kern-2pt%
\bigl(\kern-2pt{\mbox{\scriptsize$[$}n\mbox{\scriptsize$]$} \atop t}\kern-2pt\bigr)
\kern-1pt\to\kern-1pt \sC$}
is such that for~all
$\cS \,{\subset}\, [n]$ of size $t$ 
and all $\uuu \,{\in}\, \sS$, we have
$$
\supp\bigl(\cE(\uuu,\cS)\bigr) \cap \cS \,=\, \varnothing\ ;
$$ 
\item[\rm 3)]
The decoding function\/
$\cD{:}~\kern1pt \sC \to \sS$ is such that for all $\uuu \,{\in}\, \sS$,
we have $\cD\bigl(\cE(\uuu,\cS)\bigr) = \uuu$ regardless of the value of~$\cS$.
\end{enumerate}
\end{definition}

\noindent\looseness=-1
It follows immediately from \Dref{cool-encoding} that,
under differ\-en\-tial encoding, an $(n,t)$-cooling coding scheme
satisfies Property\,$\A(t)$ by avoiding state transitions on
the $t$ hottest wires, which are represented by the subset $\cS$.~~

\noindent
{\bf Example.} \looseness=-1
Consider $n=6$ and $t=2$. An $(6,3)$-cooling code is given by the following eight codesets,
\begin{align*}
C_{000} &=\{ 100000 , 101011 , 111010 , 001011 , 010001 , 110001 , 011010 \},\\
C_{001} &=\{ 010000 , 100011 , 011101 , 110011 , 111110 , 101110 , 001101 \},\\
C_{010} &=\{ 001000 , 100111 , 111000 , 101111 , 011111 , 010111 , 110000 \},\\
C_{011} &=\{ 000100 , 100101 , 011100 , 100001 , 111001 , 111101 , 011000 \},\\
C_{100} &=\{ 000010 , 100100 , 001110 , 100110 , 101010 , 101000 , 001100 \},\\
C_{101} &=\{ 000001 , 010010 , 000111 , 010011 , 010101 , 010100 , 000110 \},\\
C_{110} &=\{ 110110 , 001001 , 110101 , 111111 , 111100 , 001010 , 000011 \},\\
C_{111} &=\{ 011011 , 110010 , 101100 , 101001 , 011110 , 000101 , 110111 \}.
\end{align*}
We index the codesets by all three-bit messages.

To verify Property\,$\A(2)$, we explicitly describe the encoding function
$\cE: \sS\times  \bigl(\kern-2pt{\mbox{\scriptsize$[$}6\mbox{\scriptsize$]$} \atop 2}\kern-2pt\bigr)
\to \sC$ via the following lookup table.

\begin{center}
\begin{tabular}{|c| cccc cccc|}
\hline
$\cS$ & 000 & 001 & 010 & 011 & 100 & 101 & 110 & 111\\
\hline
$\{ 1 , 2 \}$ & 001011 & 001101 & 001000 & 000100 & 000010 & 000001 & 001001 & 000101 \\
$\{ 1 , 3 \}$ & 010001 & 010000 & 010111 & 000100 & 000010 & 000001 & 000011 & 000101 \\
$\{ 1 , 4 \}$ & 001011 & 010000 & 001000 & 011000 & 000010 & 000001 & 001001 & 011011 \\
$\{ 1 , 5 \}$ & 010001 & 010000 & 001000 & 000100 & 001100 & 000001 & 001001 & 000101 \\
$\{ 1 , 6 \}$ & 011010 & 010000 & 001000 & 000100 & 000010 & 010010 & 001010 & 011110 \\
$\{ 2 , 3 \}$ & 100000 & 100011 & 100111 & 000100 & 000010 & 000001 & 000011 & 000101 \\
$\{ 2 , 4 \}$ & 100000 & 100011 & 001000 & 100001 & 000010 & 000001 & 001001 & 101001 \\
$\{ 2 , 5 \}$ & 100000 & 001101 & 001000 & 000100 & 100100 & 000001 & 001001 & 101100 \\
$\{ 2 , 6 \}$ & 100000 & 101110 & 001000 & 000100 & 000010 & 000110 & 001010 & 101100 \\
$\{ 3 , 4 \}$ & 100000 & 010000 & 110000 & 100001 & 000010 & 000001 & 000011 & 110010 \\
$\{ 3 , 5 \}$ & 100000 & 010000 & 110000 & 000100 & 100100 & 000001 & 110101 & 000101 \\
$\{ 3 , 6 \}$ & 100000 & 010000 & 110000 & 000100 & 000010 & 010010 & 110110 & 110010 \\
$\{ 4 , 5 \}$ & 100000 & 010000 & 001000 & 100001 & 101000 & 000001 & 001001 & 101001 \\
$\{ 4 , 6 \}$ & 100000 & 010000 & 001000 & 011000 & 000010 & 010010 & 001010 & 110010 \\
$\{ 5 , 6 \}$ & 100000 & 010000 & 001000 & 000100 & 100100 & 010100 & 111100 & 101100 \\
\hline
\end{tabular}
\end{center}

\subsection{Bounds on the Size of Cooling Codes}
\label{sec:CoolBounds}

\noindent
In this subsection, we show that realizing an $(n,t)$-cooling~coding
scheme requires at least $t+1$ additional wires. That is, the number
of bits
that can be communicated over an $n$-wire bus while satisfying
Property\,$\A(t)$ is at most $k \le n-t-1$. In the next subsection, we
will present a construction that achieves this bound.
Herein, let us begin with the following lemma. 

\begin{lemma}
\label{lem:upperCC}
Let\/ $\C$ be an $(n,t)$-cooling code of size $|\C| = M$.
Then
\be{sphere-packing}
M
\ \le \
\frac{t!(n{-}t)!}{n!}\,
\sum_{w=0}^{n-t} \kern-2pt\binom{n}{w}\kern-1pt\binom{n{-}w}{t}
\ = \,\
2^{n-t}~.
\vspace{0.50ex}
\ee
\end{lemma}

\begin{proof}
For convenience, we will refer to sets $\cS \subset [n]$~of
size $|\cS| = t$ as \emph{$t$-subsets}. Given a $t$-subset $\cS$
and a vector $\xxx \,{\in}\, \Fn$, we shall say that
\emph{$\xxx$ covers $\cS$} if
$\supp(\xxx) \cap \cS = \varnothing$.
Observe that a vector of weight $w$ covers exactly
${n{-}w \choose t}$ different $t$-subsets. Therefore,
the total number of $t$-subsets (counted with multiplicity)
covered by \emph{all} the vectors in~$\Fn$ is given by
$$ 
N(n,t)
\,~\deff~
\sum_{w=0}^{n-t} \kern-2pt\binom{n}{w}\kern-1pt\binom{n{-}w}{t}~.
$$ 
Now~~consider a codeset~~$C$ in~~$\C$.~~By definition,
for~~any  {$t$-subset}~$\cS$,
there is 
at least one codeword $\xxx \,{\in}\, C$
that covers $\cS$. Hence,~the total number of $t$-subsets
(possibly counted with multiplicity) covered by {all}
the codewords of $C$ is at least ${n \choose t}$.
Since this holds for each of the $M$ codesets in $\C$,
the total number of $t$-subsets (again, counted with multiplicity)
covered by {all} the vectors~in
$\sC = C_1 \cup C_2 \cup \cdots \cup C_M$ is at least
$M{n \choose t}$. But, since the codesets
$C_1,C_2,\ldots,C_M$ are disjoint and $\sC \subseteq \Fn$,
this number cannot exceed $N(n,t)$, i.e. $M{n \choose t} \leq N(n,t)$, and the lemma follows.
\vspace{1.25ex}
\end{proof}

\looseness=-1
The proof of Lemma~\ref{lem:upperCC} can be shorten by considering
any given $t$-subset of coordinates that must be zeroes in at least
one codeword of each codeset. There are $2^{n-t}$ such codewords of length $n$
and hence the maximum number of codesets is $2^{n-t}$.
Nevertheless, there is one advantage in presenting the longer proof.
It follows from the proof of \Lref{lem:upperCC} that an $(n,t)$-cooling
code $\C$ of size $|\C| = 2^{n-t}$ would be \emph{perfect}. In
such a code,~the codesets $C_1,C_2,\ldots,C_M$ form a partition
of $\Fn$ and each~of these codesets
is a perfect covering of ${[n] \choose t}$, i.e.
all words with exactly $t$ \emph{zeroes}.
Using these observations,
we can prove that such cooling
codes do not exist, unless $t=1$ or $t\geq n-1$.

\begin{proposition}
\label{no_perfect}
Perfect $(n,t)$-cooling codes do not exist, unless $t=1$ or $t \geq n-1$.
\end{proposition}
\begin{proof}
W.l.o.g. assume that the codeset $C_1$ contains the codeword $\xxx$ which starts with
a \emph{one} followed by $n-1$ \emph{zeroes}. The only words which are not covered by
this codeword are all those words which start with a \emph{zero} and have exactly $t-1$ \emph{zeroes}
in the other $n-1$ coordinates. Let $\zzz$ be one of these
$\binom{n-1}{t-1}$ words. A codeword $\yyy$ which covers $\zzz$
must starts with a \emph{zero}. If $\yyy$ has at least $t+1$
\emph{zeroes} then it covers words with $t$ \emph{zeroes}
which are also covered by $\xxx$, and therefore $C_1$ won't be a perfect covering.
Hence, $C_1$ contains $\xxx$ and the $\binom{n-1}{t-1}$ codewords,
which start with a \emph{zero}, and have exactly $t$ \emph{zeroes}.

Now, assume w.l.o.g. that $C_2$ contains the codeowrd which starts with $n-1$ \emph{zeroes}
and ends with a \emph{one}. With similar analysis as for $C_1$ this codeset contains the
$\binom{n-1}{t-1}$ codewords, which end with a \emph{zero}, and have exactly $t$ \emph{zeroes}.

Now, by this analysis, both $C_1$ and $C_2$ must contain the $\binom{n-2}{t-2}$ words which
start and end with a \emph{zero} and have exactly $t$ \emph{zeroes}. Therefore, if
$\binom{n-2}{t-2} >0$ a perfect code cannot exist.
$\binom{n-2}{t-2} >0$, unless $t=1$ or $t \geq n-1$.

\begin{itemize}
\item If $t=1$, then there exists a perfect $(n,1)$-cooling code, with
$2^{n-1}$ codesets, each one contains a pair of complement codewords.

\item If $t=n$, then there exists a perfect $(n,n)$-cooling code,
with one codeset which contains all codewords of length~$n$.

\item If $t=n-1$ then there  exists a perfect $(n,n-1)$-cooling code,
with two codesets, one contains the all-zero codeword and the second contains
all codewords of weight \emph{one}. All other words can be partitioned arbitrarily between these two codesets.
\end{itemize}
$~~~~~~~~~~~~~~~~~~~~~~~~~~~~~~~~~~~~~~~~~~~~~~~~~~~~~~~~~~~~~~$
\end{proof}

Lemma~\ref{lem:upperCC} and Proposition~\ref{no_perfect} imply the following result.

\begin{corollary}
\label{thm5}
If $1 < t < n-1$, then the size of any $(n,t)$-cooling code\/ $\C$ is bounded by
$|\C| \le 2^{n-t}\kern-1pt-1$. Consequently, such a code cannot
support the transmission of\/ $n-t$ or more bits over
an $n$-wire bus.
\end{corollary}

Denote the maximum size of an $(n,t)$-cooling code by $C(n,t)$.
In the following subsection we will obtain lower bounds on $C(n,t)$.
We start in this subsection with a few simple bounds and values
of $C(n,t)$.

\begin{corollary}
\label{cor:boundsC}
$~$
\begin{enumerate}
\item $C(n,1)=2^{n-1}$;
\item $C(n,n-1)=2$;
\item If $2 \leq t \leq n-2$, then $C(n,t) \leq 2^{n-t}-1$.
\end{enumerate}
\end{corollary}

\vspace{0.1cm}

\begin{proposition}
\label{pro:simpleC}
If $2 \leq t \leq n-2$, then $C(n,t) \geq n-t+1$.
\end{proposition}
\begin{proof}
For $i\in [n-t+1]$, let $C_i$ contain the set of all binary words of length $n$ and weight $n-t-i+1$, i.e.
each codeword in $C_i$ has $t+i-1 \geq t$ \emph{zeroes}. Clearly, the code $\C$ with these
codesets is an $(n,t)$-cooling code.
\end{proof}

By Corollary~\ref{cor:boundsC} and Proposition~\ref{pro:simpleC} we have the following
\begin{corollary}
$C(n,n-2)=3$.
\end{corollary}

\subsection{Construction of Optimal Cooling Codes}
\label{sec:CoolCodes}

\noindent
In this subsection, we construct $(n,t)$-cooling codes that support
the transmission of up to $n-t-1$ bits over an $n$-wire bus.
By Corollary~\ref{thm5}, such cooling codes are \emph{optimal}.
Our construction is based on the notion of
\emph{spreads} and \emph{partial spreads} in projective geometry.
We will give the related equivalent definition for vector spaces. In this
section we need only spread and partial spread over $\Ftwo$, but
we will define and discuss them over $\Fq$, as those will be
required later in our exposition.~~~

Loosely speaking, a partial $\tau$-spread of the vector
space $\Fqn$ is a collection of disjoint
$\tau$-dimensional subspaces of~$\Fqn$. Formally,
a collection $V_1,V_2,\ldots,V_M$ of $\tau$-dimensional
subspaces of $\Fqn$ is said to be a \emph{partial $\tau$-spread of\/ $\Fqn$} if
\be{disjoint}
\hspace*{2ex}V_i \cap \kern-.5ptV_j\kern1pt = \{\zero\}
\hspace{2ex}\text{for all\, $i \ne j$} \hspace*{1ex}~,
\vspace{-0.75ex}
\ee
\be{union}
\Fqn \supseteq \kern1pt V_1 \cup V_2 \cup \cdots \cup V_M \hspace*{1ex}~.
\ee
If the $\tau$-dimensional subspaces form a partition of $\Fqn$ then
the partial $\tau$-spread is called a $\tau$-spread.
It is well known that such $\tau$-spreads exist if and
only~if~$\tau$~divides $n$, in which case
$M = (q^n{-}1)/(q^\tau{-}1) > q^{n-\tau}$.
For~the case where $\tau$ does not divide $n$,
\emph{partial $\tau$-spreads} with~\mbox{$M \ge q^{n-\tau}$}
have been constructed in~\cite[Theorem\,11]{EV:2011}.
Let $M_q(n,\tau)$ be the maximum size of a partial $\tau$-spread.
The value of $M_q(n,\tau)$ has been considered for many years in projective geometry,
and a survey with the known results is given in~\cite{EtSt16}.
Recently, there has been lot of activity and the question has been almost completely
solved~\cite{Kur15,Kur16,NaSi16}.~~

\begin{theorem}
\label{thm:spread_cool}
Let\/ $V_1,\kern-1ptV_2,\ldots,V_M$ be a partial $(t{+}1)$-spread of~$\Fn$, and
define the code\kern1.5pt\ $\C = \{V^*_1,V^*_2,\ldots,V^*_M\}$,
where\/
$V^*_i \kern-1.5pt=\kern-1pt V_i \kern1.5pt{\setminus} \{\zero\}$
for all\/ $i$.
Then\/ $\C$ is an $(n,t)$-cooling code of size $M \ge 2^{n-t-1}$.
\end{theorem}

\begin{proof}
\hspace*{-0.50ex}It
is obvious from \eq{disjoint} that the $M$ codesets~of~$\C$~are
disjoint subsets of $\Fn$.
It remains to verify that for any set $\cS \,{\subset}\, [n]$ of
size $t$, each of $V^*_1,\kern-1ptV^*_2,\ldots,V^*_M$ contains at
least~one vector whose support is disjoint from $\cS$.
To this end, consider an arbitrary \kern-1pt$(t{+}1)$-dimensional
subspace $V\kern-1pt$ of $\Fn$, and suppose
$\{\vvv_1,\vvv_2,\ldots,\vvv_{t+1}\}$ is a basis for $V$.
Let $\vvv'_1,\vvv'_2,\ldots,\vvv'_{t+1}$ denote the projections
of the basis vectors on the $t$ positions~in~$\cS$. These
$t+1$ vectors lie in a $t$-dimensional vector space --- the
projection of $\Fn$ on $\cS$. Hence, these vectors must be
linearly dependent, and there exist 
binary coefficients $a_1,a_2,\ldots,a_{t+1}$, not all zero, with
$
a_1\vvv'_1 + a_2\vvv'_2 + \cdots + a_{t+1}\vvv'_{t+1} \kern-1pt= \zero
$.
\kern1ptBut~then
$
\xxx = a_1\vvv_1 + a_2\vvv_2 + \cdots + a_{t+1}\vvv_{t+1}
$
is a nonzero vector in $V$ whose support does not include
any of the positions in $\cS$.~As this holds for
an arbitrary $(t{+}1)$-dimensional subspace, it must
hold for each of the subspaces $V_1,\kern-1ptV_2,\ldots,V_M$
in the partial spread.
\end{proof}

\vspace{1cm}

Whether we start with a spread or a partial spread, the~size
of the code $\C$ constructed in \Tref{thm:spread_cool} will usually be
strictly larger than $2^{n-t-1}$. We omit also the possibility for adding
another codset which contains the all-zero codeword since the encoding of
$k$-bit data generated by the source $\sS$, requires only $2^k$ codesets. In order to use such a code
to communicate $k = n-t-1$ bits, one can choose a subset
of $\C$, with $2^k$ codesets, arbitrarily. We illustrate this point in the next
subsection. Nevertheless, sometimes we will be interested in the exact number of
codesets to find bounds on $C(n,t)$, which will be discussed later in this section.

\subsection{Efficient
Encoding and Decoding of Cooling Codes}
\label{sec:CoolEncoding}

\noindent
Several efficient algorithms for coding into spreads are known.
In this subsection, we describe a particularly simple and powerful
method that was originally developed by Dumer~\cite{Dumer:1989} in
the context of coding for memories with defects.
This method involves computations in the finite field $\FF_{2^{t+1})}$.
Since $t$ is a~small constant, such computations are very efficient.
In fact, for most applications of cooling codes, $t \le 7$ will
be~more~than sufficient to cool the bus wires.
Thus the proposed encoder $\cE$ and decoder $\cD$ 
work with bytes, or even nibbles, of data.

As in the previous subsection, we set $\tau = t+1$. For simplicity and w.l.o.g. we assume
that $\tau$ divides $n$ and, hence, also $k = n-\tau$.
Presented with a $k$-bit data word $\uuu$ generated by the source $\sS$,
the encoder $\cE$ first partitions $\uuu$ into $m = k/\tau$ blocks
$\uuu_1,\uuu_2,\ldots,\uuu_m$, each consisting of $\tau$ bits. We will
refer to these blocks as \emph{$\tau$-nibbles} and think of them
as elements in the finite field $\FF_{2^\tau}$. The output of the
encoder is the $n$-bit vector:
\be{encoder}
\cE(\uuu,\cS)
\, = \,
(\beta\uuu_1,\beta\uuu_2,\ldots,\beta\uuu_m,\beta)
\in \Fn
\ee
for \looseness=-1
a carefully chosen nonzero element $\beta$ of $\FF_{2^\tau}$
that~depends on both $\uuu$ and $\cS$. Note that, given $\beta$,
computing $\cE(\uuu,\cS)$ amounts to $k/\kern-1pt\tau\kern-0.5pt$
multiplications in $\FF_{2^\tau\kern-1pt}$.
\kern-0.5ptThe operation of the decoder $\cD$ is equally
simple. Given its input \mbox{$\xxx = \cE(\uuu,\cS)$},
the decoder first partitions $\xxx$ into $\tau$-nibbles
$\xxx_1,\xxx_2,\ldots,\xxx_{m+1}$~and reads off $\beta = \xxx_{m+1}$.
The decoder then computes $\beta^{-1}$ from $\beta$ in $\FF_{2^\tau}$,
and recovers the original data word $\uuu$ as follows:~~
$$
\cD(\xxx)
\, = \,
(\beta^{-1}\xxx_1,\beta^{-1}\xxx_2,\ldots,\beta^{-1}\xxx_m)
\,=\,
(\uuu_1,\uuu_2,\ldots,\uuu_m)~.
$$
Thus decoding amounts to one inversion and $k/\tau$ multiplications
in $\FF_{2^\tau}$. It remains to explain how $\beta$ is computed.~~~

Computing $\beta$ from $\uuu$ and $\cS$ is equivalent
to solving a system of $t$
linear equations in $\tau$ unknowns over $\FF_2$.
We illustrate this with the following example for the
case $\tau = 3$.
\vspace{1.00ex}

\noindent
{\bf Example.} \looseness=-1
In this example, we will work with the finite field
$
\FF_{2^3} = \{0,1,\al,\al^2,\al^3,\al^4,\al^5,\al^6\}
$,
defined by $\al^3 = 1 + \al$. Notice that $1,\al,\al^2$
is a basis for $\FF_{2^3\kern-1pt}$ 
over $\FF_2$.
Suppose that the $t = 2$ hottest wires,
presented to the encoder via the set~$\cS$,
fall in the $3$-nibbles
$\xxx_i = \beta \uuu_i$ and $\xxx_j = \beta \uuu_j$
(if both 
fall in the same $3$-nibble, the situation is
similar). Let us write
$$
\uuu_i = 
u_0 + u_1 \al + u_2 \al^2
\hspace{2.50ex}\text{and}\hspace{3.00ex}
\uuu_j = 
u'_0 + u'_1 \al + u'_2 \al^2~.
$$
Let~us~also write
$\beta = b_0 + b_1 \al + b_2 \al^2$; our goal is to determine
the $\tau = 3$ unknowns $b_0,b_1,b_2$ from the $t=2$ constraints.
Computing $(x_0,x_1,x_2) = \beta \uuu_i$ and $(x'_0,x'_1,x'_2) = \beta \uuu_j$,
we have:~~~~~~~\vspace{-0.50ex}
\begin{align}
\label{example-first}
 x_0 &=\, u_0 b_0\,+\,u_2 b_1\,+\,u_1 b_2 \\
 x_1 &=\, u_1 b_0\,+\,(u_0 {+} u_2) b_1\,+\,(u_1 {+} u_2) b_2 \\
 x_2 &=\, u_2 b_0\,+\,u_1 b_1\,+\,(u_0 {+} u_2) b_2 \\
x'_0 &=\, u'_0 b_0\,+\,u'_2 b_1\,+\,u'_1 b_2 \\
x'_1 &=\, u'_1 b_0\,+\,(u'_0{+}u'_2) b_1\,+\,(u'_1{+}u'_2) b_2 \hspace*{2.00ex}\\
\label{example-last}
x'_2 &=\, u'_2 b_0\,+\,u'_1 b_1 \,+\,(u'_0 {+} u'_2) b_2~.
\end{align}
Equating any two of $x_0,x_1,x_2,x'_0,x'_1,x'_2$ in
\eq{example-first}\,--\,\eq{example-last} to zero ge\-nerates 
a system of 
$t=2$ linear equations in the unknowns
$b_0,b_1,b_2$ with coefficients determined by
$u_0,u_1,u_2,u'_0,u'_1,u'_2$.
Such a system can be easily solved by Gaussian elimination~or,
if necessary, more efficient methods.
\hfill\raisebox{-0.50ex}{$\Box$}\vspace{1.00ex}

In general, the complexity of computing $\beta$ from $\uuu$ and $\cS$
is $O(t^3)$, which is very small for constant $t$. The overall
complexity of encoding/decoding is \emph{linear} in the number
of wires~$n$.


\subsection{Cooling Codes for Large $t$}
\label{sec:LargetCoolCodes}

The cooling codes constructed based on spreads or partial spreads
can be used only for small $t$, or more precisely
when $t+1 \leq n/2$ since subspaces of dimension
greater that $t+1$ have a nontrivial intersection if $t+1>n/2$.
Fortunately, this is probably what is usually required for the design parameters of a
thermal system. Nevertheless, we want to find efficient cooling codes also
for larger values of $t$.
Therefore, when ${t+1 > n/2}$ we have to use another construction
which generates cooling codes of a large size. In this subsection we present a construction
which forms a sunflower
whose heart of seeds (\emph{kernel}) is a linear code $\C$
and his flowers are codes obtained by the sum of $\C$
and elements of a spread. The construction of subsection~\ref{sec:CoolCodes}, which is
based on a partial spreads can be viewed also as such a sunflower, where
the kernel is a trivial linear code (only the all-zero codeword), and the flowers
are the elements of the partial spread without the all-zero codeword.

For the new construction and for some constructions which follow we will use some
basic and more sophisticated elements in coding theory which will be defined.
First, an $[n,\kappa,d]$ code $\C$ over $\Fq$ is a $\kappa$-dimensional subspace of $\Fqn$, with minimum
Hamming distance $d$. This is the only concept we need for the basic construction.
The more sophisticated elements will be needed for a generalization of the construction
which will be given later.

\begin{theorem}
\label{thm:sunF}
Let $n,t,s,r,d$ be integers, such that $r+t \leq (n+s)/2$.
If there exists a binary $[n,s,d]$ code and a binary $[n-t,r,d]$ code does not exist,
then there exists an $(n,t)$-cooling code of size
$M > 2^{n-t-r}$.
\end{theorem}
\begin{proof}
Let $n,t,s,r,d$ be integers, such that $r+t \leq (n+s)/2$,
$K$ a binary $[n,s,d]$ code, and assume that
a binary $[n-t,r,d]$ code does not exist. Let $B$ be
an $(n-s)$-dimensional subspace of $\Fn$ such that
$$
K + B = \{ a+b ~:~ a \in K , ~ b \in B \} = \Fn~.
$$
The fact that $r+t \leq (n+s)/2$ implies that $r+t-s \leq (n-s)/2$ and hence
there exists a partial $(r+t-s)$-spread of $B$ whose size $M$ is
$M_2 (n-s,r+t-s) \geq 2^{n-r-t}$. Let $V_1,V_2,\ldots , V_M$ be the subspaces
of a related partial spread and construct the following $M$ sets:
$$
C_i = (V_i + K ) \setminus K  ~~ \text{for} ~~ i \in [M]~.
$$
Since $V_i \cap V_j$ is the null space for $i \neq j$, it follows that
$(V_i + K) \cap (V_j +K) =K$ and hence $C_i \cap C_j = \varnothing$.
Hence, we can take the $C_i$'s as the codesets in a code $\C$. To prove
that $\C$ is an $(n,t)$-cooling code it remains to be shown that for
any given subset $S \in [n]$ of size $t$ and a codeset $C_i$, $1 \in [M]$, there
exists a codeword $\xxx \in C_i$, such that $\supp (\xxx) \cap S = \varnothing$.

The code $K$ is an $s$-dimensional subspace, $V_i$ in an
$(r+t-s)$-dimensional subspace, and $K \cap V_i = \{ {\bf 0} \}$
(since $V_i \subset B$ and $K \cap B = \{ {\bf 0} \}$). Hence,
$V_i + K$ is an $(r+t)$-dimensional subspace. Let
$\{\vvv_1,\vvv_2,\ldots,\vvv_{r+t}\}$ is a basis for $V_i + K$.
Let $\vvv'_1,\vvv'_2,\ldots,\vvv'_{r+t}$ denote the projections
of the basis vectors on the $t$ positions~in~$\cS$. These
$r+t$ vectors lie in a $t$-dimensional vector space --- the
projection of $\Fn$ on $\cS$. Hence, there exists an $r$-dimensional
subspace $U_i$ spanned by these $r+t$ basis vectors (which span $V_i + K$),
such that for each $\zzz \in U_i$, we have $\supp (\zzz) \cap \cS = \varnothing$.
We can remove the $t$ coordinates which only have \emph{zero} elements in $U_i$,
from all the vectors of $U_i$, to obtain an $[n-t,r,\delta]$ code $U'_i$.
Since an $[n-t,r,d]$ code does not exist
it follows that $\delta < d$. Hence, $U_i$ contains a vector $\xxx$ which is not contained in $K$.
Since $\xxx \in U_i$ and $\xxx \notin K$, it follows that $\xxx \in C_i$, and
hence $\C$ is an $(n,t)$-cooling code and the proof has been completed.
\end{proof}
Theorem~\ref{thm:sunF} can be applied in various ways. For example,
we derive the following result.
\begin{corollary}
If $n=2^m$, $m>1$, and $n/2 < t \leq (n+m-1)/2$,
then there exists an $(n,t)$-cooling code of size
$M > 2^{n-t-1}$.
\end{corollary}
\begin{proof}
Apply Theorem~\ref{thm:sunF} with the $[n,m+1,n/2]$ first order Reed-Muller
code as the kernel code $K$, i.e. $s =m+1$.
If $r=1$ then $r+t \leq 1 + (n+m-1)/2 \leq (n+m+1)/2 = (r+ s )/2$. Clearly, an $[n-t,1,n/2]$ code
does not exist since $n-t < n/2$. Hence, by Theorem~\ref{thm:sunF} we obtain an $(n,t)$-cooling
code whose size is greater than $2^{n-t-1}$.
This code is optimal by Corollary~\ref{thm5}.
\end{proof}

Theorem~\ref{thm:sunF} can be generalized by using the concept
of generalized Hamming weights which was defined by Wei~\cite{Wei91}. For this definition
we generalize the notion of support from a word to a subcode. The \emph{support}
of a subcode $\C'$ of $\C$, $\supp (\C')$ is defined as the set of coordinates
on which $\C'$ contains codewords with nonzero coordinates, i.e.
$$
\supp (\C') ~ \deff ~ \bigl\{i \in [n] ~:~ \exists (x_1,x_2,\ldots , x_n ) \in \C',~ x_i \ne 0\bigr\}~.
$$
Now, the $r$th \emph{generalized Hamming weight}, $d_r (\C)$, of $\C$ is defined as the
minimum number of coordinates in $\supp (\C')$, for an $r$-dimensional subcode $\C'$ of $\C$, i.e.
$$
d_r (\C) ~ \deff ~ \min {\bigl\{|\supp (\C') | ~:~  \C'~\text{is~an}~r\text{-dimensional~subcode~of}~\C \bigr\}}~.
$$

\begin{theorem}
\label{thm:sunG}
If $n,t,s,r,d$ are integers, such that $r+t \leq (n+s)/2$,
and the following two requirements are satisfied:
\begin{itemize}
\item[(R1)] There exists a binary $[n,s,d]$ code $K$.

\item[(R2)] The $r$th generalized Hamming weight of $K$ is larger than $n-t$.
\end{itemize}
Then there exists an $(n,t)$-cooling code of size
$M > 2^{n-t-r}$.
\end{theorem}
\begin{proof}
The proof is along the same lines as the proof of Theorem~\ref{thm:sunF},
up to the point in which an $[n-t,r,\delta]$ code $U'_i$ is obtained from $U_i$.
Since the $r$th generalized Hamming weight of $K$ is larger than $n-t$,
it follows that the $r$-dimensional subspace $U_i$
contains a vector $\xxx$ which is not contained in $K$.
Since $\xxx \in U_i$ and $\xxx \notin K$, it follows that $\xxx \in C_i$, and
hence $\C$ is an $(n,t)$-cooling code and the proof has been completed.
\end{proof}

Theorem~\ref{thm:sunG} is stronger than and generalizes Theorem~\ref{thm:sunF}. Indeed,
consider the $[n,s,d]$ code $K$ in Theorem~\ref{thm:sunF}. If
there does not exist an $[n-t,r,d]$ code, then the
$r$th generalized Hamming weight of $K$ must be larger than $n-t$.
But, if for a given $[n,s,d]$ code $K$ the
$r$th generalized Hamming weight of $K$ is larger than $n-t$,
then an $[n-t,r,d]$ code might still exists.

An independent question is whether Theorem~\ref{thm:sunG} can improve on the results
implied by Theorem~\ref{thm:sunF}, or all the results that can be
obtained from Theorem~\ref{thm:sunG} can also be obtained from Theorem~\ref{thm:sunF}.
The answer is that Theorem~\ref{thm:sunG} improves on some of the results obtained
by Theorem~\ref{thm:sunF}. This is demonstrated with the following examples.

The first interesting example is considered in the range $1 \leq t < n \leq 100$.
Consider the following $8\times 20$ parity-check matrix $H$ over $\mathbb{F}_2$,
\[
H=
\left(
 \renewcommand\arraystretch{1}
 \arraycolsep=1.2pt
 \begin{array}{rrrrrrrrrrrrrrrrrrrr}
1 & 0 & 0 & 0 & 0 & 0 & 0 & 0 & 1 & 0 & 1 & 0 & 0 & 0 & 0 & 0 & 1 & 1 & 1 & 1 \\
0 & 1 & 0 & 0 & 0 & 0 & 0 & 0 & 1 & 1 & 0 & 0 & 0 & 0 & 1 & 1 & 0 & 0 & 1 & 1 \\
0 & 0 & 1 & 0 & 0 & 0 & 0 & 0 & 1 & 1 & 1 & 1 & 0 & 0 & 0 & 1 & 0 & 1 & 0 & 1 \\
0 & 0 & 0 & 1 & 0 & 0 & 0 & 0 & 0 & 0 & 1 & 1 & 0 & 1 & 1 & 1 & 1 & 0 & 0 & 1 \\
0 & 0 & 0 & 0 & 1 & 0 & 0 & 0 & 0 & 1 & 0 & 1 & 0 & 1 & 1 & 0 & 1 & 1 & 0 & 1 \\
0 & 0 & 0 & 0 & 0 & 1 & 0 & 0 & 0 & 1 & 1 & 0 & 0 & 1 & 0 & 1 & 0 & 1 & 1 & 0 \\
0 & 0 & 0 & 0 & 0 & 0 & 1 & 0 & 1 & 1 & 0 & 1 & 1 & 0 & 0 & 0 & 0 & 1 & 1 & 0 \\
0 & 0 & 0 & 0 & 0 & 0 & 0 & 1 & 1 & 0 & 1 & 1 & 1 & 0 & 1 & 0 & 1 & 1 & 0 & 0
\end{array}\right).
\]

One can verify that every six columns have rank at least 5.
Therefore, there exists an $[n,n-8,d]$ linear code with $d_2\ge 7$ for $n\in \{18,19,20\}$.
Thus, by Theorem~\ref{thm:sunG}, there exists an $(n,n-6)$ cooling code of size
greater than $2^4=16$ for $n\in \{18,19,20\}$.

On other hand, suppose that Theorem~\ref{thm:sunF} is used
to obtain an $(n,n-6)$-cooling code of size greater than $2^4$ for $n\in \{18,19,20\}$.
It is required by Theorem~\ref{thm:sunF} that $t+r\le (n+s)/2$,
where $t=n-6$ and $r=2$, and hence $s\ge n-8$.
Using the online codetables.de~\cite{Grassl}, the best code with
dimension $n-8$ for length $n\in\{18,19,20\}$ has minimum Hamming distance $4$.
But, there exists a $[6,2,4]$ code, and hence Theorem~\ref{thm:sunF} cannot be applied.

Another example is derived from Feng et al.~\cite{FTW92} who computed
the following bounds for the generalized Hamming weights for cyclic codes of length 255.
\begin{enumerate}[(i)]
\item The primitive double-error-correcting BCH $[255,239,5]$ code has $d_9\ge 18$.
If $t=238$, then by Theorem~\ref{thm:sunG}, there exists a $(255,238)$-cooling code of size greater than $2^8$.

On other hand, suppose that Theorem~\ref{thm:sunF} is used
to obtain a $(255,238)$-cooling code of size greater than $2^8$.
It is required by Theorem~\ref{thm:sunF} that $t+r\le (n+s)/2$,
where $n=255$, $t=238$ and $r=9$, and hence $s\ge 239$.
Using the online codetables.de~\cite{Grassl}, the best code with
dimension $239$ for length $255$ has minimum Hamming distance $5$.
But, there exists a $[17,9,5]$ code, and hence Theorem~\ref{thm:sunF} cannot apply.

\item The reversible cyclic double-error-correcting BCH $[255,238,5]$ code has $d_9\ge 19$~\cite{FTW92}.
If $t=237$, then by Theorem~\ref{thm:sunG},
there exists a $(255,237)$-cooling code with size greater than $2^9$.

On other hand, suppose that  Theorem~\ref{thm:sunF} is used
to obtain a $(255,237)$-cooling code of size greater than $2^9$.
It is required by Theorem~\ref{thm:sunF} that $237+9\le (255+s)/2$ and hence $s\ge 237$.
Using the online codetables.de~\cite{Grassl},
the best code with dimension $237$ has minimum Hamming distance $6$.
But, there exists an $[18,9,6]$ code, and hence Theorem~\ref{thm:sunF} cannot be applied.
\end{enumerate}

\subsection{Best Lower Bounds on $C(n,t)$}
\label{sec:BestLB}

Now, we will summarize our constructions by
providing our best lower bounds on $C(n,t)$ in general
and when $1 \leq t < n \leq 100$ in particular.
\[
C(n,t)~
\begin{cases}
=2^{n-t}, & \mbox{if $t=1$ or $t=n-1$,}\\
> 2^{n-t-1}, & \mbox{if $t+1\le n/2$,}\\
> 2^{n-t-2}, & \mbox{if $(n,t)\in\{(18,12),(19,13),(20,14)\}$,}\\
> M_2(n-s,r+t-s), & \mbox{if $n$, $t$, $s$, and $r$ appears in TABLE I in the Appendix,}\\
& \mbox{and $(n,t)\notin\{(18,12),(19,13),(20,14)\}$,}\\
\ge n-t+1, & \mbox{otherwise.}
\end{cases}
\]

\vspace{2.00ex}
\section{Low-Power Cooling Codes}
\label{sec:LPCC}
\vspace{0.50ex}

\noindent\looseness=-1
In this section, we present coding schemes that satisfy Properties
$\A(\kern-0.5ptt\kern-0.5pt)$ and $\B(\kern-0.5ptw\kern-0.5pt)$
simultaneously in every transmission\mbox{.\kern-1pt}
The corresponding codes are called {\dfn $(n,t,w)$-low-power cooling codes}
(or $(n,t,w)$-LPC codes for short). We suggest two
types of constructions. The first type is based on decomposition
of the complete hypergraph into disjoint perfect matchings.
The second type is based on cooling codes over GF($q$),
dual codes of $[n,\kappa,t+1]$ codes, MDS codes, spreads,
$J^+ (r,\omega)$, and concatenation codes. Finally, we will show that
also the sunflower construction can be used to obtain $(n,t,w)$-LPC codes.

As before, we assume that the coding schemes
constructed in what follows are augmented by differential encoding.
Since the codes are also $(n,\kern-1ptt)$-cooling codes,
they~conform to \Dref{cooling-def}. Thus a code $\C$ is
a collection of codesets $C_1,C_2,\ldots,C_M$, which are disjoint
subsets of $\Fn$. In~order to satisfy Property\,$\B(w)$, the codesets
of the code $\C$ must satisfy that
\be{C_i-embed}
C_1,C_2,\ldots,C_M \,\subset\, J^+ (n,w)~.
\ee
As shown in Section\,\ref{sec:Differential}, this guarantees that
the total number of state transitions on the $n$ bus wires
is at most~$w$.

\subsection{Decomposition of the Complete Hypergraph}
\label{sec:Resolutions}

The first construction applies the well known Baranyai's theorem~\cite[p. 536]{vLWi92}.
The theorem can be stated in terms of set systems, but it is
more known in the context of the decomposition of complete hypergraph on $n$ vertices.
The hyperedges of the complete hypergraph consist of all subsets of $w$ vertices.
In other words, the set of vertices is $[n]$ and the set of edges consists of all
$w$-subsets of $[n]$.
If $w$ divides $n$ then the Baranyai's theorem asserts that the hyperedges of this complete hypergraph
can be decomposed into pairwise disjoint perfect matchings, where each matching consists of disjoint
hyperedges (two hyperedges do not contain the same vertex) and each vertex is contained in exactly one
hyperedge of the matching. Clearly, each perfect matching contains $n/w$ hyperedges,
and such a decomposition contains $\frac{w}{n} \binom{n}{w} = \binom{n-1}{w-1}$
disjoint matchings. How can such a decomposition can be used to construct $(n,t,w)$-LPC codes?
The answer lies on a right choice of $n$, e.g., if $n=w(t+1+\epsilon)$, where $\epsilon$ is a nonnegative
integer, we can use each matching in such a decomposition as a codeset.
To prove that such a code is an $(n,t,w)$-LPC code, we have to prove that it
satisfies Property $\A(\kern-0.5ptt\kern-0.5pt)$. Given a set $\cS \subset [n]$ of size $t$ with
the numbers of the $t$ hottest wires,
and a perfect matching $C$, the elements of $\cS$ are contained in at most $t$ hyperedges. But, since a
matching contains $t+1+ \epsilon \geq t+1$ hyperedges, it follows that there is at least one hyperedge
in the matching which does not contain any element from~$\cS$.
Such a hyperedge is the codeword $\xxx \in  C$ for which $\supp (\xxx) \cap \cS = \varnothing$.
The most effective sets of parameters on which this construction can be
applied are when $\epsilon =0$, i.e. $(w(t+1),t,w)$-LPC code.
Clearly, we can add codesets with codewords whose weight is less than~$w$. When~$n$ is large, the contribution
of such codesets is minor. When~$n$ is small such a contribution can be important. For example, if $n=12$, $w=3$,
and $t=3$, then there are 55 codesets if we restrict ourself only to codewords of weight 3. If we use
all words of weight at most 3, then we can have a code with 81 codesets. The advantage is reduced for larger $n$ --- for example,
if $n=21$, $w=3$ and $t=6$. There are 190 codesets if only codewords of weight 3 are used. If there are codewords
of weight at most 3, then we can have a code with 224 codesets, which is less dramatic
improvements compared to the previous example.

What about encoding and decoding of this code? Here lies the big disadvantage of this method.
Unless the parameters are relatively small there is
no known efficient encoding and decoding algorithms.
Anyway, in practice, usually small parameters are used and for many of these sets
of parameters this method (code) is probably the most effective one.

\subsection{Constructions Based on Cooling Codes over $\Fq$}
\label{sec:CoolCodesFq}

The low-power codes and the cooling codes are all binary codes, as a
codeword indicates which transitions should be made on the bus-wires
during the transmission. Hence, there is
no use for non-binary codes as thermal codes. Nevertheless, non-binary codes might be useful
in constructions of binary codes as will be demonstrated in this subsection.

\begin{definition}
\label{coolingFq-def}
\,For positive integers $n$ and $\kern1pt t\kern-1pt <\kern-1pt n$, and
for a prime power $q$, an
{\dfn $(n,\kern-1pt t)_q$-cooling code} $\C$ of size $M$ is
defined as a set $\{C_1,C_2,\ldots,C_M\}$, where $C_1,C_2,\ldots,C_M$
are disjoint subsets of\/ $\Fqn$ satisfying the following property:
for any set\/ $\cS \subset [n]$ of size $|\cS| = t$ and for~all
$i \,{\in}\,\kern-1pt [M]$,
there exists a codeword\/ $\xxx \,{\in}\, C_i\kern-1pt$
with\/ $\supp(\xxx) \cap \cS\kern-1pt = \varnothing$.\vspace*{0.50ex}
\end{definition}

Similarly to Theorem~\ref{thm:spread_cool} we can prove that

\begin{lemma}
\label{thm:spread_cool_q}
Let\/ $V_1,\kern-1ptV_2,\ldots,V_M$ be a partial $(t{+}1)$-spread of\/ $\Fqn$, and
define the code\kern1.5pt\ $\C = \{V^*_1,V^*_2,\ldots,V^*_M\}$,
where\/
$V^*_i \kern-1.5pt=\kern-1pt V_i \kern1.5pt{\setminus} \{\zero\}$
for all\/ $i$, $i = 1,2,\ldots,M$.
Then\/ $\C$ is an $(n,t)_q$-cooling code of size $M \ge q^{n-t-1}$.
\end{lemma}

Theorem~\ref{thm:spread_cool_q} can be applied when $t+1 \leq n/2$. When
$t+1 > n/2$ we can generalize Theorem~\ref{thm:sunF} to
obtain $(n,t)_q$-cooling codes. The generalization is straightforward
and hence it will be omitted.

A third construction, which is not a generalization of previous constructions,
for $(n,t)_q$-cooling codes is based on the cosets
of a dual code for a linear code over $\Fq$, whose minimum Hamming distance is at least $t+1$.
We start with the related definitions and properties. For more information and proofs of
the claims given, the reader can consult with~\cite{MWS}.

Each $[n,\kappa,d]$ code $\C$ over $\Fq$ induces a partition of $\Fqn$, where
$\C_{\zzz}$, $\zzz \in \Fqn$, is a part in this partition if
$$
\C_{\zzz} = \bigl\{ \zzz + \ccc ~:~ \ccc \in \C  \bigr\}~.
$$
Each such part is called a \emph{coset of $\C$} and this partition contains $q^{n-\kappa}$
pairwise disjoint cosets.

Each $[n,\kappa,d]$ code $\C$ over $\Fq$ has a \emph{dual code} $\C^\perp$, which is the
dual subspace of $\C$. $\C^\perp$ is an $[n,n-\kappa,d']$ code.

A $\lambda \cdot q^t \times n$ matrix $\cA$ with elements from $\Fq$ is called an \emph{orthogonal array}
with \emph{strength} $t$ if each $t$-tuple over $\Fq$ appears exactly $\lambda$ times
in each projection on any $t$ columns of~$\cA$. The dual code $\C^\perp$ of an $[n,\kappa,d]$ code $\C$,
is an orthogonal array of strength $d-1$.

Finally, the \emph{Singleton bound} for an $[n,\kappa,d]$ code over $\Fq$ asserts that $d \leq n-\kappa +1$.
A code which attains this bound, with equality,
is called a \emph{maximum distance separable} code (an \emph{MDS} code in short).

\begin{lemma}
\label{lem:dual_cool}
If there exists an $[n,\kappa,t+1]$ code over $\Fq$, then there exists
an $(n,t)_q$-cooling code of size $q^\kappa$.
\end{lemma}
\begin{proof}
Let $\C$ be an $[n,\kappa,t+1]$ code over $\Fq$ and $\C^\perp$ be its dual code.
$\C^\perp$ is an orthogonal array of strength $t$,
i.e., each $t$-tuple over $\Fq$ appears the same number of times, in each projection
of $t$ columns of $\C^\perp$. Since the size of~$\C^\perp$ is $q^{n-\kappa}$, it follows
that each such $t$-tuple (including the all-zero $n$-tuple) appears $q^{n-\kappa -t}$ times
in each projection (note that $q^{n-\kappa-t} >0$ by the Singleton bound). Since a coset of $\C^\perp$ is
formed by adding a fixed vector of length $n$ to all the codewords of $\C^\perp$, it follows that
each coset is also an orthogonal array of strength $t$. Therefore, the code $\C^\perp$ and its $q^\kappa -1$ cosets,
can be taken as $q^\kappa$ codesets, to
form an $(n,t)_q$-cooling code of size $q^\kappa$.
\end{proof}

Lemma~\ref{lem:dual_cool} has some interesting consequences. First, one might asks why
the construction implied by Lemma~\ref{lem:dual_cool} was not given in Section~\ref{spreads}?
The answer is very simple. The binary codes obtained by this construction are not good enough
as the codes presented in the constructions of Section~\ref{spreads}.

A second and more important observation from Lemma~\ref{lem:dual_cool}
is a construction of a large $(n,t)_q$-cooling code if we use
an MDS code as the original code $\C$. Recall that an $[n,\kappa ,d]$ code is an MDS code
if and only if $d=n-\kappa +1$. Moreover, the dual code $\C^\perp$ of an MDS code is also
an MDS code. Finally, there is a well known conjecture about the range in which
MDS codes can exist and there are MDS codes for all parameters in this range.

\begin{conjecture}
\label{conj:MDS}
If $d \geq 3$, then there exists an $[n,\kappa ,d]$ MDS code over $\Fq$
if and only if $n \leq q+1$ for all $q$ and $2 \leq \kappa \leq q-1$, except
when $q$ is even and $\kappa \in \{ 3,q-1 \}$, in which case $n \leq q+2$.
\end{conjecture}

\begin{theorem}
If $d \geq 3$, then there exists an $[n,\kappa ,d]$ MDS code over $\Fq$
if $n \leq q+1$ for all $q$ and $2 \leq \kappa \leq q-1$, except
when $q$ is even and $\kappa \in \{ 3,q-1 \}$, in which case $n \leq q+2$..
\end{theorem}

\begin{corollary}
\label{cor:MDS_cool}
If $n \leq q+1$, then there exist an $(n,t)_q$-cooling code of size $q^{n-t}$.
\end{corollary}

We continue to present a construction which transform an $(n,t)_q$-cooling code
into an $(ns,t,nw)$-LPC code. This construction will be called a concatenation
construction since it perform concatenations of the elements in $J^+(s,w)$
implied by an $(n,t)_q$-cooling code.

\begin{theorem}
\label{thm:concatenate}
If $q \leq \sum_{i=0}^w \binom{s}{i}$ and there exists an $(m,t)_q$-cooling
code of size $M$, then there exists an $(ms,t,mw)$-LPC code of size $M$.
\end{theorem}
\begin{proof}
Let $\psi$ be an injection from $\Fq$ to $J^+(s,w)$ such that $\psi (0) = {\bf 0}$.
Let $\Psi$ be an injection from $\Fqm$ to ${\Fq}^{ms}$ defined by
$\Psi ( x_1 , x_2 , \ldots , x_m ) = ( \psi (x_1 ) , \psi (x_2 ),\ldots , \psi (x_m) )$.

Let  $\cC = \{ C_1 , C_2 , \ldots , C_M \}$ be an $(m,t)_q$-cooling code
of size $M$. Let $\cD = \{ D_1 , D_2 ,\ldots ,D_M \}$ the image of $\cC$
under $\Psi$, i.e. $D_i = \{ \Psi (\xxx) ~:~ \xxx \in C_i \}$, for $1 \leq i \leq M$.
We claim that $\cD$ is an $(ms,t,mw)$-LPC code.

The length of codewords in the codesets of $\cD$ is clearly $ms$ and since
$\psi (x_i) \leq w$ for each $x_i \in \Fq$ it follows that the weight of a codeword, in a codeset,
is at most $mw$. It remains to show that for any given set $\cS \subset [ms]$ of size~$t$, and a codeset $D_i$,
there exists a codeword $\yyy$ in $D_i$ such that ${\supp (\yyy) \cap \cS = \varnothing}$.
Since $\cC$ is an $(m,t)_q$-cooling code, it follows that for any set $\cS' \subset [m]$ of size $t$,
the codeset~$C_i$ contains a codeword $\xxx = (x_1 , x_2 ,\ldots , x_m )$ such that
${\supp (\xxx) \cap \cS' = \varnothing}$. If we partition the set of $ms$ coordinates of the codewords
in $\cD$ into $m$ consecutive sets $\cL_1, \cL_2 ,\ldots , \cL_m$, where $\cL_1$ contains the
first $s$ coordinates, $\cL_2$ the next $s$ coordinates, and so on, then the elements in $\cS$
are contained in ${t' \leq t}$ of these sets, say, $I_{j_1}, I_{j_2} , \ldots , I_{j_{t'}}$.
If we define $\cS'' = \{ j_1 , j_2 ,\ldots , j_{t'} \}$, then $\cS'' \subset [m]$ is a set
of size ${t' \leq t}$. Hence, since $\cC$ is an $(m,t)_q$-cooling code, it follows that there
exists a codeword $\zzz \in C_i$ such that $\supp (\zzz) \cap \cS'' = \varnothing$. Since
$\psi (0) = {\bf 0}$, it follows from the definition of $\Psi$ that
$\Psi (\zzz) \in D_i$ and $\supp (\Psi (\zzz)) \cap \cS = \varnothing$, i.e. we
can take $\yyy = \Psi (\zzz)$ to complete the proof.
\end{proof}

Theorem~\ref{thm:concatenate} can be combined with Theorem~\ref{thm:spread_cool_q}
and Corollary~\ref{cor:MDS_cool} to obtain the following two results.
\begin{corollary}
\label{cor:spread_t,w}
If $q\le \sum_{i=0}^{w'} \binom{s}{i}$ and $t+1\le m/2$,
then there exists an $(ms,t,mw')$-LPC code of size $M> q^{m-t-1}$.
\end{corollary}
\begin{corollary}
\label{cor:MDS_t,w}
If $q\le \sum_{i=0}^{w'} \binom{s}{i}$ and $m\le q+1$,
then there exists an $(ms,t,mw')$-LPC code of size $q^{m-t}$.
\end{corollary}

\subsection{Constructions Based on Sunflowers}
\label{sec:CoolCodesSun}

We examine the code obtained by the sunflower construction of Theorem~\ref{thm:sunF}.
Note, that in the proof of Theorem~\ref{thm:sunF} the codeword $\zzz$
of the related codeset used to show that $\supp (\zzz) \cap \cS = \varnothing$
has weight less than $d$. Therefore, we can remove from the codesets all codewords
of weight greater than $d-1$.
Thus, the sunflower construction yields an
$(n,t,d-1)$-LPC code if the related constraints are satisfied as follows.

\begin{corollary}
\label{cor:sunLPCC}
Let $n,t,s,r,d$ be integers, such that $r+t \leq (n+s)/2$.
If there exists an $[n,s,d]$ code and an $[n-t,r,d]$ linear code does not exist,
then there exists an $(n,t,d-1)$-LPC code of size $M > 2^{n-t-r}$.
\end{corollary}

Note, that a similar theorem can be obtained by using the sunflower of Theorem~\ref{thm:sunG}
which is based on the generalized Hamming weights.


\vspace{2.00ex}
\section{Error-Correcting Thermal Codes}
\label{sec:ECC}
\vspace{0.50ex}

In this section, we construct thermal codes that satisfy
Property ${\bf C}(e)$ and simultaneously
Property ${\bf A}(t)$ or Property~${\bf B}(w)$.
The idea will be to modify and generalize the constructions which were given in the previous
sections, by adding error-correction into the constructions. We will present constructions in the same
order in which they were presented so far in this work. First, we discuss $(n,w,e)$-LPEC codes,
for both nonadaptive and adaptive schemes. We
continue with $(n,t,e)$-ECC codes, and conclude with $(n,t,w,e)$-LPECC codes.

\subsection{Adaptive and Nonadaptive Low-Power Codes}
\label{sec:LPECC}

For codes which satisfy simultaneously
Properties ${\bf B}(w)$ and~${\bf C}(e)$,
we consider first nonadaptive codes along the lines discussed in Section~\ref{anticodes}.
The type of codes which were considered in Section~\ref{anticodes} are anticodes
and in particular equireplicate anticodes. For this purpose we will use again set
systems and in particular a family of block design called Steiner systems.
A \emph{Steiner system} $S(r,w,n)$ is a collection $\cB$
of $w$-subsets (called \emph{blocks}) from the $n$-set $[n]$ such that each $r$-subset of~$[n]$ is
contained in exactly one block of $\cB$. The blocks
of such a set system can be translated into a binary code~$\C$ of length $n$
and constant weight $w$ for the codewords. It is easy to verify that the code~$\C$
has minimum Hamming distance $2(w-r+1)$, i.e. the code $\C$ can correct
any $w-r$ errors and can detect any $w-r+1$ errors. The diameter of the related code
is at most $2w$, and it is less than $2w$ if an only if there is a nonempty intersection
between any two codewords. If we are not restricted to equireplicated then we can use
constant weight codes instead of Steiner systems. Constant weight codes have many applications
an hence they were intensively investigated throughout the years, e.g.~\cite{BSSS} and
references therein. A constant weight code can be equireplicate, but it does not
have to be such a code. In this subsection we consider only equireplicate codes. A Steiner system
is equireplicate, so such systems will be the basis of our construction. Information on the known
Steiner systems can be found in the main textbooks on block designs, e.g.~\cite{BJL99}.
There are well-known necessary conditions for the existence of a Steiner system $S(r,w,n)$.
For each $i$, $0 \leq i \leq r$ the number $\binom{n-i}{r-i} / \binom{w-i}{r-i}$ is an integer.
It was recently proved in~\cite{GKLO16,Kee14} that for any given $0 <r<w$, these necessary conditions
are also sufficient, except for a finite number of cases.

As was discussed in Section~\ref{anticodes}, to have at most $w$ transitions on the
bus wires in the nonadaptive scheme, the minimum distance of our code must be at most $w$,
and each codeword must be of weight at most $w/2$ if $w$ is even, and $(w+1)/2$ if $w$ is odd.
For simplicity we will assume that $w$ is even. Now, assume that we want our code to have
Property ${\bf C}(e)$. For this purpose we need a Steiner system $S(w/2-e,w/2,n)$. If $2e \leq w/2$
then we can add codewords of weight $w/2 -2e$. The number of possible such codewords is negligible
compared to the number of codewords with weight $w/2$ and
hence we omit these possible codewords in our discussion.
Similar codes can be constructed for other parameters. Assume there exists
a system set $\cB$ which is a Steiner system ${S(w/2 - e,w/2,n+1)}$ on the
point set $[n+1]$. Consider the set
$$
\cB' \deff \bigl\{ X ~:~ X \in \cB,~ n+1 \notin X \bigr\}
\bigcup  \bigl\{ X \setminus \{n+1\} ~:~ X \in \cB,~ n+1 \in X \bigr\}~.
$$
It is easy to verify that $\cB'$ is an $(n,w,e)$-LPEC code. Moreover,
this code is also equireplicate and can be used for a nonadaptive scheme.

In contrast to $(n,w)$-LP codes, where an optimal anticode $\C$
with codewords of weight $w/2$ implies that $\C$ has diameter~$w$,
the situation when we consider also Property ${\bf C}(e)$ can be slightly different.
Clearly, the diameter of the code should be~$w$, but the weight of a codeword
can be larger than $w/2$. For example, if $w/2=q$,
$n=q^2 +q+1$, and $q$ is a prime power, then an optimal such system consists of
$q^2 +q+1$ blocks which form a projective plane of order $q$~\cite[p. 224]{vLWi92}.
Recall that such a structure was considered also in subsection~\ref{sec:LP_optimal}.
The blocks of such a projective plane form a Steiner system
$S(2,q+1,q^2+q+1)$, where any two distinct blocks intersect
in exactly one point and hence the diameter of the related code is $2q=w$ as required.
The related code can correct any $q-1$ errors and can detect any $q$ errors.

Finally, for parameters where related Steiner systems do not exist or no efficient construction for such systems
is known, we can use similar constant weight codes based on the rich literature of such codes.

For adaptive $(n,w)$-LPEC codes we use as in most of our exposition
the differential encoding method. The codes which are used are Steiner
systems as in the nonadaptive case. The only difference is that the
weight of the codewords will be at most $w$ and not $w/2$ or $(w+2)/2$
as in the nonadaptive case. The number of errors which are corrected is
defined by the Steiner system as we discussed in this subsection.

\subsection{Error-Correcting Cooling Codes}
\label{sec:ErrorCorrectCoolCodes}

In this subsection we will adapt the construction based on spreads
and the sunflower construction, given in Section~\ref{spreads},
to form codes which satisfy Property ${\bf A}(t)$ and Property ${\bf C}(e)$.
The idea is start with a binary $[n,\kappa,2e+1]$ code $\C$ which corrects $e$ errors.
In $\C$ there exists at least one set $\cS$ of~$\kappa$~coordinates whose projection
on $\C$ is~$\Fkap$. This set of coordinates is called a \emph{systematic} set
of coordinates. On this set of $\cS$ coordinates,
we either apply the spread construction or the sunflower construction to obtain
a code which satisfies simultaneously Property ${\bf A}(t)$ and Property~${\bf C}(e)$.

For the construction which is based on a partial $(t+1)$-spread, we start with
our favorite binary $[n,\kappa,2e+1]$ code~$\C$, where $\kappa \geq 2(t+1)$.
In addition, we
take a partial $(t+1)$-spread (or a $(t+1)$-spread if $t+1$ divides $\kappa$)
of $\Fkap$. Since $\C$ has dimension $\kappa$,
there exists at least one set of $\kappa$ coordinates whose projection on $\C$ spans $\Fkap$.
The partial $(t+1)$-spread can be formed on these coordinates. The
codewords of $\C$ are partitioned
into codesets related to this partial spread.
Given any $t$ coordinates, each codeset has at least
one codeword with \emph{zeroes} in these $t$ coordinates
since the partial spread has dimension $t+1$. This is proved exactly as in the proof of
Theorem~\ref{thm:spread_cool}.
Moreover, the code can correct at least $e$ errors since the codesets are disjoint
and all the codewords in the codesets
are contained in the code~$\C$. Thus, we have constructed an $(n,t,e)$-ECC code.
We summary this construction with the following theorem.

\begin{theorem}
If there exists a binary $[n,\kappa,2e+1]$ code and
$\kappa \geq 2(t+1)$, then there exists an $(n,t,e)$-ECC code of size $M > 2^{\kappa -t-1}$.
\end{theorem}

The sunflower construction is adapted in a similar way to obtain an $(n,t,e)$-ECC code.
We start with our favorite binary $[n,\kappa,2e+1]$ code $\C$ and apply the
sunflower construction on $\kappa$ systematic coordinates in $\C$.
Similarly to Theorem~\ref{thm:sunF} we have the following theorem

\begin{theorem}
\label{thm:sunFECC}
Let $n$, $t$, $s$, $r$, $\kappa$, $d$, be integers such that $r+t \leq (\kappa + s)/2$
and there exists a binary $[n,\kappa,2e+1]$ code. If there exists a $[\kappa , s ,d]$ code
and a binary $[\kappa -t ,r,d]$ code does not exist, then there exists an $(n,t,e)$-ECC code of
size $M > 2^{\kappa -t-r}$.
\end{theorem}

Similarly, to Theorem~\ref{thm:sunFECC} we can adapt the construction of Theorem~\ref{thm:sunG}
to obtain an $(n,t,e)$-ECC code.

\subsection{Constructions of Low-Power Error-Correcting Cooling Codes}
\label{sec:LPECCcool}

In this subsection we consider codes which satisfy all Properties ${\bf A}(t)$,
${\bf B}(w)$, and ${\bf C}(e)$ simultaneously. The related codes are $(n,t,w,e)$-LPECC codes.
We suggest three methods to construct such codes, which reflect the three methods
in Section~\ref{sec:LPCC}, where only Properties ${\bf A}(t)$ and
${\bf B}(w)$, without Property ${\bf C}(e)$, were considered.
The first method, generalizes the decomposition of the complete hypergraph
as described in subsection~\ref{sec:Resolutions}, by considering
resolvable Steiner systems. The second method is based on dual codes and concatenation
modifies the construction in subsection~\ref{sec:CoolCodesFq}. The last construction
is based on the sunflower construction similarly to subsection~\ref{sec:CoolCodesSun}.

Our first method is a generalization for the decomposition of the complete hypergraph
into pairwise disjoint perfect matchings. For this generalization we consider the
hyperedges as blocks in a system set, or more precisely in a block design.
The related concepts in block design are \emph{resolution} and \emph{parallel classes}.
A block design (set system) is said to be \emph{resolvable} if the blocks of the
design can be partitioned into pairwise disjoint sets, called \emph{parallel classes},
where each class forms a partition of the point sets into pairwise disjoint blocks.
The whole process is called \emph{resolution}.
By this definition, the decomposition of the complete hypergraph on $n$ vertices
and hyperedges of size $w$ is a resolution of all $w$-subsets of $[n]$.
To have low-power error-correcting cooling codes we will use resolutions of Steiner system.
Recall that a Steiner system $S(r,w,n)$ is a binary code with minimum Hamming distance
$2(w-r+1)$ and hence it can correct any $w-r$ errors and can detect any $w-r+1$ errors.
Since we are interested in an $(n,t,w,e)$-LPECC code, we should start with a Steiner system
$S(w-e,w,w(t+1))$ and partition it into pairwise disjoint Steiner systems $S(1,w,n)$ (which are the parallel classes).
Such partitions of Steiner systems are known in several cases given as follows:

\begin{enumerate}
\item If $n \equiv 3 (\text{mod}~6)$ then there exists a resolvable $S(2,3,n)$~\cite{HRCW72}.

\item If $n \equiv 4 (\text{mod}~12)$ then there exists a resolvable $S(2,4,n)$~\cite{HRCW72}.

\item If $n \equiv w (\text{mod}~w(w-1))$ then for sufficiently large $n$
there exist a resolvable $S(2,w,n)$~\cite{RCW73}.

\item If $n \equiv 4~ \text{or}~8 (\text{mod}~12)$ then there exists a resolvable
$S(3,4,n)$. This was proved in~\cite{Har87} for all $n$, except for 23 cases
which were completed in~\cite{JiZh05}.

\item If $q$ is a prime power then there exists a resolvable $S(2,q,q^2)$ derived
from affine plane (which can be generated from a projective plane of order $q$ which
is equivalent to $S(2,q+1,q^2+q+1)$).
\end{enumerate}

These resolvable Steiner systems imply the existence of the following
$(n,t,w,e)$-LPECC codes:

\begin{theorem}
If $n=3(t+1)$, where $t$ is even, then there exists an $(n,t,3,1)$-LPECC code
of size larger than $\frac{n-1}{2}$.
\end{theorem}
\begin{proof}
If $t$ is even, then $n= 3(t+1) \equiv 3(\text{mod}~6)$, and hence there exists a resolvable
Steiner triple system of order $n$ whose size is $\frac{n(n-1)}{2 \cdot 3}$,
with $\frac{n-1}{2}$ parallel classes, each one of size $\frac{n}{3}$.
The addition of the all-zero vector of length $n$ in a new codeset enlarge the size of the code.
\end{proof}

Similarly, we have
\begin{theorem}
If $n=4(t+1) \equiv 4~ \text{or}~8(\text{mod}~12)$, then there exists an $(n,t,4,1)$-LPECC code
of size larger than $\frac{(n-1)(n-2)}{6}$.
\end{theorem}

\begin{theorem}
If $n=4(t+1) \equiv 4(\text{mod}~12)$, then there exists an $(n,t,4,2)$-LPECC code
of size $\frac{(n-1)(n-2)}{12}$.
\end{theorem}

\begin{theorem}
If $n = w (t+1) \equiv w (\text{mod}~w(w-1))$, then for sufficiently large $n$
there exists an $(n,t,w,w-2)$-LPECC code of size $\frac{(n-1)(n-2)}{w(w-1)}$.
\end{theorem}

\begin{theorem}
If $q$ is a prime power, then
there exists a $(q^2,q-1,q,q-2)$-LPECC code of size $q^2$.
\end{theorem}

The constructions derived, in subsection~\ref{sec:CoolCodesFq},
from codes over $\Fq$ are based on the codes
constructed in Corollaries~\ref{cor:spread_t,w} and~\ref{cor:MDS_t,w}.
These corollaries are derived from the concatenation construction presented in
Theorem~\ref{thm:concatenate}. To adapt these results to form
$(n,t,w,e)$-LPECC codes, we use constructions of two types.
The first one starts with an error-correcting code $\C$ over $\Fq$,
partition of the codewords of $\C$ by using a partial spread, and apply the concatenation construction
on the codesets derived from the partial spread. The second type of construction is based
the $(n,t,w)$-LPC codes derived from dual codes of MDS codes (see Corollary~\ref{cor:MDS_t,w}).

For the first construction we start with an $[m,\kappa,2e+1]$ code $\C$ over $\Fq$, such that
$\kappa \geq 2(t+1)$. In addition, we
take a partial $(t+1)$-spread
of $\Fqkap$. Since $\C$ has dimension $\kappa$, it follows that
there exists at least one set of $\kappa$ systematic coordinates whose projection spans $\Fqkap$.
The partial $(t+1)$-spread can be formed on these $\kappa$ systematic coordinates. The
codewords of $\C$ can be partitioned
into codesets related to this partial spread and form a new code $\C'$.
Given any set $\cS$ of size $t$, each codeset has at least
one codeword with \emph{zeroes} in the $t$ coordinates of $\cS$,
since the subspaces of the partial spread have
dimension $t+1$. The proof is exactly as the proof of
Theorem~\ref{thm:spread_cool}.
Moreover, the code $\C'$ can correct $e$ errors, since the codesets of $\C'$ are disjoint
and all the codewords in the codesets
are contained in the code $\C$ which can correct any $e$ errors. Finally, let $q\le \sum_{i=0}^{w'} \binom{s}{i}$,
and we use the concatenation construction given in Theorem~\ref{thm:concatenate} and Corollary~\ref{cor:spread_t,w}.
Thus, we obtain an $(ms,t,mw',e)$-LPECC code.
The size of the code $\C'$ depends on the largest dimension of an $[m,\kappa,2e+1]$ code $\C$,
subject to the requirement that $\kappa \geq 2(t+1)$. The resulting code $\C'$
will have size $q^{\kappa -t-1}$.

The second construction is an immediate consequence of Corollary~\ref{cor:MDS_t,w}.
The minimum Hamming distance of the low-power cooling code obtained
in Corollary~\ref{cor:MDS_t,w} is the same as the minimum Hamming distance of the related MDS code
of length $m$, $m \leq q+1$.
Therefore, the size of the code is $q^{m-t}$ and its minimum Hamming distance $t+1$.

\begin{corollary}
\label{cor:MDS_t,w,e}
If $q\le \sum_{i=0}^{w'} \binom{s}{i}$ and $m\le q+1$,
then there exists an $(ms,t,mw',\lfloor t/2 \rfloor)$-LPECC code of size $q^{m-t}$.
\end{corollary}

Finally, we want to adapt the sunflower construction
to form an $(n,t,w,e)$-LPECC code.
The idea is to use the arguments of Corollary~\ref{cor:sunLPCC},
in the construction implied by Theorem~\ref{thm:sunFECC}, which
yields the following result

\begin{corollary}
\label{cor:sunFECCt,w}
Let $n$, $t$, $s$, $r$, $\kappa$, $d$, be integers such that $r+t \leq (\kappa + s)/2$
and there exists a binary $[n,\kappa,2e+1]$ code. If there exists a $[\kappa , s ,d]$ code
and a binary $[\kappa -t ,r,d]$ code does not exist, then there exists an $(n,t,d-1,e)$-LPECC code of
size $M > 2^{\kappa -t-r}$.
\end{corollary}

Similarly, to Corollary~\ref{cor:sunFECCt,w} we can adapt the construction of Theorem~\ref{thm:sunG}
to obtain an $(n,t,w,e)$-LPECC code.

\vspace{2.00ex}
\section{Asymptotic Behavior}
\label{sec:asymptotic}
\vspace{0.50ex}

In this section we will analyze the asymptotic behavior of the thermal codes
which were constructed in the previous sections. We start in subsection~\ref{sec:asymp_cool}
where we consider only cooling codes. As was proved in Section~\ref{spreads},
when $t+1 \leq n/2$ our codes which use only $t+1$ redundancy bits are optimal and
the number of additional wires used is negligible when $k$ or $n$ are large enough.
Hence, the interesting case is the asymptotic behavior when $t+1 > n/2$ and the
sunflower construction is used.
The case is even more complicated when we consider low-power cooling codes.
Our methods with efficient encoding and decoding algorithms are not optimal,
but this does not exclude the possibility of being asymptotically optimal.
The asymptotic behavior of these codes will be considered in subsection~\ref{sec:asymp_LP_cool}.
In addition, to find asymptotically good codes,
we will consider in this subsection a new Gilbert-Varshamov type method,
called the expurgation method.

%
\subsection{\hspace*{-2pt}Asymptotic Behavior of Cooling Codes}
\vspace{-0.00ex}
\label{sec:asymp_cool}

Let $B(n,d)$ be the largest dimension of a binary $[n,\kappa,d]$ code.
For $0 < \delta,~ \tau < 1$, define
\begin{align*}
\beta(\delta) & ~\deff~ \limsup_{n\to\infty} \frac{B(n,\floor{\delta n})}{n}, \\
\varrho(\tau) & ~\deff~ \limsup_{n\to\infty} \frac{\log_2 C(n,\floor{\tau n})}{n}.
\end{align*}

By Lemma~\ref{lem:upperCC}, we have that any $(n,t)$-cooling code of size $M$,
satisfies $M \leq 2^{n-t}$. This immediately implies the following
asymptotic upper bound on the rate of an $(n,t)$-cooling code.

\begin{corollary}
\label{cor:asymp_upperCC}
If $0 < \tau < 1$, then
$$\varrho(\tau) \leq 1 - \tau~.$$
\end{corollary}

The following proposition for computing an asymptotic lower bound
on the rate of an $(n,t)$-cooling code.


\begin{proposition}
\label{prop:sunflower-asymp}
Assume $0< \epsilon <1$ satisfies the following conditions:
\begin{align}
\epsilon & > \beta(\delta) \label{eq:eps-lower}\\
\tau +\epsilon(1-\tau) & < \frac{1+\beta(\delta(1-\tau))}{2} \label{eq:eps-upper}.
\end{align}
Then $\varrho(\tau)\ge (1-\tau)(1-\epsilon)$.
\end{proposition}

\begin{proof}
By the definition of $\beta$, there exists $N_1$ such that
\begin{equation}
\label{pf:one}
B(M,\floor{\delta M})<\epsilon M ,
\end{equation}
for all $M \ge N_1$.

If $\delta'=\delta(1-\tau)$ and $0 < \epsilon_1 <1$, then
for any given large enough integer $N_2$, there exists an integer $N \ge N_2$, such that
\begin{equation}
\label{pf:two}
B(N,\floor{\delta'N})\ge \beta(\delta')N-\epsilon_1{N}~.
\end{equation}

Set $s=\floor{\beta(\delta')N-\epsilon_1 N}$, $d=\floor{\delta'N}$,
$t=\lceil{\tau N}\rceil$, $M=N-t$, and $r=\lceil{\epsilon M}\rceil$.
We claim that the conditions of Theorem~\ref{thm:sunF} are met for these parameters.

Observe that
\begin{align}
t+r &= \lceil{\tau N}\rceil +\lceil{\epsilon M}\rceil = \lceil{\tau N}\rceil +\lceil{\epsilon (N-t)}\rceil \notag \\
&=  \lceil{\tau N}\rceil +\lceil{\epsilon (N-\lceil{\tau N}\rceil)}\rceil \notag\\
&\le (\tau N+1)+ (\epsilon(1-\tau)N+1)\notag\\
&= \tau N+\epsilon(1-\tau)N+2. \label{pf:three}
\end{align}
On the other hand,
\begin{align}
\frac{N+s}{2}& =\frac{N+\floor{\beta(\delta')N -\epsilon_1 N}}{2} \notag\\
&\ge \frac{N+\beta(\delta')N-\epsilon_1N-1}{2}. \label{pf:four}
\end{align}

Now, let
\begin{equation}
\label{pf:five}
N_2\ge \frac{5/2}{{(1+\beta(\delta')-\epsilon_1)}/{2}-(\tau +\epsilon(1-\tau ))}.
\end{equation}
Furthermore, since~\eqref{eq:eps-upper} holds, it follows that there exists
$0 < \epsilon_2 <1$, such that for all $0< \epsilon_1 <\epsilon_2$
the denominator in the right hand side of~\eqref{pf:five} is strictly positive.
Hence, since $N > N_2$, we have that
\[\left({(1+\beta(\delta')-\epsilon_1)}/{2}-(\tau+\epsilon(1-\tau))\right)N\ge 5/2, \]
or,
\begin{equation}
\label{eq:sum_condF1}
\tau N+\epsilon(1-\tau)N+2\le  \frac{N+\beta(\delta')N-\epsilon_1N-1}{2}.
\end{equation}
Combining~\eqref{eq:sum_condF1} with Inequalities \eqref{pf:three} and \eqref{pf:four},
we have that $t+r\le (N+s)/2$, and the first condition of
Theorem~\ref{thm:sunF} is satisfied.

Next, since $s \leq {\beta(\delta')N} - \epsilon_1 N$, it follows from \eqref{pf:two}
that there exists an $[N,s,d]$ code, and the second condition of
Theorem~\ref{thm:sunF} is satisfied.
Finally, observe that
\begin{equation}
\label{eq:M}
M=N-t\ge N - (\tau N+1)= (1-\tau)N-1.
\end{equation}
If we choose
\begin{equation}
\label{pf:six}
N_2\ge \frac{N_1+1}{1-\tau},
\end{equation}
then since $N \geq N_2$, it follows by~\eqref{pf:six} that $N(1-\tau ) \geq N_1 +1$
which implies by~\eqref{eq:M} that $M \geq N_1$.
Therefore, from~\eqref{pf:one}, we infer that $B(M,\floor{\delta M})<\epsilon M\le r$.
Since $\floor{\delta M}\le \floor{\delta(1-\tau) N}=d$,
it follows that an $[N-t,r,d]$ code does not exist, and the third condition of
Theorem~\ref{thm:sunF} is satisfied.

In summary, to satisfy the three conditions of Theorem~\ref{thm:sunF},
we need $N_2$ to satisfy~\eqref{pf:five} and~\eqref{pf:six}.
Additionally, to guarantee that $N\ge N_0$ we require that $N_2 \geq N_0$
and hence we have that

{\small \[N_2= \max \left\{N_0, \frac{N_1+1}{1-\lambda},
 \frac{5/2}{{(1+\beta(\delta')-\epsilon_1)}/{2}-(\lambda+\epsilon(1-\lambda))}\right\}.\]}
Therefore, by Theorem~\ref{thm:sunF} there exists an $(N,t)$-cooling
code of size greater than $2^{N-t-r}$. Now,
\begin{align}
N-t-r & = N - \lceil{\tau N}\rceil - \lceil{\epsilon M}\rceil \notag \\
&= N - \lceil{\tau N}\rceil - \lceil{\epsilon (N- \lceil \tau N \rceil )}\rceil \notag \\
&\geq N - \tau N - \epsilon N- \epsilon \tau N -2 \notag \\
&= N (1- \tau)(1 - \epsilon) -2 \notag .
\end{align}
This implies that there exists an $(N,t)$-cooling code of size greater than
$2^{(1-\lambda)(1-\epsilon)N-2}$.
Since $t=\lceil{\lambda N}\rceil \ge \floor{\lambda N}$, it follows that there also exists an
$(N,\floor{\lambda N})$-cooling code of size greater than $2^{(1-\lambda)(1-\epsilon)N-2}$.
Thus,
$$
\frac{\log_2 C(N,\floor{\lambda N})}{N} \ge (1-\lambda)(1-\epsilon)-\frac2N ,
$$
which implies that $\varrho(\tau)\ge (1-\tau)(1-\epsilon)$.
\end{proof}

To apply Proposition~\ref{prop:sunflower-asymp} we have to use
the best known lower bound on $\beta (\cdot)$ in~\eqref{eq:eps-lower} and
the best known upper bound on $\beta (\cdot)$ in~\eqref{eq:eps-upper}.
For lower bound we will use the Gilbert-Varshamov bound~\cite{Gil52,Var57}.
For upper bound we will use the McEliece-Rodemich-Rumsey-Welch (MRRW) bound~\cite{MRRW}.
Specifically, we have

\begin{enumerate}
\item For the lower bound, the Gilbert-Varshamov bound implies that $\beta(\delta)\ge 1-H(\delta) +o(1)$,
where $H(\cdot)$ is the binary entropy function, given by
$$
H (x)\, \ \deff\  -x \log_2 \!x \,-\, (1-x) \log_2(1-x)~.
$$

\item As for the upper bound, the MRRW bound (equation (1.4) in~\cite{MRRW}) implies that
\begin{align}
\label{eq:mrrw}
& \beta (\delta) \leq {\rm MRRW}(\delta) \notag \\
& =\min_{0\le u\le 1-2\delta} 1+g\left(u^2\right)-g\left(u^2+2\delta u+2\delta\right),
\end{align}
where $g(x)=H((1-\sqrt{1-x})/2)$.

\end{enumerate}

Hence, we have the following corollary.

\begin{corollary}
\label{cor-sunflower-asymp-2}
Suppose that $0 <\epsilon ,~ \delta, ~\tau <1$ satisfy the following conditions:
\begin{align}
\epsilon & > {\rm MRRW}(\delta), \label{eq:eps-lower-22}\\
\tau+\epsilon(1-\tau) & < \frac{2-H(\delta(1-\tau))}{2}~.
\label{eq:eps-upper-23}
\end{align}
Then $\varrho (\tau)\ge (1-\tau)(1-\epsilon)$.
\end{corollary}
\begin{proof}
Since $\epsilon > {\rm MRRW}(\delta)$ by~\eqref{eq:eps-lower-22} and
$\beta(\delta)\le {\rm MRRW}(\delta)$ by~\eqref{eq:mrrw}, it follows that
$\epsilon > \beta (\delta)$ and hence Inequality~\eqref{eq:eps-lower} is satisfied.

From the Gilbert-Varshamov bound, we have that $\beta(\delta(1-\tau))\ge 1-H(\delta(1-\tau))+o(1)$.
Hence, we have
\begin{align*}
\tau+\epsilon(1-\tau) & < \frac{2-H(\delta(1-\tau))}{2} \\
& = \frac{1+(1-H(\delta(1-\tau)))}{2}\\
& \le \frac{1+\beta(\delta(1-\tau))}{2}.
\end{align*}
\noindent Therefore, Inequality~\eqref{eq:eps-upper} is satisfied.

Since the conditions in Proposition~\ref{prop:sunflower-asymp} are satisfied,
it follows from the proposition that $\varrho (\tau)\ge (1-\tau)(1-\epsilon)$.
\end{proof}

Now, fix $0 < \tau < 1$ and in what follows, we compute a lower bound on $\varrho(\tau)$
implied by Corollary~\ref{cor-sunflower-asymp-2}. Define the set
\vspace{-5mm}

{\small
\begin{equation*}
E(\tau) ~\deff~ \{\epsilon: \mbox{there exists a $\delta$ such that
$\epsilon,\delta$ satisfy \eqref{eq:eps-lower-22} and \eqref{eq:eps-upper-23}}\}.
\end{equation*}
}

\noindent This definition implies that by Corollary~\ref{cor-sunflower-asymp-2},
for all $\epsilon\in E(\tau)$,
we can have that
$\varrho(\tau)\ge (1-\tau)(1-\epsilon)$.
Maximizing the right hand side, we obtain
\begin{equation}\label{eq:asymp-sunflower}
\varrho(\tau)\ge (1-\tau)(1-\epsilon^*),
\end{equation}
where $\epsilon^*$ is the value of $\epsilon$ for which $(1-\tau)(1-\epsilon)$ is maximized
when $\epsilon^*$ is substituted for $\epsilon$. For the next computation
we define $E^*(\tau) ~\deff~ \epsilon^*$.
Via numerical computations, we have that
\[
E^*(\tau)=
\begin{cases}
\approx 0, & \mbox{if $0 < \tau \le 0.687$}\\
\mbox{between 0 and 1}, & \mbox{if $0.687\le \tau \le 0.737$}\\
\approx 1 & \mbox{if $0.737\le \tau < 1$}.
\end{cases}
\]
\noindent The plot of the lower bound \eqref{eq:asymp-sunflower} for $\varrho(\tau)$ is given
in Figure~\ref{asymp-cooling}. Without giving the formal analyze we also add a curve for the case that
the Gilbert-Varshamov bound is tight (at least in the binary case).

\begin{figure}
\begin{center}
\includegraphics[width=9cm]{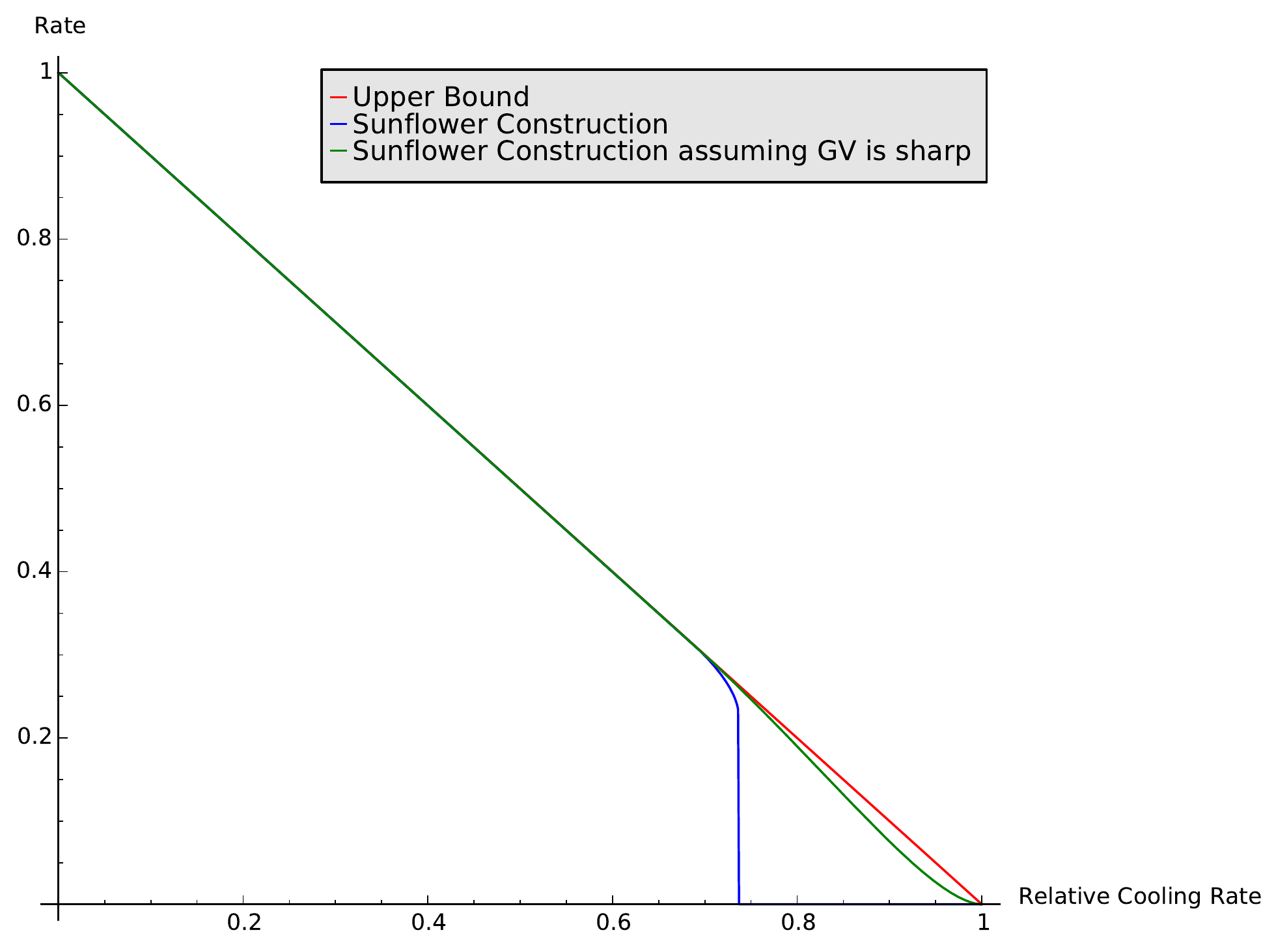}
\end{center}
\caption{Asymptotic Rates for Cooling Codes}
\label{asymp-cooling}
\end{figure}

Finally, we obtain the following result which is consistent with the figure,
and proves that for $\tau \leq 0.687$ our constructions are asymptotically optimal.

\begin{corollary}
\label{cor:less687}
If $\tau\le 0.687$, then $\varrho(\tau)=1-\tau$.
\end{corollary}
\begin{proof}
If $\tau\le 0.687$, then the value of $1- H((1-\tau)/2)/(1-\tau)$ is strictly positive.
If we set ${\delta=1/2}$, then \eqref{eq:eps-lower-22} and \eqref{eq:eps-upper-23} are reduced to
\begin{align}
\epsilon & >0\\
\tau+\epsilon(1-\tau) & <\frac{2-H((1-\tau)/2)}{2}.
\end{align}
As a consequence, $E(\tau)$ contains the interval $(0,1- H((1-\tau)/2)/(1-\tau))$.
Hence, the expression ${(1-\tau)(1-\epsilon)}$ is maximized for
$\epsilon$ which tends to 0. Therefore, $\varrho(\tau) \geq 1-\tau$
and since by Corollary~\ref{cor:asymp_upperCC}, we have $\varrho(\tau) \leq 1-\tau$,
it follows that $\varrho(\tau)=1-\tau$.
\end{proof}

%
\subsection{\hspace*{-2pt}Asymptotic Behavior of low-power Cooling Codes}
\vspace{-0.00ex}
\label{sec:asymp_LP_cool}

The asymptotic analysis when the codes have Property ${\bf B}(w)$
in addition to Property ${\bf A}(t)$ is more complicated, needless to
say that our methods are slightly less efficient compared to the methods
used to construct codes when only ${\bf A}(t)$ is satisfied.
Let $C(n,t,w)$ be the largest size of an $(n,t,w)$-LPC code and
for a fixed $0 < \tau <1$ and a fixed $0 < \omega <1$, consider the asymptotic rate

$$
\varrho (\tau,\omega) ~\deff~ \limsup_{n\to\infty} \frac{\log_2 C(n,\floor{\tau n},\floor{\omega n})}{n}.
$$

If $t=o(n)$ and $0 < \omega < 1$,
then we claim that Corollary~\ref{cor:spread_t,w} provides a family of $(n,t,\floor{\omega n})$-LPC codes
whose rates approach at least $H(\omega)$.
Formally, we claim that if $\C_n$ is such an $(n,t,\floor{\omega n})$-LPC code, then
\begin{equation}
\label{eq:remark}
\lim_{n\to \infty} \frac{\log_2|\C_n|}{n} \geq H(\omega).
\end{equation}
To prove~\eqref{eq:remark}, consider any given $\epsilon >0$, a prime power~$q$ and an integer $s >1$ such that
$$
q\le \sum_{i=0}^{\floor{\omega s}} \binom{s}{i}~\text{and}~\frac{\log_2q}{s}\ge H(\omega)-\epsilon.
$$
For each $n$, set $m=\floor{n/s}$. Since $t=o(n)$, it follows that $t+1\le m/2$ for sufficiently large $n$.
Applying Corollary~\ref{cor:spread_t,w}, we obtain an $(n,t,\floor{\omega n})$-LPC code $\C_n$ of size at least
$q^{m-t-1}$. Now, we have that
\begin{align*}
\lim_{n\to \infty} \frac{\log|\C_n|}{n}
& \ge \lim_{n\to\infty} \frac{\log_2q}{s}\left(1-\frac {t+1}m\right) \\
&\ge \left(H(\omega)-\epsilon\right)\cdot \lim_{n\to\infty}\left( 1-\frac {t+1}m\right) \\
&=H(\omega)-\epsilon,
\end{align*}
\noindent and hence $\lim_{n\to \infty} \frac{\log|\C_n|}{n} \geq H(\omega)$,
where the last equality follows since $t=o(n)$.

The next theorem describes a very simple construction which will be called
in the sequel the \emph{expurgation construction}.

\begin{theorem}
\label{thm:expurgation}
If $M > \sum_{i=w+1}^n \binom{n}{i}$ and an $(n,t)_q$-cooling code~$\C$ of size $M$ exists,
then there exists an $(n,t,w)$-LPC code of size
at least $M-\sum_{i=w+1}^n \binom{n}{i}$.
\end{theorem}

\begin{proof}
The theorem follows as an immediate consequence implied by removing all codewords of weight larger than $w$
from all the codesets of~$\C$.
\end{proof}

\begin{corollary}
\label{cor-expurgation}
If $\tau\le 0.687$, $\omega\ge 1/2$, and $H(\omega)<1-\tau$,
then $\varrho(\tau,\omega) \geq 1-\tau -o(1)$.
\end{corollary}

\begin{proof}
By the sunflower construction (see Theorem~\ref{thm:sunF})
there exists a family of $(n,t)$-cooling
codes whose size is at least $2^{n(1-\tau-\epsilon')}$,
where $t=\floor{\tau n}$, $r =\floor{\epsilon' n}$, and $\epsilon' = o(1)$
(see subsection~\ref{sec:asymp_cool}).
It is well known by using the Stirling's approximation~\cite[p. 310]{MWS} that if
$w=\floor{\omega n}$, then $\sum_{i=w+1}^n \binom{n}{i}\le 2^{n H(\omega)}$.

Now, applying the expurgation construction on a code from this
family yields an $(n,t,w)$-LPC code whose size
is at least $2^{n(1-\tau -\epsilon')}\left(1-2^{n(H(\omega)-1+\tau +\epsilon')}\right)$.
In other words, the rate, $\varrho (\tau , \omega )$,
of this family of $(n,t)$-cooling codes is at least
$$
(1-\tau -\epsilon')+\frac{\log \left(1-2^{n(H(\omega)-1+\tau +\epsilon')}\right)}{n},
$$
which is at least $1-\tau - \epsilon'$.
\end{proof}

Clearly, $\varrho (\tau,\omega) \leq \varrho (\tau)$ and hence combining
Corollaries~\ref{cor:asymp_upperCC} and~\ref{cor-expurgation} implies that
\begin{corollary}
\label{cor:expurgation2}
If $\tau\le 0.687$, $\omega\ge 1/2$, and $H(\omega)<1-\tau$,
then $\varrho(\tau,\omega) =1-\tau$.
\end{corollary}

Unfortunately, Theorem \ref{thm:expurgation} is nonconstructive and
the domain of $(\tau,\omega)$, where Corollary \ref{cor-expurgation} is applicable, is limited.
In the sequel we will consider other parameters outside this domain.
It should be no surprise that in most cases, we observed that codes obtained by the sunflower construction
have larger size than those obtained by the expurgation construction.
Nevertheless, in certain instances, the LPC codes obtained from Theorem~\ref{thm:expurgation}
has a larger size as compared to LPC codes resulting from the sunflower construction.
For such an example, consider $t=1$, ${w=\frac{2}{3}(n-1)}$, where $3$ divides $n-1$.
By Proposition~\ref{no_perfect},
there exists an $(n,1)$-cooling code of size $2^{n-1}$.
Theorem~\ref{thm:expurgation} yields an $(n,1,\frac{2}{3}(n-1))$-LPC code of size
at least $2^{n-1}-\sum_{i=w+1}^n \binom{n}{i}$.
Note that
\begin{equation}\label{ex:lpcc}
2^{n-1}-\sum_{i=w+1}^n \binom{n}{i}\ge 2^{n-1}-2^{nH(2/3)}\ge 2^{n-1}-2^{0.92n}.
\end{equation}

By the Griesmer bound~\cite[p. 546]{MWS}, there is no $[n-1,2,\frac{2}{3}(n-1)+1]$ code and
hence Theorem~\ref{thm:sunF} yields an $(n,1,\frac{2}{3}(n-1))$-LPC code of size
$M_2(n,3)+1$.

Since \[ \frac{5}{14}2^n>2^{0.92n}+1 \mbox{ for }n \geq 19,\]
it follows that $2^{n-1}-2^{0.92n}>M_2 (n,3) +1$, and hence
the expurgation construction yields an $(n,t,w)$-LPC code of larger size,
compared to the code obtained by the sunflower construction, in this case.

Recall, that Corollary~\ref{cor:sunLPCC} is used to apply the sunflower construction
and obtain $(n,t,d-1)$-LPC codes of size greater than $M_2(n-s,r+t-s)>2^{n-t-r}$,
where $r$ is given in Corollary~\ref{cor:sunLPCC}.
Furthermore, for $0 < \tau,~\omega < 1$, let $\delta=\omega/(1-\tau)$ and
define $\epsilon=\epsilon(\tau,\delta)$ as
\begin{equation*}
\label{eq:Etau}
\epsilon(\tau,\delta) ~\deff~ \inf \{\epsilon: \mbox{$\epsilon$ satisfy \eqref{eq:eps-lower-22} and \eqref{eq:eps-upper-23}}\}.
\end{equation*}
With this setting,
Corollary~\ref{cor-sunflower-asymp-2} implies that $\varrho(\tau,\omega)\ge (1-\tau)(1-\epsilon)$,
and by Corollary~\ref{cor:less687} and its proof we have
\begin{corollary}
If $\tau\le 0.687$ and $\omega\ge (1-\tau)/2$, then $\varrho(\tau,\omega)=1-\tau$.
\end{corollary}

We continue to examine the asymptotic behavior of $(n,t,w)$-LPC codes
constructed from cooling codes over $\Fq$ with concatenation.

\begin{corollary}
\label{cor:asym_LPCCq}
For given $0 <\tau, ~\omega<1$,
suppose that there exists a prime power $q$ and an integer $s>1$ such that
\begin{equation}
\label{eq:qr}
q\le \sum_{i=0}^{\floor{\omega s}} \binom{s}{i} \text{~~and~~} \tau s \le \frac 12.
\end{equation}
Then $\varrho(\tau,\omega)\ge (1-\tau s)\frac{\log_2 q}{s}$.
\end{corollary}
\begin{proof}
For a given $n$, set $m=\ceil{\frac{n}{s}+\frac{1}{\tau s}}$ and $t=\floor{\tau n}$.
Since $\tau s \le 1/2$, if follows that
$$
t+1\le \tau n+1= \left( \frac{n}{s}+\frac{1}{\tau s} \right) \tau s\le m/2~.
$$
Hence, by Corollary~\ref{cor:spread_t,w},
there exists an $(ms,t,m \floor{\omega s})$-LPC code of size at least $q^{m-t-1}$.
Therefore, we have that
\begin{align*}
\varrho(\tau,\omega) & \ge \lim_{{ms}\to\infty} \frac{m-t-1}{ms}\log_2 q\\
&=  \frac{\log_2 q}{s}\lim_{{ms}\to\infty} \frac{sm-st}{ms}\\
&=  \frac{\log_2 q}{s}\lim_{{ms}\to\infty} \left( 1 - \frac tm \right) \\
&=  \frac{\log_2 q}{s}\lim_{{ms}\to\infty}\left(1-\frac {\tau n}{n/s+1/(\tau s)}\right)\\
&>  \frac{\log_2 q}{s}\lim_{{ms}\to\infty}\left(1-\frac {\tau n}{n/s}\right)\\
&=\frac{\log_2 q}{s}(1-\tau s).
\end{align*}
\end{proof}

Note, that to apply Corollary~\ref{cor:asym_LPCCq} we have to be careful in choosing $q$ and $s$
such that equation~\eqref{eq:qr} is satisfied.

\vspace{2.00ex}
\section{Conclusion and Future Research}
\label{sec:conclude}
\vspace{0.50ex}

High temperatures have dramatic negative effects on interconnect
performance and, hence, it is important to suggest techniques to reduce
the power consumption of on-chip buses. We have suggested coding techniques to
balance and reduce the power consumption of on-chip buses. Our codes can have three
features (properties).
A code is a "cooling code" if there are no transitions on the $t$ hottest wires (Property ${\bf A} (t)$).
A code has "low-power" if it reduces the number of transitions on
the bus wires to at most $w$ transitions (Property ${\bf B} (w)$).
A code is "error-correcting" if it can correct any $e$ errors of bus transitions (Property ${\bf C} (e)$). A code can
have some of these properties. Six subsets out of the possible eight subsets of these properties
are interesting in our context and they are considered in the following sections and subsections:

\begin{enumerate}
\item $(n,w)$-low-power codes ($(n,w)$-LP codes) were considered in
Section~\ref{anticodes} and subsection~\ref{sec:Differential}.

\item $(n,t)$-cooling codes were considered in Section~\ref{spreads}.

\item $(n,t,w)$-low-power cooling codes ($(n,t,w)$-LPC codes) were considered
in Section~\ref{sec:LPCC}.

\item $(n,w,e)$-low-power error-correcting codes ($(n,w,e)$-LPEC codes)
were considered in subection~\ref{sec:LPECC}).

\item $(n,t,e)$-error-correcting cooling codes ($(n,t,e)$-ECC codes)
were considered in subsection~\ref{sec:ErrorCorrectCoolCodes}.

\item $(n,t,w,e)$-low-power error-correcting cooling codes ($(n,t,w,e)$-LPECC codes)
were considered in\linebreak subsection~\ref{sec:LPECCcool}.
\end{enumerate}

Our cooling codes without error-correction are optimal when $t+1 \leq n/2$ and
can be proved to be optimal also in some cases when $t+1 > n/2$. In all these cases the
redundancy of these codes is $t+1$. This redundancy, compared to the related best known
error-correcting code, is also obtained for error-correcting
cooling codes. Finally, the asymptotic analysis shows that in most cases our codes are
asymptotically optimal.

For the combination of low-power cooling code with or without error-correction, our
codes fall short of the known upper bounds and closing this gap is one of the problems for future research.
In these cases, we would like to see efficient encoding and decoding algorithms.
We would also like to improve our bounds and have asymptotic optimal codes
for the case were $\tau = \frac{t}{n} > 0.687$. Finally, we would like to see more cases where sunflower
construction with the generalized Hamming weights implies a large code.
As we mention in the Introduction, a follow up work
will present more construction, especially for practical parameters. Examples and numerical
experiments will be given and also comparison between the various constructions with emphasis
on practical parameters. It will be illustrated and discussed how practical our methods
and constructions are.

\section*{Acknowledgement}

The authors would like to thank Chaoping Xing and Eitan Yaakobi for many helpful discussions.
They also thank Shuangqing Liu for pointing on an error in the proof of Proposition 4 in an earlier draft.

{\newpage
\begin{center}
{\sc Appendix}\\
TABLE I. Admissible Parameters for Sunflower Constructions
\end{center}
\begin{multicols*}{6}
\TrickSupertabularIntoMulticols
\renewcommand{\arraystretch}{0.7521}
\centering
\tiny
\tablehead{\hline
$n$ & $t$ & $r$ & $s$ & $d$ \\
\hline
}
\tabletail{
\hline
}
\vspace{-5mm}
\begin{supertabular}{|c c c c c|}

5 & 2 & 1 & 1 & 5 \\
\hline
6 & 3 & 1 & 2 & 4 \\
\hline
7 & 3 & 1 & 1 & 7 \\
7 & 4 & 1 & 3 & 4 \\
\hline
8 & 4 & 1 & 2 & 5 \\
8 & 5 & 1 & 4 & 4 \\
\hline
9 & 4 & 1 & 1 & 9 \\
9 & 5 & 2 & 5 & 3 \\
\hline
10 & 5 & 1 & 2 & 6 \\
10 & 6 & 2 & 6 & 3 \\
\hline
11 & 5 & 1 & 1 & 11 \\
11 & 6 & 1 & 3 & 6 \\
11 & 7 & 2 & 7 & 3 \\
\hline
12 & 6 & 1 & 2 & 8 \\
12 & 7 & 1 & 4 & 6 \\
12 & 8 & 2 & 8 & 3 \\
\hline
13 & 6 & 1 & 1 & 13 \\
13 & 7 & 1 & 3 & 7 \\
13 & 8 & 2 & 7 & 4 \\
13 & 9 & 2 & 9 & 3 \\
\hline
14 & 7 & 1 & 2 & 9 \\
14 & 8 & 1 & 4 & 7 \\
14 & 9 & 2 & 8 & 4 \\
14 & 10 & 2 & 10 & 3 \\
\hline
15 & 7 & 1 & 1 & 15 \\
15 & 8 & 1 & 3 & 8 \\
15 & 9 & 1 & 5 & 7 \\
15 & 10 & 2 & 9 & 4 \\
15 & 11 & 2 & 11 & 3 \\
\hline
16 & 8 & 1 & 2 & 10 \\
16 & 9 & 1 & 4 & 8 \\
16 & 10 & 2 & 8 & 5 \\
16 & 11 & 2 & 10 & 4 \\
\hline
17 & 8 & 1 & 1 & 17 \\
17 & 9 & 1 & 3 & 9 \\
17 & 10 & 1 & 5 & 8 \\
17 & 11 & 2 & 9 & 5 \\
17 & 12 & 2 & 11 & 4 \\
\hline
18 & 9 & 1 & 2 & 12 \\
18 & 10 & 2 & 9 & 6 \\
18 & 11 & 1 & 6 & 8 \\
18 & 12 & 3$^*$ & 12 & 4 \\
18 & 13 & 2 & 12 & 4 \\
\hline
19 & 9 & 1 & 1 & 19 \\
19 & 10 & 1 & 3 & 10 \\
19 & 11 & 2 & 7 & 8 \\
19 & 12 & 1 & 7 & 8 \\
19 & 13 & 3$^*$ & 13 & 4 \\
19 & 14 & 2 & 13 & 4 \\
\hline
20 & 10 & 1 & 2 & 13 \\
20 & 11 & 1 & 4 & 10 \\
20 & 12 & 2 & 8 & 8 \\
20 & 13 & 1 & 8 & 8 \\
20 & 14 & 3$^*$ & 14 & 4 \\
20 & 15 & 2 & 14 & 4 \\
\hline
21 & 10 & 1 & 1 & 21 \\
21 & 11 & 1 & 3 & 12 \\
21 & 12 & 1 & 5 & 10 \\
21 & 13 & 2 & 9 & 8 \\
21 & 14 & 1 & 9 & 8 \\
21 & 15 & 3 & 15 & 4 \\
21 & 16 & 2 & 15 & 4 \\
\hline
22 & 11 & 1 & 2 & 14 \\
22 & 12 & 1 & 4 & 11 \\
22 & 13 & 2 & 8 & 8 \\
22 & 14 & 2 & 10 & 8 \\
22 & 15 & 1 & 10 & 8 \\
22 & 16 & 3 & 16 & 4 \\
22 & 17 & 2 & 16 & 4 \\
\hline
23 & 11 & 1 & 1 & 23 \\
23 & 12 & 1 & 3 & 12 \\
23 & 13 & 1 & 5 & 11 \\
23 & 14 & 2 & 9 & 8 \\
23 & 15 & 2 & 11 & 8 \\
23 & 16 & 1 & 11 & 8 \\
23 & 17 & 3 & 17 & 4 \\
23 & 18 & 2 & 17 & 4 \\
\hline
24 & 12 & 1 & 2 & 16 \\
24 & 13 & 1 & 4 & 12 \\
24 & 14 & 2 & 12 & 8 \\
24 & 15 & 2 & 10 & 8 \\
24 & 16 & 2 & 12 & 8 \\
24 & 17 & 1 & 12 & 8 \\
24 & 18 & 3 & 18 & 4 \\
24 & 19 & 2 & 18 & 4 \\
\hline
25 & 12 & 1 & 1 & 25 \\
25 & 13 & 1 & 3 & 14 \\
25 & 14 & 1 & 5 & 12 \\
25 & 15 & 2 & 9 & 8 \\
25 & 16 & 2 & 11 & 8 \\
25 & 17 & 2 & 13 & 6 \\
25 & 18 & 2 & 15 & 5 \\
25 & 19 & 3 & 19 & 4 \\
25 & 20 & 2 & 19 & 4 \\
\hline
26 & 13 & 1 & 2 & 17 \\
26 & 14 & 1 & 4 & 13 \\
26 & 15 & 1 & 6 & 12 \\
26 & 16 & 2 & 10 & 8 \\
26 & 17 & 2 & 12 & 8 \\
26 & 18 & 2 & 14 & 6 \\
26 & 19 & 2 & 16 & 5 \\
26 & 20 & 3 & 20 & 4 \\
26 & 21 & 2 & 20 & 4 \\
\hline
27 & 13 & 1 & 1 & 27 \\
27 & 14 & 1 & 3 & 15 \\
27 & 15 & 1 & 5 & 13 \\
27 & 16 & 1 & 7 & 12 \\
27 & 17 & 2 & 11 & 8 \\
27 & 18 & 2 & 13 & 8 \\
27 & 19 & 2 & 15 & 6 \\
27 & 20 & 2 & 17 & 5 \\
27 & 21 & 3 & 21 & 4 \\
27 & 22 & 2 & 21 & 4 \\
\hline
28 & 14 & 1 & 2 & 18 \\
28 & 15 & 1 & 4 & 14 \\
28 & 16 & 2 & 8 & 11 \\
28 & 17 & 2 & 10 & 10 \\
28 & 18 & 2 & 12 & 8 \\
28 & 19 & 2 & 14 & 8 \\
28 & 20 & 2 & 16 & 6 \\
28 & 21 & 2 & 18 & 5 \\
28 & 22 & 3 & 22 & 4 \\
28 & 23 & 2 & 22 & 4 \\
\hline
29 & 14 & 1 & 1 & 29 \\
29 & 15 & 1 & 3 & 16 \\
29 & 16 & 1 & 5 & 14 \\
29 & 17 & 2 & 9 & 11 \\
29 & 18 & 2 & 11 & 9 \\
29 & 19 & 2 & 13 & 8 \\
29 & 20 & 2 & 15 & 7 \\
29 & 21 & 2 & 17 & 6 \\
29 & 22 & 2 & 19 & 5 \\
29 & 23 & 3 & 23 & 4 \\
29 & 24 & 2 & 23 & 4 \\
\hline
30 & 15 & 1 & 2 & 20 \\
30 & 16 & 1 & 4 & 16 \\
30 & 17 & 1 & 6 & 14 \\
30 & 18 & 2 & 10 & 11 \\
30 & 19 & 2 & 12 & 9 \\
30 & 20 & 2 & 14 & 8 \\
30 & 21 & 2 & 16 & 7 \\
30 & 22 & 2 & 18 & 6 \\
30 & 23 & 2 & 20 & 5 \\
30 & 24 & 3 & 24 & 4 \\
30 & 25 & 2 & 24 & 4 \\
\hline
31 & 15 & 1 & 1 & 31 \\
31 & 16 & 1 & 3 & 17 \\
31 & 17 & 1 & 5 & 16 \\
31 & 18 & 2 & 9 & 12 \\
31 & 19 & 2 & 11 & 11 \\
31 & 20 & 2 & 13 & 9 \\
31 & 21 & 2 & 15 & 8 \\
31 & 22 & 2 & 17 & 7 \\
31 & 23 & 2 & 19 & 6 \\
31 & 24 & 2 & 21 & 5 \\
31 & 25 & 3 & 25 & 4 \\
31 & 26 & 2 & 25 & 4 \\
\hline
32 & 16 & 1 & 2 & 21 \\
32 & 17 & 1 & 4 & 16 \\
32 & 18 & 1 & 6 & 16 \\
32 & 19 & 2 & 10 & 12 \\
32 & 20 & 2 & 12 & 10 \\
32 & 21 & 2 & 14 & 8 \\
32 & 22 & 2 & 16 & 8 \\
32 & 23 & 3 & 20 & 6 \\
32 & 24 & 2 & 20 & 6 \\
32 & 25 & 2 & 22 & 5 \\
32 & 26 & 3 & 26 & 4 \\
32 & 27 & 2 & 26 & 4 \\
\hline
33 & 16 & 1 & 1 & 33 \\
33 & 17 & 1 & 3 & 18 \\
33 & 18 & 1 & 5 & 16 \\
33 & 19 & 2 & 9 & 12 \\
33 & 20 & 2 & 11 & 12 \\
33 & 21 & 2 & 13 & 10 \\
33 & 22 & 2 & 15 & 8 \\
33 & 23 & 2 & 17 & 8 \\
33 & 24 & 3 & 21 & 6 \\
33 & 25 & 2 & 21 & 6 \\
33 & 26 & 2 & 23 & 5 \\
\hline
34 & 17 & 1 & 2 & 22 \\
34 & 18 & 1 & 4 & 17 \\
34 & 19 & 1 & 6 & 16 \\
34 & 20 & 2 & 10 & 12 \\
34 & 21 & 2 & 12 & 12 \\
34 & 22 & 2 & 14 & 10 \\
34 & 23 & 2 & 16 & 8 \\
34 & 24 & 2 & 18 & 8 \\
34 & 25 & 3 & 22 & 6 \\
34 & 26 & 2 & 22 & 6 \\
\hline
35 & 17 & 1 & 1 & 35 \\
35 & 18 & 1 & 3 & 20 \\
35 & 19 & 2 & 7 & 16 \\
35 & 20 & 1 & 7 & 16 \\
35 & 21 & 2 & 11 & 12 \\
35 & 22 & 2 & 13 & 11 \\
35 & 23 & 2 & 15 & 10 \\
35 & 24 & 2 & 17 & 8 \\
35 & 25 & 2 & 19 & 8 \\
35 & 26 & 3 & 23 & 6 \\
35 & 27 & 2 & 23 & 6 \\
\hline
36 & 18 & 1 & 2 & 24 \\
36 & 19 & 1 & 4 & 18 \\
36 & 20 & 2 & 8 & 16 \\
36 & 21 & 1 & 8 & 16 \\
36 & 22 & 2 & 12 & 12 \\
36 & 23 & 2 & 14 & 11 \\
36 & 24 & 2 & 16 & 10 \\
36 & 25 & 2 & 18 & 8 \\
36 & 26 & 2 & 20 & 8 \\
36 & 27 & 3 & 24 & 6 \\
36 & 28 & 2 & 24 & 6 \\
\hline
37 & 18 & 1 & 1 & 37 \\
37 & 19 & 1 & 3 & 20 \\
37 & 20 & 1 & 5 & 18 \\
37 & 21 & 2 & 9 & 15 \\
37 & 22 & 2 & 11 & 13 \\
37 & 23 & 2 & 13 & 12 \\
37 & 24 & 2 & 15 & 10 \\
37 & 25 & 2 & 17 & 9 \\
37 & 26 & 2 & 19 & 8 \\
37 & 27 & 2 & 21 & 8 \\
37 & 28 & 3 & 25 & 6 \\
37 & 29 & 2 & 25 & 6 \\
\hline
38 & 19 & 1 & 2 & 25 \\
38 & 20 & 1 & 4 & 20 \\
38 & 21 & 1 & 6 & 18 \\
38 & 22 & 2 & 10 & 14 \\
38 & 23 & 2 & 12 & 13 \\
38 & 24 & 2 & 14 & 12 \\
38 & 25 & 2 & 16 & 10 \\
38 & 26 & 2 & 18 & 9 \\
38 & 27 & 2 & 20 & 8 \\
38 & 28 & 2 & 22 & 8 \\
38 & 29 & 3 & 26 & 6 \\
38 & 30 & 2 & 26 & 6 \\
\hline
39 & 19 & 1 & 1 & 39 \\
39 & 20 & 1 & 3 & 22 \\
39 & 21 & 1 & 5 & 19 \\
39 & 22 & 2 & 9 & 16 \\
39 & 23 & 2 & 11 & 14 \\
39 & 24 & 2 & 13 & 12 \\
39 & 25 & 2 & 15 & 12 \\
39 & 26 & 2 & 17 & 10 \\
39 & 27 & 2 & 19 & 9 \\
39 & 28 & 2 & 21 & 8 \\
39 & 29 & 2 & 23 & 7 \\
39 & 30 & 3 & 27 & 6 \\
39 & 31 & 2 & 27 & 6 \\
\hline
40 & 20 & 1 & 2 & 26 \\
40 & 21 & 1 & 4 & 20 \\
40 & 22 & 2 & 8 & 16 \\
40 & 23 & 2 & 10 & 16 \\
40 & 24 & 2 & 12 & 14 \\
40 & 25 & 2 & 14 & 12 \\
40 & 26 & 2 & 16 & 12 \\
40 & 27 & 2 & 18 & 10 \\
40 & 28 & 2 & 20 & 9 \\
40 & 29 & 2 & 22 & 8 \\
40 & 30 & 2 & 24 & 7 \\
40 & 31 & 3 & 28 & 6 \\
40 & 32 & 2 & 28 & 6 \\
\hline
41 & 20 & 1 & 1 & 41 \\
41 & 21 & 1 & 3 & 23 \\
41 & 22 & 1 & 5 & 20 \\
41 & 23 & 2 & 9 & 16 \\
41 & 24 & 2 & 11 & 16 \\
41 & 25 & 2 & 13 & 14 \\
41 & 26 & 2 & 15 & 12 \\
41 & 27 & 2 & 17 & 12 \\
41 & 28 & 2 & 19 & 10 \\
41 & 29 & 2 & 21 & 9 \\
41 & 30 & 2 & 23 & 8 \\
41 & 31 & 2 & 25 & 7 \\
41 & 32 & 3 & 29 & 6 \\
41 & 33 & 2 & 29 & 6 \\
\hline
42 & 21 & 1 & 2 & 28 \\
42 & 22 & 1 & 4 & 22 \\
42 & 23 & 1 & 6 & 20 \\
42 & 24 & 2 & 10 & 16 \\
42 & 25 & 2 & 12 & 15 \\
42 & 26 & 2 & 14 & 13 \\
42 & 27 & 2 & 16 & 12 \\
42 & 28 & 2 & 18 & 12 \\
42 & 29 & 2 & 20 & 10 \\
42 & 30 & 3 & 24 & 8 \\
42 & 31 & 2 & 24 & 8 \\
42 & 32 & 2 & 26 & 7 \\
42 & 33 & 3 & 30 & 6 \\
42 & 34 & 2 & 30 & 6 \\
\hline
43 & 21 & 1 & 1 & 43 \\
43 & 22 & 1 & 3 & 24 \\
43 & 23 & 1 & 5 & 21 \\
43 & 24 & 1 & 7 & 20 \\
43 & 25 & 2 & 11 & 16 \\
43 & 26 & 2 & 13 & 15 \\
43 & 27 & 2 & 15 & 13 \\
43 & 28 & 2 & 17 & 12 \\
43 & 29 & 2 & 19 & 12 \\
43 & 30 & 2 & 21 & 10 \\
43 & 31 & 3 & 25 & 8 \\
43 & 32 & 2 & 25 & 8 \\
43 & 33 & 2 & 27 & 7 \\
43 & 34 & 3 & 31 & 6 \\
43 & 35 & 2 & 31 & 6 \\
\hline
44 & 22 & 1 & 2 & 29 \\
44 & 23 & 1 & 4 & 23 \\
44 & 24 & 1 & 6 & 21 \\
44 & 25 & 2 & 10 & 17 \\
44 & 26 & 2 & 12 & 16 \\
44 & 27 & 2 & 14 & 15 \\
44 & 28 & 2 & 16 & 13 \\
44 & 29 & 2 & 18 & 12 \\
44 & 30 & 2 & 20 & 12 \\
44 & 31 & 2 & 22 & 10 \\
44 & 32 & 3 & 26 & 8 \\
44 & 33 & 2 & 26 & 8 \\
44 & 34 & 2 & 28 & 7 \\
44 & 35 & 3 & 32 & 6 \\
44 & 36 & 2 & 32 & 6 \\
\hline
45 & 22 & 1 & 1 & 45 \\
45 & 23 & 1 & 3 & 25 \\
45 & 24 & 1 & 5 & 22 \\
45 & 25 & 2 & 9 & 18 \\
45 & 26 & 2 & 11 & 16 \\
45 & 27 & 2 & 13 & 16 \\
45 & 28 & 2 & 15 & 14 \\
45 & 29 & 2 & 17 & 12 \\
45 & 30 & 2 & 19 & 12 \\
45 & 31 & 2 & 21 & 12 \\
45 & 32 & 2 & 23 & 10 \\
45 & 33 & 3 & 27 & 8 \\
45 & 34 & 2 & 27 & 8 \\
45 & 35 & 2 & 29 & 7 \\
45 & 36 & 3 & 33 & 6 \\
45 & 37 & 2 & 33 & 6 \\
\hline
46 & 23 & 1 & 2 & 30 \\
46 & 24 & 1 & 4 & 24 \\
46 & 25 & 1 & 6 & 22 \\
46 & 26 & 2 & 10 & 18 \\
46 & 27 & 2 & 12 & 16 \\
46 & 28 & 2 & 14 & 16 \\
46 & 29 & 2 & 16 & 14 \\
46 & 30 & 2 & 18 & 12 \\
46 & 31 & 2 & 20 & 12 \\
46 & 32 & 2 & 22 & 12 \\
46 & 33 & 2 & 24 & 10 \\
46 & 34 & 3 & 28 & 8 \\
46 & 35 & 2 & 28 & 8 \\
46 & 36 & 2 & 30 & 7 \\
46 & 37 & 3 & 34 & 6 \\
46 & 38 & 2 & 34 & 6 \\
\hline
47 & 23 & 1 & 1 & 47 \\
47 & 24 & 1 & 3 & 26 \\
47 & 25 & 1 & 5 & 24 \\
47 & 26 & 1 & 7 & 22 \\
47 & 27 & 2 & 11 & 18 \\
47 & 28 & 2 & 13 & 16 \\
47 & 29 & 2 & 15 & 16 \\
47 & 30 & 2 & 17 & 14 \\
47 & 31 & 2 & 19 & 12 \\
47 & 32 & 2 & 21 & 12 \\
47 & 33 & 2 & 23 & 12 \\
47 & 34 & 2 & 25 & 9 \\
47 & 35 & 3 & 29 & 8 \\
47 & 36 & 2 & 29 & 8 \\
47 & 37 & 2 & 31 & 7 \\
47 & 38 & 3 & 35 & 6 \\
47 & 39 & 2 & 35 & 6 \\
\hline
48 & 24 & 1 & 2 & 32 \\
48 & 25 & 1 & 4 & 24 \\
48 & 26 & 1 & 6 & 24 \\
48 & 27 & 1 & 8 & 22 \\
48 & 28 & 2 & 12 & 17 \\
48 & 29 & 2 & 14 & 16 \\
48 & 30 & 2 & 16 & 16 \\
48 & 31 & 2 & 18 & 13 \\
48 & 32 & 2 & 20 & 12 \\
48 & 33 & 2 & 22 & 12 \\
48 & 34 & 2 & 24 & 12 \\
48 & 35 & 2 & 26 & 9 \\
48 & 36 & 3 & 30 & 8 \\
48 & 37 & 2 & 30 & 8 \\
48 & 38 & 3 & 34 & 6 \\
48 & 39 & 3 & 36 & 6 \\
48 & 40 & 2 & 36 & 6 \\
\hline
49 & 24 & 1 & 1 & 49 \\
49 & 25 & 1 & 3 & 28 \\
49 & 26 & 1 & 5 & 24 \\
49 & 27 & 1 & 7 & 23 \\
49 & 28 & 2 & 11 & 19 \\
49 & 29 & 2 & 13 & 17 \\
49 & 30 & 2 & 15 & 16 \\
49 & 31 & 2 & 17 & 14 \\
49 & 32 & 2 & 19 & 12 \\
49 & 33 & 2 & 21 & 12 \\
49 & 34 & 2 & 23 & 12 \\
49 & 35 & 2 & 25 & 10 \\
49 & 36 & 2 & 27 & 9 \\
49 & 37 & 3 & 31 & 8 \\
49 & 38 & 2 & 31 & 8 \\
49 & 39 & 3 & 35 & 6 \\
49 & 40 & 3 & 37 & 5 \\
49 & 41 & 5 & 43 & 3 \\
\hline
50 & 25 & 1 & 2 & 33 \\
50 & 26 & 1 & 4 & 26 \\
50 & 27 & 1 & 6 & 24 \\
50 & 28 & 1 & 8 & 23 \\
50 & 29 & 2 & 12 & 18 \\
50 & 30 & 2 & 14 & 17 \\
50 & 31 & 2 & 16 & 16 \\
50 & 32 & 2 & 18 & 14 \\
50 & 33 & 2 & 20 & 12 \\
50 & 34 & 2 & 22 & 12 \\
50 & 35 & 2 & 24 & 12 \\
50 & 36 & 2 & 26 & 10 \\
50 & 37 & 3 & 30 & 8 \\
50 & 38 & 3 & 32 & 8 \\
50 & 39 & 2 & 32 & 8 \\
50 & 40 & 3 & 36 & 6 \\
50 & 41 & 3 & 38 & 5 \\
50 & 42 & 5 & 44 & 3 \\
\hline
51 & 25 & 1 & 1 & 51 \\
51 & 26 & 1 & 3 & 28 \\
51 & 27 & 1 & 5 & 25 \\
51 & 28 & 1 & 7 & 24 \\
51 & 29 & 2 & 11 & 20 \\
51 & 30 & 2 & 13 & 18 \\
51 & 31 & 2 & 15 & 16 \\
51 & 32 & 2 & 17 & 16 \\
51 & 33 & 2 & 19 & 14 \\
51 & 34 & 2 & 21 & 12 \\
51 & 35 & 2 & 23 & 12 \\
51 & 36 & 2 & 25 & 11 \\
51 & 37 & 2 & 27 & 10 \\
51 & 38 & 3 & 31 & 8 \\
51 & 39 & 3 & 33 & 8 \\
51 & 40 & 2 & 33 & 8 \\
51 & 41 & 3 & 37 & 6 \\
51 & 42 & 3 & 39 & 5 \\
51 & 43 & 5 & 45 & 3 \\
\hline
52 & 26 & 1 & 2 & 34 \\
52 & 27 & 1 & 4 & 27 \\
52 & 28 & 1 & 6 & 25 \\
52 & 29 & 1 & 8 & 24 \\
52 & 30 & 2 & 12 & 20 \\
52 & 31 & 2 & 14 & 18 \\
52 & 32 & 2 & 16 & 16 \\
52 & 33 & 2 & 18 & 16 \\
52 & 34 & 2 & 20 & 14 \\
52 & 35 & 2 & 22 & 12 \\
52 & 36 & 2 & 24 & 12 \\
52 & 37 & 3 & 28 & 10 \\
52 & 38 & 2 & 28 & 10 \\
52 & 39 & 3 & 32 & 8 \\
52 & 40 & 3 & 34 & 8 \\
52 & 41 & 2 & 34 & 8 \\
52 & 42 & 3 & 38 & 6 \\
52 & 43 & 3 & 40 & 5 \\
52 & 44 & 5 & 46 & 3 \\
\hline
53 & 26 & 1 & 1 & 53 \\
53 & 27 & 1 & 3 & 30 \\
53 & 28 & 1 & 5 & 26 \\
53 & 29 & 2 & 9 & 23 \\
53 & 30 & 2 & 11 & 20 \\
53 & 31 & 2 & 13 & 20 \\
53 & 32 & 2 & 15 & 18 \\
53 & 33 & 2 & 17 & 16 \\
53 & 34 & 2 & 19 & 15 \\
53 & 35 & 2 & 21 & 13 \\
53 & 36 & 2 & 23 & 12 \\
53 & 37 & 2 & 25 & 12 \\
53 & 38 & 3 & 29 & 10 \\
53 & 39 & 2 & 29 & 10 \\
53 & 40 & 3 & 33 & 8 \\
53 & 41 & 3 & 35 & 8 \\
53 & 42 & 2 & 35 & 8 \\
53 & 43 & 3 & 39 & 6 \\
53 & 44 & 3 & 41 & 5 \\
53 & 45 & 5 & 47 & 3 \\
\hline
54 & 27 & 1 & 2 & 36 \\
54 & 28 & 1 & 4 & 28 \\
54 & 29 & 1 & 6 & 26 \\
54 & 30 & 2 & 10 & 22 \\
54 & 31 & 2 & 12 & 20 \\
54 & 32 & 2 & 14 & 20 \\
54 & 33 & 2 & 16 & 18 \\
54 & 34 & 2 & 18 & 16 \\
54 & 35 & 2 & 20 & 15 \\
54 & 36 & 2 & 22 & 13 \\
54 & 37 & 2 & 24 & 12 \\
54 & 38 & 2 & 26 & 12 \\
54 & 39 & 3 & 30 & 10 \\
54 & 40 & 2 & 30 & 10 \\
54 & 41 & 3 & 34 & 8 \\
54 & 42 & 3 & 36 & 8 \\
54 & 43 & 2 & 36 & 8 \\
54 & 44 & 3 & 40 & 6 \\
54 & 45 & 3 & 42 & 5 \\
54 & 46 & 5 & 48 & 3 \\
\hline
55 & 27 & 1 & 1 & 55 \\
55 & 28 & 1 & 3 & 31 \\
55 & 29 & 1 & 5 & 28 \\
55 & 30 & 2 & 9 & 24 \\
55 & 31 & 2 & 11 & 22 \\
55 & 32 & 2 & 13 & 20 \\
55 & 33 & 2 & 15 & 20 \\
55 & 34 & 2 & 17 & 16 \\
55 & 35 & 2 & 19 & 16 \\
55 & 36 & 2 & 21 & 15 \\
55 & 37 & 2 & 23 & 13 \\
55 & 38 & 2 & 25 & 12 \\
55 & 39 & 2 & 27 & 12 \\
55 & 40 & 3 & 31 & 10 \\
55 & 41 & 2 & 31 & 10 \\
55 & 42 & 3 & 35 & 8 \\
55 & 43 & 3 & 37 & 8 \\
55 & 44 & 2 & 37 & 8 \\
55 & 45 & 3 & 41 & 6 \\
55 & 46 & 3 & 43 & 5 \\
55 & 47 & 5 & 49 & 3 \\
\hline
56 & 28 & 1 & 2 & 37 \\
56 & 29 & 1 & 4 & 29 \\
56 & 30 & 1 & 6 & 28 \\
56 & 31 & 2 & 10 & 24 \\
56 & 32 & 2 & 12 & 22 \\
56 & 33 & 2 & 14 & 20 \\
56 & 34 & 2 & 16 & 20 \\
56 & 35 & 2 & 18 & 16 \\
56 & 36 & 2 & 20 & 16 \\
56 & 37 & 2 & 22 & 14 \\
56 & 38 & 3 & 26 & 12 \\
56 & 39 & 2 & 26 & 12 \\
56 & 40 & 2 & 28 & 12 \\
56 & 41 & 3 & 32 & 9 \\
56 & 42 & 4 & 36 & 8 \\
56 & 43 & 3 & 36 & 8 \\
56 & 44 & 3 & 38 & 8 \\
56 & 45 & 2 & 38 & 8 \\
56 & 46 & 3 & 42 & 6 \\
56 & 47 & 3 & 44 & 5 \\
56 & 48 & 5 & 50 & 3 \\
\hline
57 & 28 & 1 & 1 & 57 \\
57 & 29 & 1 & 3 & 32 \\
57 & 30 & 1 & 5 & 28 \\
57 & 31 & 2 & 9 & 24 \\
57 & 32 & 2 & 11 & 23 \\
57 & 33 & 2 & 13 & 21 \\
57 & 34 & 2 & 15 & 20 \\
57 & 35 & 2 & 17 & 18 \\
57 & 36 & 2 & 19 & 16 \\
57 & 37 & 2 & 21 & 16 \\
57 & 38 & 2 & 23 & 14 \\
57 & 39 & 3 & 27 & 12 \\
57 & 40 & 2 & 27 & 12 \\
57 & 41 & 2 & 29 & 12 \\
57 & 42 & 3 & 33 & 9 \\
57 & 43 & 4 & 37 & 8 \\
57 & 44 & 3 & 37 & 8 \\
57 & 45 & 3 & 39 & 8 \\
57 & 46 & 2 & 39 & 8 \\
57 & 47 & 3 & 43 & 6 \\
57 & 48 & 3 & 45 & 5 \\
57 & 49 & 5 & 51 & 3 \\
\hline
58 & 29 & 1 & 2 & 38 \\
58 & 30 & 1 & 4 & 30 \\
58 & 31 & 1 & 6 & 28 \\
58 & 32 & 2 & 10 & 24 \\
58 & 33 & 2 & 12 & 23 \\
58 & 34 & 2 & 14 & 20 \\
58 & 35 & 2 & 16 & 20 \\
58 & 36 & 2 & 18 & 18 \\
58 & 37 & 2 & 20 & 16 \\
58 & 38 & 2 & 22 & 16 \\
58 & 39 & 2 & 24 & 14 \\
58 & 40 & 3 & 28 & 12 \\
58 & 41 & 2 & 28 & 12 \\
58 & 42 & 2 & 30 & 12 \\
58 & 43 & 3 & 34 & 9 \\
58 & 44 & 4 & 38 & 8 \\
58 & 45 & 3 & 38 & 8 \\
58 & 46 & 3 & 40 & 8 \\
58 & 47 & 2 & 40 & 8 \\
58 & 48 & 3 & 44 & 6 \\
58 & 49 & 3 & 46 & 5 \\
58 & 50 & 5 & 52 & 3 \\
\hline
59 & 29 & 1 & 1 & 59 \\
59 & 30 & 1 & 3 & 33 \\
59 & 31 & 1 & 5 & 30 \\
59 & 32 & 1 & 7 & 28 \\
59 & 33 & 2 & 11 & 24 \\
59 & 34 & 2 & 13 & 23 \\
59 & 35 & 2 & 15 & 20 \\
59 & 36 & 2 & 17 & 19 \\
59 & 37 & 2 & 19 & 17 \\
59 & 38 & 2 & 21 & 16 \\
59 & 39 & 2 & 23 & 16 \\
59 & 40 & 2 & 25 & 14 \\
59 & 41 & 3 & 29 & 12 \\
59 & 42 & 2 & 29 & 12 \\
59 & 43 & 2 & 31 & 12 \\
59 & 44 & 3 & 35 & 9 \\
59 & 45 & 4 & 39 & 8 \\
59 & 46 & 3 & 39 & 8 \\
59 & 47 & 3 & 41 & 8 \\
59 & 48 & 2 & 41 & 8 \\
59 & 49 & 3 & 45 & 6 \\
59 & 50 & 3 & 47 & 5 \\
59 & 51 & 5 & 53 & 3 \\
\hline
60 & 30 & 1 & 2 & 40 \\
60 & 31 & 1 & 4 & 32 \\
60 & 32 & 1 & 6 & 30 \\
60 & 33 & 2 & 10 & 25 \\
60 & 34 & 2 & 12 & 24 \\
60 & 35 & 2 & 14 & 22 \\
60 & 36 & 2 & 16 & 20 \\
60 & 37 & 2 & 18 & 18 \\
60 & 38 & 2 & 20 & 17 \\
60 & 39 & 2 & 22 & 16 \\
60 & 40 & 2 & 24 & 16 \\
60 & 41 & 2 & 26 & 14 \\
60 & 42 & 3 & 30 & 12 \\
60 & 43 & 2 & 30 & 12 \\
60 & 44 & 2 & 32 & 12 \\
60 & 45 & 3 & 36 & 9 \\
60 & 46 & 4 & 40 & 8 \\
60 & 47 & 3 & 40 & 8 \\
60 & 48 & 3 & 42 & 8 \\
60 & 49 & 2 & 42 & 8 \\
60 & 50 & 3 & 46 & 6 \\
60 & 51 & 3 & 48 & 5 \\
60 & 52 & 5 & 54 & 3 \\
\hline
61 & 30 & 1 & 1 & 61 \\
61 & 31 & 1 & 3 & 34 \\
61 & 32 & 1 & 5 & 31 \\
61 & 33 & 1 & 7 & 29 \\
61 & 34 & 2 & 11 & 24 \\
61 & 35 & 2 & 13 & 24 \\
61 & 36 & 2 & 15 & 22 \\
61 & 37 & 2 & 17 & 20 \\
61 & 38 & 2 & 19 & 18 \\
61 & 39 & 2 & 21 & 16 \\
61 & 40 & 2 & 23 & 16 \\
61 & 41 & 2 & 25 & 16 \\
61 & 42 & 2 & 27 & 14 \\
61 & 43 & 3 & 31 & 12 \\
61 & 44 & 2 & 31 & 12 \\
61 & 45 & 2 & 33 & 12 \\
61 & 46 & 3 & 37 & 9 \\
61 & 47 & 4 & 41 & 8 \\
61 & 48 & 3 & 41 & 8 \\
61 & 49 & 3 & 43 & 8 \\
61 & 50 & 2 & 43 & 8 \\
61 & 51 & 3 & 47 & 6 \\
61 & 52 & 3 & 49 & 5 \\
61 & 53 & 5 & 55 & 3 \\
\hline
62 & 31 & 1 & 2 & 41 \\
62 & 32 & 1 & 4 & 32 \\
62 & 33 & 1 & 6 & 31 \\
62 & 34 & 2 & 10 & 26 \\
62 & 35 & 2 & 12 & 24 \\
62 & 36 & 2 & 14 & 24 \\
62 & 37 & 2 & 16 & 22 \\
62 & 38 & 2 & 18 & 20 \\
62 & 39 & 2 & 20 & 18 \\
62 & 40 & 2 & 22 & 16 \\
62 & 41 & 2 & 24 & 16 \\
62 & 42 & 2 & 26 & 16 \\
62 & 43 & 2 & 28 & 14 \\
62 & 44 & 3 & 32 & 12 \\
62 & 45 & 2 & 32 & 12 \\
62 & 46 & 2 & 34 & 12 \\
62 & 47 & 3 & 38 & 9 \\
62 & 48 & 4 & 42 & 8 \\
62 & 49 & 3 & 42 & 8 \\
62 & 50 & 3 & 44 & 8 \\
62 & 51 & 2 & 44 & 8 \\
62 & 52 & 3 & 48 & 6 \\
62 & 53 & 3 & 50 & 5 \\
62 & 54 & 5 & 56 & 3 \\
\hline
63 & 31 & 1 & 1 & 63 \\
63 & 32 & 1 & 3 & 36 \\
63 & 33 & 1 & 5 & 32 \\
63 & 34 & 1 & 7 & 31 \\
63 & 35 & 2 & 11 & 26 \\
63 & 36 & 2 & 13 & 24 \\
63 & 37 & 2 & 15 & 24 \\
63 & 38 & 2 & 17 & 22 \\
63 & 39 & 2 & 19 & 20 \\
63 & 40 & 2 & 21 & 18 \\
63 & 41 & 2 & 23 & 16 \\
63 & 42 & 2 & 25 & 16 \\
63 & 43 & 2 & 27 & 16 \\
63 & 44 & 2 & 29 & 14 \\
63 & 45 & 3 & 33 & 12 \\
63 & 46 & 2 & 33 & 12 \\
63 & 47 & 2 & 35 & 12 \\
63 & 48 & 3 & 39 & 9 \\
63 & 49 & 4 & 43 & 8 \\
63 & 50 & 3 & 43 & 8 \\
63 & 51 & 3 & 45 & 8 \\
63 & 52 & 2 & 45 & 8 \\
63 & 53 & 3 & 49 & 6 \\
63 & 54 & 3 & 51 & 5 \\
63 & 55 & 5 & 57 & 3 \\
\hline
64 & 32 & 1 & 2 & 42 \\
64 & 33 & 1 & 4 & 33 \\
64 & 34 & 1 & 6 & 32 \\
64 & 35 & 2 & 10 & 28 \\
64 & 36 & 2 & 12 & 25 \\
64 & 37 & 2 & 14 & 24 \\
64 & 38 & 2 & 16 & 24 \\
64 & 39 & 2 & 18 & 22 \\
64 & 40 & 2 & 20 & 19 \\
64 & 41 & 2 & 22 & 18 \\
64 & 42 & 2 & 24 & 16 \\
64 & 43 & 2 & 26 & 16 \\
64 & 44 & 2 & 28 & 16 \\
64 & 45 & 2 & 30 & 14 \\
64 & 46 & 3 & 34 & 12 \\
64 & 47 & 2 & 34 & 12 \\
64 & 48 & 2 & 36 & 12 \\
64 & 49 & 3 & 40 & 9 \\
64 & 50 & 4 & 44 & 8 \\
64 & 51 & 3 & 44 & 8 \\
64 & 52 & 3 & 46 & 8 \\
64 & 53 & 2 & 46 & 8 \\
64 & 54 & 3 & 50 & 6 \\
64 & 55 & 3 & 52 & 5 \\
\hline
65 & 32 & 1 & 1 & 65 \\
65 & 33 & 1 & 3 & 36 \\
65 & 34 & 1 & 5 & 32 \\
65 & 35 & 1 & 7 & 32 \\
65 & 36 & 2 & 11 & 27 \\
65 & 37 & 2 & 13 & 25 \\
65 & 38 & 2 & 15 & 24 \\
65 & 39 & 2 & 17 & 23 \\
65 & 40 & 2 & 19 & 20 \\
65 & 41 & 2 & 21 & 19 \\
65 & 42 & 2 & 23 & 18 \\
65 & 43 & 2 & 25 & 16 \\
65 & 44 & 2 & 27 & 16 \\
65 & 45 & 2 & 29 & 15 \\
65 & 46 & 2 & 31 & 13 \\
65 & 47 & 3 & 35 & 12 \\
65 & 48 & 2 & 35 & 12 \\
65 & 49 & 2 & 37 & 11 \\
65 & 50 & 5 & 45 & 8 \\
65 & 51 & 4 & 45 & 8 \\
65 & 52 & 3 & 45 & 8 \\
65 & 53 & 3 & 47 & 7 \\
65 & 54 & 4 & 51 & 6 \\
65 & 55 & 3 & 51 & 6 \\
65 & 56 & 3 & 53 & 5 \\
\hline
66 & 33 & 1 & 2 & 44 \\
66 & 34 & 1 & 4 & 34 \\
66 & 35 & 1 & 6 & 32 \\
66 & 36 & 2 & 10 & 28 \\
66 & 37 & 2 & 12 & 26 \\
66 & 38 & 2 & 14 & 24 \\
66 & 39 & 2 & 16 & 24 \\
66 & 40 & 2 & 18 & 23 \\
66 & 41 & 2 & 20 & 20 \\
66 & 42 & 2 & 22 & 18 \\
66 & 43 & 2 & 24 & 18 \\
66 & 44 & 2 & 26 & 16 \\
66 & 45 & 2 & 28 & 16 \\
66 & 46 & 2 & 30 & 15 \\
66 & 47 & 2 & 32 & 13 \\
66 & 48 & 3 & 36 & 12 \\
66 & 49 & 2 & 36 & 12 \\
66 & 50 & 2 & 38 & 11 \\
66 & 51 & 5 & 46 & 8 \\
66 & 52 & 4 & 46 & 8 \\
66 & 53 & 3 & 46 & 8 \\
66 & 54 & 3 & 48 & 7 \\
66 & 55 & 4 & 52 & 6 \\
66 & 56 & 3 & 52 & 6 \\
66 & 57 & 5 & 58 & 4 \\
\hline
67 & 33 & 1 & 1 & 67 \\
67 & 34 & 1 & 3 & 38 \\
67 & 35 & 1 & 5 & 33 \\
67 & 36 & 1 & 7 & 32 \\
67 & 37 & 2 & 11 & 28 \\
67 & 38 & 2 & 13 & 26 \\
67 & 39 & 2 & 15 & 24 \\
67 & 40 & 2 & 17 & 24 \\
67 & 41 & 2 & 19 & 20 \\
67 & 42 & 2 & 21 & 20 \\
67 & 43 & 2 & 23 & 18 \\
67 & 44 & 2 & 25 & 16 \\
67 & 45 & 2 & 27 & 16 \\
67 & 46 & 2 & 29 & 16 \\
67 & 47 & 2 & 31 & 15 \\
67 & 48 & 2 & 33 & 13 \\
67 & 49 & 3 & 37 & 12 \\
67 & 50 & 2 & 37 & 12 \\
67 & 51 & 2 & 39 & 11 \\
67 & 52 & 5 & 47 & 8 \\
67 & 53 & 4 & 47 & 8 \\
67 & 54 & 3 & 47 & 8 \\
67 & 55 & 3 & 49 & 7 \\
67 & 56 & 4 & 53 & 6 \\
67 & 57 & 3 & 53 & 6 \\
67 & 58 & 5 & 59 & 4 \\
\hline
68 & 34 & 1 & 2 & 45 \\
68 & 35 & 1 & 4 & 36 \\
68 & 36 & 2 & 8 & 32 \\
68 & 37 & 1 & 8 & 32 \\
68 & 38 & 2 & 12 & 28 \\
68 & 39 & 2 & 14 & 25 \\
68 & 40 & 2 & 16 & 24 \\
68 & 41 & 2 & 18 & 24 \\
68 & 42 & 2 & 20 & 20 \\
68 & 43 & 2 & 22 & 20 \\
68 & 44 & 2 & 24 & 18 \\
68 & 45 & 2 & 26 & 16 \\
68 & 46 & 2 & 28 & 16 \\
68 & 47 & 2 & 30 & 16 \\
68 & 48 & 2 & 32 & 15 \\
68 & 49 & 2 & 34 & 13 \\
68 & 50 & 3 & 38 & 12 \\
68 & 51 & 2 & 38 & 12 \\
68 & 52 & 2 & 40 & 11 \\
68 & 53 & 5 & 48 & 8 \\
68 & 54 & 4 & 48 & 8 \\
68 & 55 & 3 & 48 & 8 \\
68 & 56 & 3 & 50 & 7 \\
68 & 57 & 4 & 54 & 6 \\
68 & 58 & 3 & 54 & 6 \\
68 & 59 & 5 & 60 & 4 \\
\hline
69 & 34 & 1 & 1 & 69 \\
69 & 35 & 1 & 3 & 39 \\
69 & 36 & 1 & 5 & 34 \\
69 & 37 & 2 & 9 & 31 \\
69 & 38 & 2 & 11 & 29 \\
69 & 39 & 2 & 13 & 27 \\
69 & 40 & 2 & 15 & 25 \\
69 & 41 & 2 & 17 & 24 \\
69 & 42 & 2 & 19 & 22 \\
69 & 43 & 2 & 21 & 20 \\
69 & 44 & 2 & 23 & 20 \\
69 & 45 & 2 & 25 & 18 \\
69 & 46 & 2 & 27 & 16 \\
69 & 47 & 2 & 29 & 16 \\
69 & 48 & 2 & 31 & 16 \\
69 & 49 & 2 & 33 & 15 \\
69 & 50 & 2 & 35 & 13 \\
69 & 51 & 3 & 39 & 12 \\
69 & 52 & 2 & 39 & 12 \\
69 & 53 & 2 & 41 & 11 \\
69 & 54 & 5 & 49 & 8 \\
69 & 55 & 4 & 49 & 8 \\
69 & 56 & 3 & 49 & 8 \\
69 & 57 & 5 & 55 & 6 \\
69 & 58 & 4 & 55 & 6 \\
69 & 59 & 3 & 55 & 6 \\
69 & 60 & 5 & 61 & 4 \\
\hline
70 & 35 & 1 & 2 & 46 \\
70 & 36 & 1 & 4 & 36 \\
70 & 37 & 1 & 6 & 34 \\
70 & 38 & 2 & 10 & 31 \\
70 & 39 & 2 & 12 & 28 \\
70 & 40 & 2 & 14 & 26 \\
70 & 41 & 2 & 16 & 25 \\
70 & 42 & 2 & 18 & 24 \\
70 & 43 & 2 & 20 & 21 \\
70 & 44 & 2 & 22 & 20 \\
70 & 45 & 2 & 24 & 20 \\
70 & 46 & 2 & 26 & 18 \\
70 & 47 & 2 & 28 & 16 \\
70 & 48 & 2 & 30 & 16 \\
70 & 49 & 2 & 32 & 16 \\
70 & 50 & 2 & 34 & 15 \\
70 & 51 & 2 & 36 & 13 \\
70 & 52 & 3 & 40 & 12 \\
70 & 53 & 2 & 40 & 12 \\
70 & 54 & 3 & 44 & 10 \\
70 & 55 & 5 & 50 & 8 \\
70 & 56 & 4 & 50 & 8 \\
70 & 57 & 3 & 50 & 8 \\
70 & 58 & 5 & 56 & 6 \\
70 & 59 & 4 & 56 & 6 \\
70 & 60 & 3 & 56 & 6 \\
70 & 61 & 5 & 62 & 4 \\
\hline
71 & 35 & 1 & 1 & 71 \\
71 & 36 & 1 & 3 & 40 \\
71 & 37 & 1 & 5 & 36 \\
71 & 38 & 1 & 7 & 34 \\
71 & 39 & 2 & 11 & 31 \\
71 & 40 & 2 & 13 & 28 \\
71 & 41 & 2 & 15 & 26 \\
71 & 42 & 2 & 17 & 24 \\
71 & 43 & 2 & 19 & 23 \\
71 & 44 & 2 & 21 & 21 \\
71 & 45 & 2 & 23 & 20 \\
71 & 46 & 2 & 25 & 20 \\
71 & 47 & 2 & 27 & 18 \\
71 & 48 & 2 & 29 & 16 \\
71 & 49 & 2 & 31 & 16 \\
71 & 50 & 2 & 33 & 16 \\
71 & 51 & 2 & 35 & 15 \\
71 & 52 & 3 & 39 & 12 \\
71 & 53 & 3 & 41 & 12 \\
71 & 54 & 2 & 41 & 12 \\
71 & 55 & 3 & 45 & 10 \\
71 & 56 & 5 & 51 & 8 \\
71 & 57 & 4 & 51 & 8 \\
71 & 58 & 3 & 51 & 8 \\
71 & 59 & 5 & 57 & 6 \\
71 & 60 & 4 & 57 & 6 \\
71 & 61 & 3 & 57 & 6 \\
71 & 62 & 5 & 63 & 4 \\
\hline
72 & 36 & 1 & 2 & 48 \\
72 & 37 & 1 & 4 & 38 \\
72 & 38 & 1 & 6 & 35 \\
72 & 39 & 2 & 10 & 32 \\
72 & 40 & 2 & 12 & 30 \\
72 & 41 & 2 & 14 & 28 \\
72 & 42 & 2 & 16 & 26 \\
72 & 43 & 2 & 18 & 24 \\
72 & 44 & 2 & 20 & 22 \\
72 & 45 & 2 & 22 & 20 \\
72 & 46 & 2 & 24 & 20 \\
72 & 47 & 2 & 26 & 20 \\
72 & 48 & 2 & 28 & 18 \\
72 & 49 & 2 & 30 & 16 \\
72 & 50 & 2 & 32 & 16 \\
72 & 51 & 2 & 34 & 16 \\
72 & 52 & 2 & 36 & 15 \\
72 & 53 & 3 & 40 & 12 \\
72 & 54 & 3 & 42 & 12 \\
72 & 55 & 2 & 42 & 12 \\
72 & 56 & 3 & 46 & 10 \\
72 & 57 & 5 & 52 & 8 \\
72 & 58 & 4 & 52 & 8 \\
72 & 59 & 3 & 52 & 8 \\
72 & 60 & 5 & 58 & 6 \\
72 & 61 & 4 & 58 & 6 \\
72 & 62 & 3 & 58 & 6 \\
72 & 63 & 5 & 64 & 4 \\
\hline
73 & 36 & 1 & 1 & 73 \\
73 & 37 & 1 & 3 & 41 \\
73 & 38 & 1 & 5 & 36 \\
73 & 39 & 2 & 9 & 32 \\
73 & 40 & 2 & 11 & 32 \\
73 & 41 & 2 & 13 & 29 \\
73 & 42 & 2 & 15 & 28 \\
73 & 43 & 2 & 17 & 26 \\
73 & 44 & 2 & 19 & 24 \\
73 & 45 & 2 & 21 & 22 \\
73 & 46 & 2 & 23 & 20 \\
73 & 47 & 2 & 25 & 20 \\
73 & 48 & 2 & 27 & 20 \\
73 & 49 & 2 & 29 & 17 \\
73 & 50 & 2 & 31 & 16 \\
73 & 51 & 2 & 33 & 16 \\
73 & 52 & 2 & 35 & 16 \\
73 & 53 & 2 & 37 & 14 \\
73 & 54 & 3 & 41 & 12 \\
73 & 55 & 4 & 45 & 10 \\
73 & 56 & 3 & 45 & 10 \\
73 & 57 & 3 & 47 & 10 \\
73 & 58 & 5 & 53 & 8 \\
73 & 59 & 4 & 53 & 8 \\
73 & 60 & 3 & 53 & 8 \\
73 & 61 & 5 & 59 & 6 \\
73 & 62 & 4 & 59 & 6 \\
73 & 63 & 3 & 59 & 6 \\
73 & 64 & 5 & 65 & 4 \\
\hline
74 & 37 & 1 & 2 & 49 \\
74 & 38 & 1 & 4 & 39 \\
74 & 39 & 1 & 6 & 36 \\
74 & 40 & 2 & 10 & 32 \\
74 & 41 & 2 & 12 & 32 \\
74 & 42 & 2 & 14 & 28 \\
74 & 43 & 2 & 16 & 28 \\
74 & 44 & 2 & 18 & 25 \\
74 & 45 & 2 & 20 & 24 \\
74 & 46 & 2 & 22 & 22 \\
74 & 47 & 2 & 24 & 20 \\
74 & 48 & 2 & 26 & 20 \\
74 & 49 & 2 & 28 & 18 \\
74 & 50 & 2 & 30 & 17 \\
74 & 51 & 2 & 32 & 16 \\
74 & 52 & 2 & 34 & 16 \\
74 & 53 & 2 & 36 & 16 \\
74 & 54 & 2 & 38 & 14 \\
74 & 55 & 3 & 42 & 12 \\
74 & 56 & 4 & 46 & 10 \\
74 & 57 & 3 & 46 & 10 \\
74 & 58 & 3 & 48 & 9 \\
74 & 59 & 5 & 54 & 8 \\
74 & 60 & 4 & 54 & 8 \\
74 & 61 & 3 & 54 & 8 \\
74 & 62 & 5 & 60 & 6 \\
74 & 63 & 4 & 60 & 6 \\
74 & 64 & 3 & 60 & 6 \\
74 & 65 & 5 & 66 & 4 \\
\hline
75 & 37 & 1 & 1 & 75 \\
75 & 38 & 1 & 3 & 42 \\
75 & 39 & 1 & 5 & 38 \\
75 & 40 & 1 & 7 & 36 \\
75 & 41 & 2 & 11 & 32 \\
75 & 42 & 2 & 13 & 30 \\
75 & 43 & 2 & 15 & 28 \\
75 & 44 & 2 & 17 & 26 \\
75 & 45 & 2 & 19 & 24 \\
75 & 46 & 2 & 21 & 24 \\
75 & 47 & 2 & 23 & 22 \\
75 & 48 & 2 & 25 & 20 \\
75 & 49 & 2 & 27 & 20 \\
75 & 50 & 2 & 29 & 18 \\
75 & 51 & 2 & 31 & 17 \\
75 & 52 & 2 & 33 & 16 \\
75 & 53 & 2 & 35 & 16 \\
75 & 54 & 3 & 39 & 13 \\
75 & 55 & 3 & 41 & 12 \\
75 & 56 & 3 & 43 & 12 \\
75 & 57 & 4 & 47 & 10 \\
75 & 58 & 3 & 47 & 10 \\
75 & 59 & 3 & 49 & 9 \\
75 & 60 & 5 & 55 & 8 \\
75 & 61 & 4 & 55 & 8 \\
75 & 62 & 3 & 55 & 8 \\
75 & 63 & 5 & 61 & 6 \\
75 & 64 & 4 & 61 & 6 \\
75 & 65 & 3 & 61 & 6 \\
75 & 66 & 5 & 67 & 4 \\
\hline
76 & 38 & 1 & 2 & 50 \\
76 & 39 & 1 & 4 & 40 \\
76 & 40 & 1 & 6 & 37 \\
76 & 41 & 2 & 10 & 32 \\
76 & 42 & 2 & 12 & 32 \\
76 & 43 & 2 & 14 & 30 \\
76 & 44 & 2 & 16 & 28 \\
76 & 45 & 2 & 18 & 26 \\
76 & 46 & 2 & 20 & 24 \\
76 & 47 & 2 & 22 & 23 \\
76 & 48 & 2 & 24 & 21 \\
76 & 49 & 2 & 26 & 20 \\
76 & 50 & 2 & 28 & 20 \\
76 & 51 & 2 & 30 & 18 \\
76 & 52 & 2 & 32 & 17 \\
76 & 53 & 2 & 34 & 16 \\
76 & 54 & 2 & 36 & 16 \\
76 & 55 & 4 & 42 & 12 \\
76 & 56 & 3 & 42 & 12 \\
76 & 57 & 3 & 44 & 11 \\
76 & 58 & 4 & 48 & 10 \\
76 & 59 & 3 & 48 & 10 \\
76 & 60 & 3 & 50 & 9 \\
76 & 61 & 5 & 56 & 8 \\
76 & 62 & 4 & 56 & 8 \\
76 & 63 & 3 & 56 & 8 \\
76 & 64 & 5 & 62 & 6 \\
76 & 65 & 4 & 62 & 6 \\
76 & 66 & 3 & 62 & 6 \\
76 & 67 & 5 & 68 & 4 \\
\hline
77 & 38 & 1 & 1 & 77 \\
77 & 39 & 1 & 3 & 44 \\
77 & 40 & 1 & 5 & 39 \\
77 & 41 & 2 & 9 & 34 \\
77 & 42 & 2 & 11 & 32 \\
77 & 43 & 2 & 13 & 32 \\
77 & 44 & 2 & 15 & 29 \\
77 & 45 & 2 & 17 & 28 \\
77 & 46 & 2 & 19 & 24 \\
77 & 47 & 2 & 21 & 24 \\
77 & 48 & 2 & 23 & 23 \\
77 & 49 & 2 & 25 & 21 \\
77 & 50 & 2 & 27 & 20 \\
77 & 51 & 2 & 29 & 18 \\
77 & 52 & 2 & 31 & 18 \\
77 & 53 & 2 & 33 & 17 \\
77 & 54 & 2 & 35 & 16 \\
77 & 55 & 2 & 37 & 16 \\
77 & 56 & 4 & 43 & 12 \\
77 & 57 & 3 & 43 & 12 \\
77 & 58 & 3 & 45 & 11 \\
77 & 59 & 4 & 49 & 10 \\
77 & 60 & 3 & 49 & 10 \\
77 & 61 & 3 & 51 & 9 \\
77 & 62 & 5 & 57 & 8 \\
77 & 63 & 4 & 57 & 8 \\
77 & 64 & 3 & 57 & 8 \\
77 & 65 & 5 & 63 & 6 \\
77 & 66 & 4 & 63 & 6 \\
77 & 67 & 3 & 63 & 6 \\
77 & 68 & 5 & 69 & 4 \\
\hline
78 & 39 & 1 & 2 & 52 \\
78 & 40 & 1 & 4 & 40 \\
78 & 41 & 1 & 6 & 38 \\
78 & 42 & 2 & 10 & 34 \\
78 & 43 & 2 & 12 & 32 \\
78 & 44 & 2 & 14 & 31 \\
78 & 45 & 2 & 16 & 29 \\
78 & 46 & 2 & 18 & 27 \\
78 & 47 & 2 & 20 & 24 \\
78 & 48 & 2 & 22 & 24 \\
78 & 49 & 2 & 24 & 22 \\
78 & 50 & 2 & 26 & 20 \\
78 & 51 & 2 & 28 & 20 \\
78 & 52 & 2 & 30 & 18 \\
78 & 53 & 2 & 32 & 18 \\
78 & 54 & 2 & 34 & 17 \\
78 & 55 & 2 & 36 & 16 \\
78 & 56 & 2 & 38 & 16 \\
78 & 57 & 4 & 44 & 12 \\
78 & 58 & 3 & 44 & 12 \\
78 & 59 & 3 & 46 & 11 \\
78 & 60 & 4 & 50 & 10 \\
78 & 61 & 3 & 50 & 10 \\
78 & 62 & 6 & 58 & 8 \\
78 & 63 & 5 & 58 & 8 \\
78 & 64 & 4 & 58 & 8 \\
78 & 65 & 3 & 58 & 8 \\
78 & 66 & 5 & 64 & 6 \\
78 & 67 & 4 & 64 & 6 \\
78 & 68 & 3 & 64 & 6 \\
78 & 69 & 5 & 70 & 4 \\
\hline
79 & 39 & 1 & 1 & 79 \\
79 & 40 & 1 & 3 & 44 \\
79 & 41 & 1 & 5 & 40 \\
79 & 42 & 1 & 7 & 38 \\
79 & 43 & 2 & 11 & 33 \\
79 & 44 & 2 & 13 & 32 \\
79 & 45 & 2 & 15 & 31 \\
79 & 46 & 2 & 17 & 28 \\
79 & 47 & 2 & 19 & 26 \\
79 & 48 & 2 & 21 & 24 \\
79 & 49 & 2 & 23 & 24 \\
79 & 50 & 2 & 25 & 22 \\
79 & 51 & 2 & 27 & 20 \\
79 & 52 & 2 & 29 & 20 \\
79 & 53 & 2 & 31 & 18 \\
79 & 54 & 2 & 33 & 18 \\
79 & 55 & 2 & 35 & 17 \\
79 & 56 & 2 & 37 & 16 \\
79 & 57 & 2 & 39 & 16 \\
79 & 58 & 4 & 45 & 12 \\
79 & 59 & 3 & 45 & 12 \\
79 & 60 & 4 & 49 & 10 \\
79 & 61 & 4 & 51 & 10 \\
79 & 62 & 3 & 51 & 10 \\
79 & 63 & 6 & 59 & 8 \\
79 & 64 & 5 & 59 & 8 \\
79 & 65 & 4 & 59 & 8 \\
79 & 66 & 3 & 59 & 8 \\
79 & 67 & 5 & 65 & 6 \\
79 & 68 & 4 & 65 & 6 \\
79 & 69 & 3 & 65 & 6 \\
79 & 70 & 5 & 71 & 4 \\
\hline
80 & 40 & 1 & 2 & 53 \\
80 & 41 & 1 & 4 & 42 \\
80 & 42 & 1 & 6 & 40 \\
80 & 43 & 2 & 10 & 35 \\
80 & 44 & 2 & 12 & 32 \\
80 & 45 & 2 & 14 & 32 \\
80 & 46 & 2 & 16 & 30 \\
80 & 47 & 2 & 18 & 28 \\
80 & 48 & 2 & 20 & 25 \\
80 & 49 & 2 & 22 & 24 \\
80 & 50 & 2 & 24 & 24 \\
80 & 51 & 2 & 26 & 22 \\
80 & 52 & 2 & 28 & 20 \\
80 & 53 & 2 & 30 & 20 \\
80 & 54 & 2 & 32 & 18 \\
80 & 55 & 2 & 34 & 18 \\
80 & 56 & 2 & 36 & 17 \\
80 & 57 & 2 & 38 & 16 \\
80 & 58 & 2 & 40 & 16 \\
80 & 59 & 4 & 46 & 12 \\
80 & 60 & 3 & 46 & 12 \\
80 & 61 & 4 & 50 & 10 \\
80 & 62 & 4 & 52 & 10 \\
80 & 63 & 3 & 52 & 10 \\
80 & 64 & 6 & 60 & 8 \\
80 & 65 & 5 & 60 & 8 \\
80 & 66 & 4 & 60 & 8 \\
80 & 67 & 3 & 60 & 8 \\
80 & 68 & 5 & 66 & 6 \\
80 & 69 & 4 & 66 & 6 \\
80 & 70 & 3 & 66 & 6 \\
80 & 71 & 5 & 72 & 4 \\
\hline
81 & 40 & 1 & 1 & 81 \\
81 & 41 & 1 & 3 & 46 \\
81 & 42 & 1 & 5 & 40 \\
81 & 43 & 1 & 7 & 39 \\
81 & 44 & 2 & 11 & 34 \\
81 & 45 & 2 & 13 & 32 \\
81 & 46 & 2 & 15 & 32 \\
81 & 47 & 2 & 17 & 28 \\
81 & 48 & 2 & 19 & 26 \\
81 & 49 & 2 & 21 & 24 \\
81 & 50 & 2 & 23 & 24 \\
81 & 51 & 2 & 25 & 24 \\
81 & 52 & 2 & 27 & 22 \\
81 & 53 & 2 & 29 & 20 \\
81 & 54 & 2 & 31 & 20 \\
81 & 55 & 2 & 33 & 18 \\
81 & 56 & 2 & 35 & 18 \\
81 & 57 & 2 & 37 & 17 \\
81 & 58 & 2 & 39 & 16 \\
81 & 59 & 4 & 45 & 12 \\
81 & 60 & 4 & 47 & 12 \\
81 & 61 & 3 & 47 & 12 \\
81 & 62 & 4 & 51 & 10 \\
81 & 63 & 4 & 53 & 10 \\
81 & 64 & 3 & 53 & 10 \\
81 & 65 & 6 & 61 & 8 \\
81 & 66 & 5 & 61 & 8 \\
81 & 67 & 4 & 61 & 8 \\
81 & 68 & 3 & 61 & 8 \\
81 & 69 & 5 & 67 & 6 \\
81 & 70 & 4 & 67 & 6 \\
81 & 71 & 3 & 67 & 6 \\
81 & 72 & 5 & 73 & 4 \\
\hline
82 & 41 & 1 & 2 & 54 \\
82 & 42 & 1 & 4 & 43 \\
82 & 43 & 1 & 6 & 40 \\
82 & 44 & 2 & 10 & 36 \\
82 & 45 & 2 & 12 & 33 \\
82 & 46 & 2 & 14 & 32 \\
82 & 47 & 2 & 16 & 32 \\
82 & 48 & 2 & 18 & 28 \\
82 & 49 & 2 & 20 & 26 \\
82 & 50 & 2 & 22 & 24 \\
82 & 51 & 2 & 24 & 24 \\
82 & 52 & 2 & 26 & 23 \\
82 & 53 & 2 & 28 & 21 \\
82 & 54 & 2 & 30 & 20 \\
82 & 55 & 2 & 32 & 19 \\
82 & 56 & 2 & 34 & 18 \\
82 & 57 & 2 & 36 & 18 \\
82 & 58 & 2 & 38 & 17 \\
82 & 59 & 2 & 40 & 16 \\
82 & 60 & 4 & 46 & 12 \\
82 & 61 & 4 & 48 & 12 \\
82 & 62 & 3 & 48 & 12 \\
82 & 63 & 4 & 52 & 10 \\
82 & 64 & 4 & 54 & 10 \\
82 & 65 & 3 & 54 & 10 \\
82 & 66 & 6 & 62 & 8 \\
82 & 67 & 5 & 62 & 8 \\
82 & 68 & 4 & 62 & 8 \\
82 & 69 & 3 & 62 & 8 \\
82 & 70 & 5 & 68 & 6 \\
82 & 71 & 4 & 68 & 6 \\
82 & 72 & 3 & 68 & 6 \\
82 & 73 & 5 & 74 & 4 \\
\hline
83 & 41 & 1 & 1 & 83 \\
83 & 42 & 1 & 3 & 47 \\
83 & 43 & 1 & 5 & 42 \\
83 & 44 & 1 & 7 & 40 \\
83 & 45 & 2 & 11 & 36 \\
83 & 46 & 2 & 13 & 32 \\
83 & 47 & 2 & 15 & 32 \\
83 & 48 & 2 & 17 & 29 \\
83 & 49 & 2 & 19 & 28 \\
83 & 50 & 2 & 21 & 26 \\
83 & 51 & 2 & 23 & 24 \\
83 & 52 & 2 & 25 & 24 \\
83 & 53 & 2 & 27 & 23 \\
83 & 54 & 2 & 29 & 21 \\
83 & 55 & 2 & 31 & 20 \\
83 & 56 & 2 & 33 & 19 \\
83 & 57 & 2 & 35 & 18 \\
83 & 58 & 2 & 37 & 18 \\
83 & 59 & 2 & 39 & 17 \\
83 & 60 & 3 & 43 & 14 \\
83 & 61 & 4 & 47 & 12 \\
83 & 62 & 4 & 49 & 12 \\
83 & 63 & 3 & 49 & 12 \\
83 & 64 & 4 & 53 & 10 \\
83 & 65 & 4 & 55 & 10 \\
83 & 66 & 3 & 55 & 10 \\
83 & 67 & 6 & 63 & 8 \\
83 & 68 & 5 & 63 & 8 \\
83 & 69 & 4 & 63 & 8 \\
83 & 70 & 3 & 63 & 8 \\
83 & 71 & 5 & 69 & 5 \\
83 & 72 & 7 & 75 & 4 \\
83 & 73 & 6 & 75 & 4 \\
83 & 74 & 5 & 75 & 4 \\
\hline
84 & 42 & 1 & 2 & 56 \\
84 & 43 & 1 & 4 & 44 \\
84 & 44 & 1 & 6 & 41 \\
84 & 45 & 1 & 8 & 40 \\
84 & 46 & 2 & 12 & 35 \\
84 & 47 & 2 & 14 & 32 \\
84 & 48 & 2 & 16 & 32 \\
84 & 49 & 2 & 18 & 28 \\
84 & 50 & 2 & 20 & 28 \\
84 & 51 & 2 & 22 & 25 \\
84 & 52 & 2 & 24 & 24 \\
84 & 53 & 2 & 26 & 24 \\
84 & 54 & 2 & 28 & 22 \\
84 & 55 & 2 & 30 & 20 \\
84 & 56 & 2 & 32 & 20 \\
84 & 57 & 2 & 34 & 19 \\
84 & 58 & 2 & 36 & 18 \\
84 & 59 & 2 & 38 & 18 \\
84 & 60 & 2 & 40 & 17 \\
84 & 61 & 3 & 44 & 14 \\
84 & 62 & 4 & 48 & 12 \\
84 & 63 & 4 & 50 & 12 \\
84 & 64 & 3 & 50 & 12 \\
84 & 65 & 4 & 54 & 10 \\
84 & 66 & 4 & 56 & 10 \\
84 & 67 & 3 & 56 & 10 \\
84 & 68 & 6 & 64 & 8 \\
84 & 69 & 5 & 64 & 8 \\
84 & 70 & 4 & 64 & 8 \\
84 & 71 & 3 & 64 & 8 \\
84 & 72 & 5 & 70 & 5 \\
84 & 73 & 7 & 76 & 4 \\
84 & 74 & 6 & 76 & 4 \\
84 & 75 & 5 & 76 & 4 \\
\hline
85 & 42 & 1 & 1 & 85 \\
85 & 43 & 1 & 3 & 48 \\
85 & 44 & 1 & 5 & 43 \\
85 & 45 & 2 & 9 & 39 \\
85 & 46 & 2 & 11 & 36 \\
85 & 47 & 2 & 13 & 34 \\
85 & 48 & 2 & 15 & 32 \\
85 & 49 & 2 & 17 & 31 \\
85 & 50 & 2 & 19 & 28 \\
85 & 51 & 2 & 21 & 28 \\
85 & 52 & 2 & 23 & 24 \\
85 & 53 & 2 & 25 & 24 \\
85 & 54 & 2 & 27 & 24 \\
85 & 55 & 2 & 29 & 22 \\
85 & 56 & 2 & 31 & 20 \\
85 & 57 & 2 & 33 & 20 \\
85 & 58 & 2 & 35 & 19 \\
85 & 59 & 2 & 37 & 18 \\
85 & 60 & 2 & 39 & 18 \\
85 & 61 & 2 & 41 & 17 \\
85 & 62 & 3 & 45 & 14 \\
85 & 63 & 4 & 49 & 12 \\
85 & 64 & 4 & 51 & 12 \\
85 & 65 & 3 & 51 & 12 \\
85 & 66 & 4 & 55 & 10 \\
85 & 67 & 4 & 57 & 10 \\
85 & 68 & 3 & 57 & 10 \\
85 & 69 & 6 & 65 & 8 \\
85 & 70 & 5 & 65 & 8 \\
85 & 71 & 4 & 65 & 8 \\
85 & 72 & 3 & 65 & 8 \\
85 & 73 & 5 & 71 & 5 \\
85 & 74 & 7 & 77 & 4 \\
85 & 75 & 6 & 77 & 4 \\
85 & 76 & 5 & 77 & 4 \\
\hline
86 & 43 & 1 & 2 & 57 \\
86 & 44 & 1 & 4 & 45 \\
86 & 45 & 1 & 6 & 42 \\
86 & 46 & 2 & 10 & 39 \\
86 & 47 & 2 & 12 & 36 \\
86 & 48 & 2 & 14 & 33 \\
86 & 49 & 2 & 16 & 32 \\
86 & 50 & 2 & 18 & 29 \\
86 & 51 & 2 & 20 & 28 \\
86 & 52 & 2 & 22 & 26 \\
86 & 53 & 2 & 24 & 24 \\
86 & 54 & 2 & 26 & 24 \\
86 & 55 & 2 & 28 & 24 \\
86 & 56 & 2 & 30 & 22 \\
86 & 57 & 2 & 32 & 20 \\
86 & 58 & 2 & 34 & 20 \\
86 & 59 & 2 & 36 & 19 \\
86 & 60 & 2 & 38 & 18 \\
86 & 61 & 2 & 40 & 18 \\
86 & 62 & 2 & 42 & 17 \\
86 & 63 & 3 & 46 & 14 \\
86 & 64 & 4 & 50 & 12 \\
86 & 65 & 4 & 52 & 12 \\
86 & 66 & 3 & 52 & 12 \\
86 & 67 & 4 & 56 & 10 \\
86 & 68 & 4 & 58 & 10 \\
86 & 69 & 3 & 58 & 10 \\
86 & 70 & 6 & 66 & 8 \\
86 & 71 & 5 & 66 & 8 \\
86 & 72 & 4 & 66 & 8 \\
86 & 73 & 3 & 66 & 8 \\
86 & 74 & 5 & 72 & 5 \\
86 & 75 & 7 & 78 & 4 \\
86 & 76 & 6 & 78 & 4 \\
86 & 77 & 5 & 78 & 4 \\
\hline
87 & 43 & 1 & 1 & 87 \\
87 & 44 & 1 & 3 & 49 \\
87 & 45 & 1 & 5 & 44 \\
87 & 46 & 1 & 7 & 42 \\
87 & 47 & 2 & 11 & 38 \\
87 & 48 & 2 & 13 & 35 \\
87 & 49 & 2 & 15 & 33 \\
87 & 50 & 2 & 17 & 32 \\
87 & 51 & 2 & 19 & 28 \\
87 & 52 & 2 & 21 & 28 \\
87 & 53 & 2 & 23 & 26 \\
87 & 54 & 2 & 25 & 24 \\
87 & 55 & 2 & 27 & 24 \\
87 & 56 & 2 & 29 & 24 \\
87 & 57 & 2 & 31 & 22 \\
87 & 58 & 2 & 33 & 20 \\
87 & 59 & 2 & 35 & 20 \\
87 & 60 & 3 & 39 & 18 \\
87 & 61 & 2 & 39 & 18 \\
87 & 62 & 2 & 41 & 18 \\
87 & 63 & 2 & 43 & 17 \\
87 & 64 & 3 & 47 & 14 \\
87 & 65 & 4 & 51 & 12 \\
87 & 66 & 4 & 53 & 12 \\
87 & 67 & 3 & 53 & 12 \\
87 & 68 & 4 & 57 & 10 \\
87 & 69 & 4 & 59 & 10 \\
87 & 70 & 3 & 59 & 10 \\
87 & 71 & 6 & 67 & 8 \\
87 & 72 & 5 & 67 & 8 \\
87 & 73 & 4 & 67 & 8 \\
87 & 74 & 3 & 67 & 8 \\
87 & 75 & 5 & 73 & 5 \\
87 & 76 & 7 & 79 & 4 \\
87 & 77 & 6 & 79 & 4 \\
87 & 78 & 5 & 79 & 4 \\
\hline
88 & 44 & 1 & 2 & 58 \\
88 & 45 & 1 & 4 & 46 \\
88 & 46 & 1 & 6 & 44 \\
88 & 47 & 2 & 10 & 40 \\
88 & 48 & 2 & 12 & 37 \\
88 & 49 & 2 & 14 & 34 \\
88 & 50 & 2 & 16 & 32 \\
88 & 51 & 2 & 18 & 30 \\
88 & 52 & 2 & 20 & 28 \\
88 & 53 & 2 & 22 & 28 \\
88 & 54 & 2 & 24 & 25 \\
88 & 55 & 2 & 26 & 24 \\
88 & 56 & 2 & 28 & 24 \\
88 & 57 & 2 & 30 & 23 \\
88 & 58 & 3 & 34 & 20 \\
88 & 59 & 2 & 34 & 20 \\
88 & 60 & 2 & 36 & 20 \\
88 & 61 & 3 & 40 & 18 \\
88 & 62 & 2 & 40 & 18 \\
88 & 63 & 2 & 42 & 18 \\
88 & 64 & 2 & 44 & 17 \\
88 & 65 & 3 & 48 & 14 \\
88 & 66 & 4 & 52 & 12 \\
88 & 67 & 4 & 54 & 12 \\
88 & 68 & 3 & 54 & 12 \\
88 & 69 & 4 & 58 & 10 \\
88 & 70 & 4 & 60 & 10 \\
88 & 71 & 3 & 60 & 10 \\
88 & 72 & 6 & 68 & 8 \\
88 & 73 & 5 & 68 & 8 \\
88 & 74 & 4 & 68 & 8 \\
88 & 75 & 3 & 68 & 8 \\
88 & 76 & 5 & 74 & 5 \\
88 & 77 & 7 & 80 & 4 \\
88 & 78 & 6 & 80 & 4 \\
88 & 79 & 5 & 80 & 4 \\
\hline
89 & 44 & 1 & 1 & 89 \\
89 & 45 & 1 & 3 & 50 \\
89 & 46 & 1 & 5 & 45 \\
89 & 47 & 1 & 7 & 43 \\
89 & 48 & 2 & 11 & 40 \\
89 & 49 & 2 & 13 & 36 \\
89 & 50 & 2 & 15 & 34 \\
89 & 51 & 2 & 17 & 32 \\
89 & 52 & 2 & 19 & 30 \\
89 & 53 & 2 & 21 & 28 \\
89 & 54 & 2 & 23 & 28 \\
89 & 55 & 2 & 25 & 25 \\
89 & 56 & 2 & 27 & 24 \\
89 & 57 & 2 & 29 & 24 \\
89 & 58 & 2 & 31 & 23 \\
89 & 59 & 3 & 35 & 20 \\
89 & 60 & 2 & 35 & 20 \\
89 & 61 & 2 & 37 & 20 \\
89 & 62 & 3 & 41 & 18 \\
89 & 63 & 2 & 41 & 18 \\
89 & 64 & 2 & 43 & 18 \\
89 & 65 & 2 & 45 & 17 \\
89 & 66 & 3 & 49 & 14 \\
89 & 67 & 4 & 53 & 12 \\
89 & 68 & 4 & 55 & 12 \\
89 & 69 & 3 & 55 & 12 \\
89 & 70 & 4 & 59 & 10 \\
89 & 71 & 4 & 61 & 10 \\
89 & 72 & 3 & 61 & 10 \\
89 & 73 & 6 & 69 & 8 \\
89 & 74 & 5 & 69 & 8 \\
89 & 75 & 4 & 69 & 8 \\
89 & 76 & 3 & 69 & 8 \\
89 & 77 & 5 & 75 & 5 \\
89 & 78 & 7 & 81 & 4 \\
89 & 79 & 6 & 81 & 4 \\
89 & 80 & 5 & 81 & 4 \\
\hline
90 & 45 & 1 & 2 & 60 \\
90 & 46 & 1 & 4 & 48 \\
90 & 47 & 1 & 6 & 44 \\
90 & 48 & 2 & 10 & 40 \\
90 & 49 & 2 & 12 & 39 \\
90 & 50 & 2 & 14 & 36 \\
90 & 51 & 2 & 16 & 33 \\
90 & 52 & 2 & 18 & 32 \\
90 & 53 & 2 & 20 & 30 \\
90 & 54 & 2 & 22 & 28 \\
90 & 55 & 2 & 24 & 26 \\
90 & 56 & 2 & 26 & 24 \\
90 & 57 & 2 & 28 & 24 \\
90 & 58 & 2 & 30 & 24 \\
90 & 59 & 2 & 32 & 22 \\
90 & 60 & 3 & 36 & 20 \\
90 & 61 & 2 & 36 & 20 \\
90 & 62 & 2 & 38 & 20 \\
90 & 63 & 3 & 42 & 18 \\
90 & 64 & 2 & 42 & 18 \\
90 & 65 & 2 & 44 & 18 \\
90 & 66 & 3 & 48 & 14 \\
90 & 67 & 3 & 50 & 14 \\
90 & 68 & 4 & 54 & 12 \\
90 & 69 & 4 & 56 & 12 \\
90 & 70 & 3 & 56 & 12 \\
90 & 71 & 4 & 60 & 10 \\
90 & 72 & 4 & 62 & 10 \\
90 & 73 & 3 & 62 & 10 \\
90 & 74 & 6 & 70 & 7 \\
90 & 75 & 7 & 74 & 6 \\
90 & 76 & 6 & 74 & 6 \\
90 & 77 & 5 & 74 & 6 \\
90 & 78 & 5 & 76 & 5 \\
90 & 79 & 7 & 82 & 4 \\
90 & 80 & 6 & 82 & 4 \\
90 & 81 & 5 & 82 & 4 \\
\hline
91 & 45 & 1 & 1 & 91 \\
91 & 46 & 1 & 3 & 52 \\
91 & 47 & 1 & 5 & 46 \\
91 & 48 & 1 & 7 & 44 \\
91 & 49 & 2 & 11 & 40 \\
91 & 50 & 2 & 13 & 37 \\
91 & 51 & 2 & 15 & 36 \\
91 & 52 & 2 & 17 & 32 \\
91 & 53 & 2 & 19 & 31 \\
91 & 54 & 2 & 21 & 29 \\
91 & 55 & 2 & 23 & 28 \\
91 & 56 & 2 & 25 & 26 \\
91 & 57 & 2 & 27 & 24 \\
91 & 58 & 2 & 29 & 24 \\
91 & 59 & 2 & 31 & 24 \\
91 & 60 & 2 & 33 & 22 \\
91 & 61 & 3 & 37 & 20 \\
91 & 62 & 2 & 37 & 20 \\
91 & 63 & 2 & 39 & 20 \\
91 & 64 & 3 & 43 & 18 \\
91 & 65 & 2 & 43 & 18 \\
91 & 66 & 2 & 45 & 18 \\
91 & 67 & 3 & 49 & 14 \\
91 & 68 & 3 & 51 & 14 \\
91 & 69 & 4 & 55 & 12 \\
91 & 70 & 4 & 57 & 11 \\
91 & 71 & 5 & 61 & 10 \\
91 & 72 & 4 & 61 & 10 \\
91 & 73 & 4 & 63 & 10 \\
91 & 74 & 3 & 63 & 10 \\
91 & 75 & 6 & 71 & 7 \\
91 & 76 & 7 & 75 & 6 \\
91 & 77 & 6 & 75 & 6 \\
91 & 78 & 5 & 75 & 6 \\
91 & 79 & 5 & 77 & 5 \\
91 & 80 & 7 & 83 & 4 \\
91 & 81 & 6 & 83 & 4 \\
91 & 82 & 5 & 83 & 4 \\
\hline
92 & 46 & 1 & 2 & 61 \\
92 & 47 & 1 & 4 & 48 \\
92 & 48 & 1 & 6 & 46 \\
92 & 49 & 1 & 8 & 44 \\
92 & 50 & 2 & 12 & 40 \\
92 & 51 & 2 & 14 & 36 \\
92 & 52 & 2 & 16 & 34 \\
92 & 53 & 2 & 18 & 32 \\
92 & 54 & 2 & 20 & 31 \\
92 & 55 & 2 & 22 & 29 \\
92 & 56 & 2 & 24 & 26 \\
92 & 57 & 2 & 26 & 25 \\
92 & 58 & 2 & 28 & 24 \\
92 & 59 & 2 & 30 & 24 \\
92 & 60 & 2 & 32 & 22 \\
92 & 61 & 3 & 36 & 20 \\
92 & 62 & 3 & 38 & 20 \\
92 & 63 & 2 & 38 & 20 \\
92 & 64 & 2 & 40 & 20 \\
92 & 65 & 3 & 44 & 18 \\
92 & 66 & 2 & 44 & 18 \\
92 & 67 & 4 & 50 & 14 \\
92 & 68 & 3 & 50 & 14 \\
92 & 69 & 3 & 52 & 13 \\
92 & 70 & 4 & 56 & 12 \\
92 & 71 & 4 & 58 & 11 \\
92 & 72 & 5 & 62 & 10 \\
92 & 73 & 4 & 62 & 10 \\
92 & 74 & 4 & 64 & 10 \\
92 & 75 & 3 & 64 & 10 \\
92 & 76 & 6 & 72 & 7 \\
92 & 77 & 7 & 76 & 6 \\
92 & 78 & 6 & 76 & 6 \\
92 & 79 & 5 & 76 & 6 \\
92 & 80 & 5 & 78 & 5 \\
92 & 81 & 7 & 84 & 4 \\
92 & 82 & 6 & 84 & 4 \\
92 & 83 & 5 & 84 & 4 \\
\hline
93 & 46 & 1 & 1 & 93 \\
93 & 47 & 1 & 3 & 52 \\
93 & 48 & 1 & 5 & 48 \\
93 & 49 & 1 & 7 & 46 \\
93 & 50 & 2 & 11 & 40 \\
93 & 51 & 2 & 13 & 38 \\
93 & 52 & 2 & 15 & 36 \\
93 & 53 & 2 & 17 & 34 \\
93 & 54 & 2 & 19 & 32 \\
93 & 55 & 2 & 21 & 31 \\
93 & 56 & 2 & 23 & 29 \\
93 & 57 & 2 & 25 & 26 \\
93 & 58 & 2 & 27 & 24 \\
93 & 59 & 2 & 29 & 24 \\
93 & 60 & 2 & 31 & 24 \\
93 & 61 & 2 & 33 & 22 \\
93 & 62 & 3 & 37 & 20 \\
93 & 63 & 3 & 39 & 20 \\
93 & 64 & 2 & 39 & 20 \\
93 & 65 & 2 & 41 & 20 \\
93 & 66 & 3 & 45 & 18 \\
93 & 67 & 2 & 45 & 18 \\
93 & 68 & 4 & 51 & 14 \\
93 & 69 & 3 & 51 & 14 \\
93 & 70 & 3 & 53 & 13 \\
93 & 71 & 4 & 57 & 12 \\
93 & 72 & 4 & 59 & 11 \\
93 & 73 & 5 & 63 & 10 \\
93 & 74 & 4 & 63 & 10 \\
93 & 75 & 4 & 65 & 10 \\
93 & 76 & 3 & 65 & 10 \\
93 & 77 & 6 & 73 & 7 \\
93 & 78 & 7 & 77 & 6 \\
93 & 79 & 6 & 77 & 6 \\
93 & 80 & 5 & 77 & 6 \\
93 & 81 & 5 & 79 & 5 \\
93 & 82 & 7 & 85 & 4 \\
93 & 83 & 6 & 85 & 4 \\
93 & 84 & 5 & 85 & 4 \\
\hline
94 & 47 & 1 & 2 & 62 \\
94 & 48 & 1 & 4 & 49 \\
94 & 49 & 1 & 6 & 47 \\
94 & 50 & 2 & 10 & 42 \\
94 & 51 & 2 & 12 & 40 \\
94 & 52 & 2 & 14 & 37 \\
94 & 53 & 2 & 16 & 36 \\
94 & 54 & 2 & 18 & 33 \\
94 & 55 & 2 & 20 & 32 \\
94 & 56 & 2 & 22 & 31 \\
94 & 57 & 2 & 24 & 28 \\
94 & 58 & 2 & 26 & 26 \\
94 & 59 & 2 & 28 & 24 \\
94 & 60 & 2 & 30 & 24 \\
94 & 61 & 2 & 32 & 23 \\
94 & 62 & 3 & 36 & 20 \\
94 & 63 & 3 & 38 & 20 \\
94 & 64 & 3 & 40 & 20 \\
94 & 65 & 2 & 40 & 20 \\
94 & 66 & 2 & 42 & 20 \\
94 & 67 & 3 & 46 & 16 \\
94 & 68 & 3 & 48 & 15 \\
94 & 69 & 4 & 52 & 14 \\
94 & 70 & 3 & 52 & 14 \\
94 & 71 & 3 & 54 & 13 \\
94 & 72 & 4 & 58 & 12 \\
94 & 73 & 4 & 60 & 11 \\
94 & 74 & 5 & 64 & 10 \\
94 & 75 & 4 & 64 & 10 \\
94 & 76 & 4 & 66 & 10 \\
94 & 77 & 3 & 66 & 10 \\
94 & 78 & 6 & 74 & 7 \\
94 & 79 & 7 & 78 & 6 \\
94 & 80 & 6 & 78 & 6 \\
94 & 81 & 5 & 78 & 6 \\
94 & 82 & 5 & 80 & 5 \\
94 & 83 & 7 & 86 & 4 \\
94 & 84 & 6 & 86 & 4 \\
94 & 85 & 5 & 86 & 4 \\
\hline
95 & 47 & 1 & 1 & 95 \\
95 & 48 & 1 & 3 & 54 \\
95 & 49 & 1 & 5 & 48 \\
95 & 50 & 1 & 7 & 47 \\
95 & 51 & 2 & 11 & 41 \\
95 & 52 & 2 & 13 & 40 \\
95 & 53 & 2 & 15 & 37 \\
95 & 54 & 2 & 17 & 35 \\
95 & 55 & 2 & 19 & 32 \\
95 & 56 & 2 & 21 & 32 \\
95 & 57 & 2 & 23 & 31 \\
95 & 58 & 2 & 25 & 28 \\
95 & 59 & 2 & 27 & 26 \\
95 & 60 & 2 & 29 & 24 \\
95 & 61 & 2 & 31 & 24 \\
95 & 62 & 2 & 33 & 23 \\
95 & 63 & 3 & 37 & 20 \\
95 & 64 & 3 & 39 & 20 \\
95 & 65 & 3 & 41 & 20 \\
95 & 66 & 2 & 41 & 20 \\
95 & 67 & 2 & 43 & 20 \\
95 & 68 & 3 & 47 & 16 \\
95 & 69 & 3 & 49 & 15 \\
95 & 70 & 4 & 53 & 14 \\
95 & 71 & 3 & 53 & 14 \\
95 & 72 & 5 & 59 & 12 \\
95 & 73 & 4 & 59 & 12 \\
95 & 74 & 4 & 61 & 11 \\
95 & 75 & 5 & 65 & 10 \\
95 & 76 & 4 & 65 & 10 \\
95 & 77 & 4 & 67 & 10 \\
95 & 78 & 3 & 67 & 10 \\
95 & 79 & 6 & 75 & 7 \\
95 & 80 & 7 & 79 & 6 \\
95 & 81 & 6 & 79 & 6 \\
95 & 82 & 5 & 79 & 6 \\
95 & 83 & 5 & 81 & 5 \\
95 & 84 & 7 & 87 & 4 \\
95 & 85 & 6 & 87 & 4 \\
95 & 86 & 5 & 87 & 4 \\
\hline
96 & 48 & 1 & 2 & 64 \\
96 & 49 & 1 & 4 & 50 \\
96 & 50 & 1 & 6 & 48 \\
96 & 51 & 1 & 8 & 46 \\
96 & 52 & 2 & 12 & 40 \\
96 & 53 & 2 & 14 & 38 \\
96 & 54 & 2 & 16 & 36 \\
96 & 55 & 2 & 18 & 34 \\
96 & 56 & 2 & 20 & 32 \\
96 & 57 & 2 & 22 & 32 \\
96 & 58 & 2 & 24 & 28 \\
96 & 59 & 2 & 26 & 26 \\
96 & 60 & 2 & 28 & 25 \\
96 & 61 & 2 & 30 & 24 \\
96 & 62 & 2 & 32 & 24 \\
96 & 63 & 3 & 36 & 20 \\
96 & 64 & 3 & 38 & 20 \\
96 & 65 & 3 & 40 & 20 \\
96 & 66 & 3 & 42 & 20 \\
96 & 67 & 2 & 42 & 20 \\
96 & 68 & 2 & 44 & 20 \\
96 & 69 & 3 & 48 & 16 \\
96 & 70 & 4 & 52 & 14 \\
96 & 71 & 4 & 54 & 14 \\
96 & 72 & 3 & 54 & 14 \\
96 & 73 & 5 & 60 & 12 \\
96 & 74 & 4 & 60 & 12 \\
96 & 75 & 4 & 62 & 11 \\
96 & 76 & 5 & 66 & 10 \\
96 & 77 & 4 & 66 & 10 \\
96 & 78 & 4 & 68 & 9 \\
96 & 79 & 6 & 74 & 8 \\
96 & 80 & 8 & 80 & 6 \\
96 & 81 & 7 & 80 & 6 \\
96 & 82 & 6 & 80 & 6 \\
96 & 83 & 5 & 80 & 6 \\
96 & 84 & 5 & 82 & 5 \\
96 & 85 & 7 & 88 & 4 \\
96 & 86 & 6 & 88 & 4 \\
96 & 87 & 5 & 88 & 4 \\
\hline
97 & 48 & 1 & 1 & 97 \\
97 & 49 & 1 & 3 & 55 \\
97 & 50 & 1 & 5 & 48 \\
97 & 51 & 1 & 7 & 48 \\
97 & 52 & 2 & 11 & 42 \\
97 & 53 & 2 & 13 & 40 \\
97 & 54 & 2 & 15 & 38 \\
97 & 55 & 2 & 17 & 36 \\
97 & 56 & 2 & 19 & 33 \\
97 & 57 & 2 & 21 & 32 \\
97 & 58 & 2 & 23 & 32 \\
97 & 59 & 2 & 25 & 28 \\
97 & 60 & 2 & 27 & 26 \\
97 & 61 & 2 & 29 & 25 \\
97 & 62 & 2 & 31 & 24 \\
97 & 63 & 2 & 33 & 24 \\
97 & 64 & 3 & 37 & 20 \\
97 & 65 & 3 & 39 & 20 \\
97 & 66 & 3 & 41 & 20 \\
97 & 67 & 3 & 43 & 20 \\
97 & 68 & 2 & 43 & 20 \\
97 & 69 & 2 & 45 & 20 \\
97 & 70 & 3 & 49 & 16 \\
97 & 71 & 4 & 53 & 14 \\
97 & 72 & 4 & 55 & 14 \\
97 & 73 & 3 & 55 & 14 \\
97 & 74 & 5 & 61 & 12 \\
97 & 75 & 4 & 61 & 12 \\
97 & 76 & 4 & 63 & 11 \\
97 & 77 & 5 & 67 & 10 \\
97 & 78 & 4 & 67 & 10 \\
97 & 79 & 4 & 69 & 9 \\
97 & 80 & 6 & 75 & 8 \\
97 & 81 & 8 & 81 & 6 \\
97 & 82 & 7 & 81 & 6 \\
97 & 83 & 6 & 81 & 6 \\
97 & 84 & 5 & 81 & 6 \\
97 & 85 & 5 & 83 & 5 \\
97 & 86 & 7 & 89 & 4 \\
97 & 87 & 6 & 89 & 4 \\
97 & 88 & 5 & 89 & 4 \\
\hline
98 & 49 & 1 & 2 & 65 \\
98 & 50 & 1 & 4 & 52 \\
98 & 51 & 1 & 6 & 48 \\
98 & 52 & 1 & 8 & 47 \\
98 & 53 & 2 & 12 & 41 \\
98 & 54 & 2 & 14 & 40 \\
98 & 55 & 2 & 16 & 36 \\
98 & 56 & 2 & 18 & 35 \\
98 & 57 & 2 & 20 & 33 \\
98 & 58 & 2 & 22 & 32 \\
98 & 59 & 2 & 24 & 30 \\
98 & 60 & 2 & 26 & 28 \\
98 & 61 & 2 & 28 & 26 \\
98 & 62 & 2 & 30 & 25 \\
98 & 63 & 2 & 32 & 24 \\
98 & 64 & 2 & 34 & 24 \\
98 & 65 & 3 & 38 & 20 \\
98 & 66 & 3 & 40 & 20 \\
98 & 67 & 3 & 42 & 20 \\
98 & 68 & 3 & 44 & 20 \\
98 & 69 & 2 & 44 & 20 \\
98 & 70 & 2 & 46 & 20 \\
98 & 71 & 3 & 50 & 16 \\
98 & 72 & 4 & 54 & 14 \\
98 & 73 & 4 & 56 & 14 \\
98 & 74 & 3 & 56 & 14 \\
98 & 75 & 5 & 62 & 12 \\
98 & 76 & 4 & 62 & 12 \\
98 & 77 & 4 & 64 & 11 \\
98 & 78 & 5 & 68 & 10 \\
98 & 79 & 4 & 68 & 10 \\
98 & 80 & 4 & 70 & 9 \\
98 & 81 & 6 & 76 & 8 \\
98 & 82 & 8 & 82 & 6 \\
98 & 83 & 7 & 82 & 6 \\
98 & 84 & 6 & 82 & 6 \\
98 & 85 & 5 & 82 & 6 \\
98 & 86 & 5 & 84 & 5 \\
98 & 87 & 7 & 90 & 4 \\
98 & 88 & 6 & 90 & 4 \\
98 & 89 & 5 & 90 & 4 \\
\hline
99 & 49 & 1 & 1 & 99 \\
99 & 50 & 1 & 3 & 56 \\
99 & 51 & 1 & 5 & 50 \\
99 & 52 & 1 & 7 & 48 \\
99 & 53 & 2 & 11 & 44 \\
99 & 54 & 2 & 13 & 40 \\
99 & 55 & 2 & 15 & 40 \\
99 & 56 & 2 & 17 & 36 \\
99 & 57 & 2 & 19 & 34 \\
99 & 58 & 2 & 21 & 32 \\
99 & 59 & 2 & 23 & 32 \\
99 & 60 & 2 & 25 & 30 \\
99 & 61 & 2 & 27 & 27 \\
99 & 62 & 2 & 29 & 26 \\
99 & 63 & 2 & 31 & 25 \\
99 & 64 & 2 & 33 & 24 \\
99 & 65 & 2 & 35 & 24 \\
99 & 66 & 3 & 39 & 20 \\
99 & 67 & 3 & 41 & 20 \\
99 & 68 & 3 & 43 & 20 \\
99 & 69 & 3 & 45 & 20 \\
99 & 70 & 2 & 45 & 20 \\
99 & 71 & 2 & 47 & 20 \\
99 & 72 & 3 & 51 & 16 \\
99 & 73 & 4 & 55 & 14 \\
99 & 74 & 4 & 57 & 14 \\
99 & 75 & 3 & 57 & 14 \\
99 & 76 & 5 & 63 & 12 \\
99 & 77 & 4 & 63 & 12 \\
99 & 78 & 4 & 65 & 11 \\
99 & 79 & 5 & 69 & 10 \\
99 & 80 & 4 & 69 & 10 \\
99 & 81 & 4 & 71 & 9 \\
99 & 82 & 6 & 77 & 8 \\
99 & 83 & 8 & 83 & 6 \\
99 & 84 & 7 & 83 & 6 \\
99 & 85 & 6 & 83 & 6 \\
99 & 86 & 5 & 83 & 6 \\
99 & 87 & 5 & 85 & 5 \\
99 & 88 & 7 & 91 & 4 \\
99 & 89 & 6 & 91 & 4 \\
99 & 90 & 5 & 91 & 4 \\
\hline
100 & 50 & 1 & 2 & 66 \\
100 & 51 & 1 & 4 & 52 \\
100 & 52 & 1 & 6 & 49 \\
100 & 53 & 1 & 8 & 48 \\
100 & 54 & 2 & 12 & 42 \\
100 & 55 & 2 & 14 & 40 \\
100 & 56 & 2 & 16 & 38 \\
100 & 57 & 2 & 18 & 36 \\
100 & 58 & 2 & 20 & 34 \\
100 & 59 & 2 & 22 & 32 \\
100 & 60 & 2 & 24 & 32 \\
100 & 61 & 2 & 26 & 30 \\
100 & 62 & 2 & 28 & 27 \\
100 & 63 & 2 & 30 & 26 \\
100 & 64 & 2 & 32 & 25 \\
100 & 65 & 2 & 34 & 24 \\
100 & 66 & 2 & 36 & 23 \\
100 & 67 & 3 & 40 & 20 \\
100 & 68 & 3 & 42 & 20 \\
100 & 69 & 3 & 44 & 20 \\
100 & 70 & 3 & 46 & 20 \\
100 & 71 & 2 & 46 & 20 \\
100 & 72 & 2 & 48 & 20 \\
100 & 73 & 3 & 52 & 16 \\
100 & 74 & 4 & 56 & 14 \\
100 & 75 & 4 & 58 & 14 \\
100 & 76 & 3 & 58 & 14 \\
100 & 77 & 5 & 64 & 12 \\
100 & 78 & 4 & 64 & 12 \\
100 & 79 & 4 & 66 & 11 \\
100 & 80 & 5 & 70 & 10 \\
100 & 81 & 4 & 70 & 10 \\
100 & 82 & 4 & 72 & 9 \\
100 & 83 & 6 & 78 & 8 \\
100 & 84 & 8 & 84 & 6 \\
100 & 85 & 7 & 84 & 6 \\
100 & 86 & 6 & 84 & 6 \\
100 & 87 & 5 & 84 & 6 \\
100 & 88 & 5 & 86 & 5 \\
100 & 89 & 7 & 92 & 4 \\
100 & 90 & 6 & 92 & 4 \\
100 & 91 & 5 & 92 & 4 \\
\hline
\hline

\end{supertabular}
\end{multicols*}



\bibliographystyle{IEEEtran}

\end{document}